\documentclass[prd,aps,amsfonts,eqsecnum,superscriptaddress,nofootinbib,longbibliography,notitlepage]{revtex4-1}

\usepackage{graphicx}
\usepackage{xcolor}
\usepackage{subcaption}
\usepackage{rotating}
\usepackage{amsmath,amssymb,graphics,amsthm,isomath}
\usepackage{physics }
\usepackage{verbatim}
\usepackage{circuitikz}
\usepackage{tikz}
\usepackage{esint}
\usepackage{blkarray}

\usepackage[colorlinks=true, urlcolor=violet, linkcolor=blue, citecolor=red, hyperindex=true, linktocpage=true]{hyperref}
\usepackage[capitalise,compress]{cleveref}

\newtheorem{thm}{Theorem}
\numberwithin{thm}{section}
\newtheorem{cor}[thm]{Corollary}

\newtheorem{prop}[thm]{Proposition}
\newtheorem{assumpn}[thm]{Assumption}
\newtheorem{algo}[thm]{Algorithm}

\newtheorem{obs}[thm]{Observation}
\newtheorem{defn}[thm]{Definition}

\makeatletter
\renewcommand{\p@subsection}{}
\renewcommand{\p@subsubsection}{}
\makeatother

\usepackage{mathtools}
\tikzstyle{densely dashed}=          [dash pattern=on 4pt off 3pt]

\usepackage{dsfont}

\makeatletter
\pgfcircdeclarebipolescaled{instruments}
{
    \anchor{text}{\pgfextracty{\pgf@circ@res@up}{\northeast}
        \pgfpoint{-.5\wd\pgfnodeparttextbox}{
            \dimexpr.5\dp\pgfnodeparttextbox+.5\ht\pgfnodeparttextbox+\pgf@circ@res@up\relax
        }
    }
}
{\ctikzvalof{bipoles/oscope/height}}
{josephson}
{\ctikzvalof{bipoles/oscope/height}}
{\ctikzvalof{bipoles/oscope/width}}
{
    \pgf@circ@setlinewidth{bipoles}{\pgfstartlinewidth}
    \pgfextracty{\pgf@circ@res@up}{\northeast}
    \pgfextractx{\pgf@circ@res@right}{\northeast}
    \pgfextractx{\pgf@circ@res@left}{\southwest}
    \pgfextracty{\pgf@circ@res@down}{\southwest}
    \pgfmathsetlength{\pgf@circ@res@step}{0.25*\pgf@circ@res@up}
    \pgfscope
        \pgfpathrectanglecorners{\pgfpoint{\pgf@circ@res@left}{\pgf@circ@res@down}}{\pgfpoint{\pgf@circ@res@right}{\pgf@circ@res@up}}
        \pgf@circ@draworfill
    \endpgfscope
    \pgfscope
      \pgfpathmoveto{\pgfpoint{\pgf@circ@res@left}{\pgf@circ@res@up}}%
      \pgfpathlineto{\pgfpoint{\pgf@circ@res@right}{\pgf@circ@res@down}}%
      \pgfpathmoveto{\pgfpoint{\pgf@circ@res@right}{\pgf@circ@res@up}}%
      \pgfpathlineto{\pgfpoint{\pgf@circ@res@left}{\pgf@circ@res@down}}%
      \pgfusepath{draw}
    \endpgfscope
}
\def\pgf@circ@josephson@path#1{\pgf@circ@bipole@path{josephson}{#1}}
\tikzset{josephson/.style = {\circuitikzbasekey, /tikz/to path=\pgf@circ@josephson@path, l=#1}}

\definecolor{capc}{HTML}{5ec962}
\definecolor{indc}{HTML}{21918c}
\definecolor{pc}{HTML}{3b528b}
\definecolor{sc}{HTML}{440154}
\definecolor{gg}{HTML}{555555}

\allowdisplaybreaks

\begin{document}

\title{Stochastic theory of nonlinear electrical circuits in thermal equilibrium}
\author{Andrew Osborne}
\email{andrew.osborne-1@colorado.edu}
\affiliation{Department of Physics and Center for Theory of Quantum Matter, University of Colorado, Boulder CO 80309, USA}
\author{Andrew Lucas}
\email{andrew.j.lucas@colorado.edu}
\affiliation{Department of Physics and Center for Theory of Quantum Matter, University of Colorado, Boulder CO 80309, USA}

\begin{abstract}
We revisit the theory of dissipative mechanics in RLC circuits, allowing for  circuit elements to have nonlinear constitutive relations, and for the circuit to have arbitrary topology.  We systematically generalize the dissipationless Hamiltonian mechanics of an LC circuit to account for resistors and incorporate the physical postulate that the resulting RLC circuit thermalizes with its environment at a constant positive temperature. Our theory explains stochastic fluctuations, or Johnson noise, which are mandated by the fluctuation-dissipation theorem.  Assuming Gaussian Markovian noise, we obtain exact expressions for multiplicative Johnson noise through nonlinear resistors in circuits with convenient (parasitic) capacitors and/or inductors.  With linear resistors, our formalism is describable using a Kubo-Martin-Schwinger-invariant Lagrangian formalism for dissipative thermal systems.  Generalizing our technique to quantum circuits could lead to an alternative way to study decoherence in nonlinear superconducting circuits without the Caldeira-Leggett formalism.
\end{abstract}

\maketitle

\tableofcontents

\section{Introduction}
Qubits made out of superconducting circuits \cite{koch2007,manucharyan_fluxonium_2009,gyenis_2021,krantz_quantum_2019,kjaergaard_superconducting_2020,blais_circuit_2021} represent a promising candidate qubit for building a scalable fault-tolerant quantum computer. Some schemes even exist to use dissipation to assist in error correction \cite{nathan2024selfcorrecting,Miexp}.  As with all quantum technologies, one of the important challenges in building robust qubits is protecting them against decoherence with a thermal environment. 

A common way that dissipation with environmental degrees of freedom is modeled in the literature is via the Caldeira-Leggett model \cite{CALDEIRA1983587}, in which a dissipative bath is modeled by a large number of dissipationless degrees of freedom.  This approach is frequently used in the literature on circuit quantization \cite{Sorin_2021,minev2021energyparticipation,vool_introduction_2017,Abdo_2013,blais_circuit_2021,spain22}. Yet in a quantum setting, it is rather undesirable to include a very large number of degrees of freedom (which causes enormous computational overhead) simply to account for the existence of dissipation.  In principle, it would be more desirable to simply model the dissipative dynamics with a Lindblad master equation, describing the open quantum dynamics of only the relevant circuit degrees of freedom.

Before such a task can be achieved, it is desirable to first develop a thorough understanding of the \emph{classical} dynamics of electrical circuits interacting with a thermal bath, within a framework that is amenable to quantization.  For dissipationless LC circuits, this has recently been achieved in an intriniscally Hamiltonian formulation of circuit mechanics \cite{osborne2023symplectic,Parra-Rodriguez:2023ykw}, building on earlier work \cite{brayton,bahar_generalized,weissmathis}.  This Hamiltonian formulation contrasts with the more conventional Lagrangian approach used in much of the superconducting circuit literature \cite{koch2007,jens_timedep,Ciani:2023ubt,vool_introduction_2017,nigg_black-box_2012,ulrich_dual_2016,riwar_circuit_2022,blais_circuit_2021,kjaergaard_superconducting_2020,thanh_le_building_2020,devoret_fqi,brayvi_the_future,krantz_quantum_2019}.  For dissipative circuits with linear resistors, \cite{Parra-Rodriguez:2023ykw,Parra-Rodriguez:2024vtx,mariantoni2024quantum} have also  incorporated the resistors within a classical Hamiltonian formalism via Rayleigh dissipative functions.

This paper provides an alternative route towards \emph{systematically} incorporating dissipation into the mechanical theory of electrical circuits.  Following the recent literature on effective theories for thermal systems \cite{haehl2016fluid,eft1,eft2,jensen2018dissipative,Liulec}, we write down a dissipative Lagrangian for the dynamics of an RLC circuit with Gaussian white noise.  The approach is closely related to the much older Martin-Siggia-Rose (MSR) Lagrangian \cite{MSR} for stochastic dynamics.  Indeed, a \emph{crucial} aspect of our construction is that it necessarily incorporates both dissipation along with stochastic fluctuations, which are mandated by the fluctuation-dissipation theorem.  Indeed, it is not consistent with statistical mechanics to neglect the presence of noise when modeling a dissipative system.  Hence, we must develop a description of circuit dynamics for circuits in thermal equilibrium at temperature $T$, ensuring our description is compatible with the laws of statistical physics.  We can accomplish this without needing to worry about the microscopic details of the environment and its coupling to the relevant circuit degrees of freedom.  In our view, this makes our starting point more physically natural than a Rayleigh dissipative function, although we will ultimately see how to reproduce (and generalize) the physics of Rayleigh dissipative functions within our framework.

As our formalism necessarily incorporates both dissipative and stochastic physics, it allows us to reproduce the physics of Johnson noise \cite{JohnJohnson,Nyquist,Twiss}, which was an important historical benchmark for testing the fluctuation-dissipation theorem itself (or, more precisely, testing that physical systems are in thermal equilibrium).  When resistors are linear, we will explain how the MSR theory reproduces the textbook predictions for Johnson noise in complete generality.  However, at least in principle, it is also transparent in our framework how to model \emph{nonlinear resistors}, which imply \emph{multiplicative Johnson noise}.   There have been inconsistent predictions in the literature \cite{martinis,vonoppen} on the correct form of multiplicative Johnson noise; our formalism gives unambiguous predictions which can be derived transparently.  In particular, for circuits where resistors have suitable ``parasitic" capacitors or inductors alongside them, we are able to give simple analytic expressions for multiplicative Johnson noise, regardless of the global circuit topology.  In the appropriate limiting case, our results agree with \cite{vonoppen}.  

This paper focuses on classical circuits, as the problem is already somewhat involved.  Nevertheless, we do expect that our methods could be a starting point  for building intrinsically dissipative descriptions of quantum RLC circuits, in future work.

\section{Review of stochastic and dissipative mechanics}\label{sec:background}
 In a physical setting, the fluctuation-dissipation theorem requires both dissipation and stochastic noise, so it is sensible to demand that any systematic formalism for modeling dissipative circuit dynamics can account for both.  In this paper, we will focus on modeling circuits that are in thermal equilibrium at finite temperature $T$, such that the probability of finding the circuit in configuration $\xi$ in equilibrium is \begin{equation}
    P(\xi) = Z^{-1} \mathrm{e}^{-\beta E(\xi)}. \label{eq:thermaleq}
\end{equation}
Here $E$ is the combined energy of the capacitive and inductive elements, while $\beta$ is inverse temperature.  $Z$ is a normalization constant which, as we will see, is not important.

%The first is to rely on the technology of stochastic differential equations and the second is to construct an partial differential equation satisfied by some probability distribution. 
%Both of these approaches will be of use to  us going forward. We begin with a very brief discussion of stochastic differential equations.

Our interest is in first--order stochastic differential equations (a.k.a. Langevin equations) with Gaussian white noise:
\begin{equation}\label{eqn:sde}
    \dot x_\alpha = f_\alpha(x) + M_{\alpha \mu}\xi_\mu(t)
\end{equation}
where $\xi_\mu(t)$ are independent and identically distributed Gaussian random white noise: \begin{equation}
    \langle \xi_\mu(t) \xi_\nu(s)\rangle = 2\delta_{\mu\nu} \delta(t-s).
\end{equation}
We use $\mu\nu$ indices to denote the noise variables vs. $\alpha\beta$ indices to represent the physical degrees of freedom, because we may have a smaller number of noise variables than coordinates $x_\alpha$.   Here and below we invoke the Einstein summation convention on repeated indices, unless otherwise stated.
If $M_{\alpha\mu}$ is a constant matrix independent of $x$, then this Langevin equation is well-defined; otherwise, more care is required \cite{Huang:2023eyz} to regulate the stochastic equation.  We will explain later how to correctly regularize the problem for systems in thermal equilibrium.  Because we will find it useful to distinguish between constant and non-constant noise, we introduce the following terminology: \begin{defn}[Multiplicative noise]
    Stochastic equation \eqref{eqn:sde} has \textbf{multiplicative noise} if $M_{\alpha\mu}(x)$ is not constant. 
\end{defn}

An important question is what constraints should be placed on \eqref{eqn:sde} to ensure that thermal equilibrium \eqref{eq:thermaleq} describes a statistical steady state to which the stochastic dynamics relaxes.  We now describe a few related ways to achieve this goal.

\subsection{The Martin-Siggia-Rose formalism}
The first approach we will use to study \eqref{eqn:sde} is based on generalizing Lagrangian mechanics to dissipative and stochastic systems.  This approach is inspired by the desire to calculate the transition probability $P(x_\alpha(t)=a_\alpha|x_\alpha(0)=b_\alpha)$ -- namely, the probability to go from microstate $b$ to microstate $a$ in time $t$. One popular way of trying to evaluate this transition probability (density function) is by using the Martin-Siggia-Rose (MSR) path integral: \cite{MSR}  \begin{equation} \label{eq:MSR1st}
    P(x_\alpha(t)=a_\alpha|x_\alpha(0)=b_\alpha) = \int\limits_{x(t)=a, x(0)=b} \mathcal{D}x \mathcal{D}\xi \; \delta(\dot x_\alpha - f_\alpha(x)-M_{\alpha\mu}\xi_\mu) \exp\left[-\frac{1}{4}\int \mathrm{d}t \; \xi_\mu \xi_\mu \right].
\end{equation} 
The path integral over $\xi$ corresponds to averaging over the random noise, and is subject to appropriate boundary conditions.  
If we interpret (\ref{eqn:sde}) in the Ito formulation (see e.g. \cite{Huang:2023eyz} for details), then we can perform standard path integral manipulations to write \begin{equation}
      P(x_\alpha(t)=a_\alpha|x_\alpha(0)=b_\alpha) = \int\limits_{x(t)=a, x(0)=b} \mathcal{D}x \mathcal{D}\xi \mathcal{D}\pi \; \exp\left[\mathrm{i}\int \mathrm{d}t \left( \pi_\alpha (\dot x_\alpha - f_\alpha(x)-M_{\alpha \mu}\xi_\mu)+ \frac{\mathrm{i}}{4} \xi_\mu \xi_\mu  \right)\right],
\end{equation}
where $\pi$ is a Lagrange multiplier enforcing the equation of motion.  In the presence of Gaussian noise, we can then integrate out Gaussian variable $\xi_\mu$, and we obtain \begin{equation}
    P(x_\alpha(t)=a_\alpha|x_\alpha(0)=b_\alpha) = \int\limits_{x(t)=a, x(0)=b} \mathcal{D}x \mathcal{D}\pi \; \exp\left[\mathrm{i}\int \mathrm{d}t \left( \pi_\alpha (\dot x_\alpha - f_\alpha(x))+ \mathrm{i}M_{\alpha \mu}M_{\beta \mu} \pi_\alpha \pi_\beta \right)\right]. \label{eq:26}
\end{equation}
This transition probability looks exactly like Feynman's path integral, weighted by $\exp[\mathrm{i}\int \mathrm{d}t \; L]$.  The resulting Lagrangian is given a name: \begin{defn}[MSR Lagrangian]
    An \textbf{MSR Lagrangian} $L(\pi_\alpha, x_\alpha)$ is a complex valued function of the form \begin{equation}
        L = \pi_\alpha \dot x_\alpha - \mathcal{H}(\pi_\alpha, x_\alpha) \label{eq:MSR}
    \end{equation}
    such that $\mathcal{H}(0,x_\alpha)=0$ and $\mathrm{Im}(\mathcal{H}) \le 0$.
\end{defn}
In this paper, as in \eqref{eq:26}, $\mathcal{H}$ will be at most quadratic in $\pi_\alpha$, because the noise in the stochastic process is Gaussian.  While the MSR Lagrangian is what appears in the admittedly nonrigorous path integral, in many useful cases, we will see that it can be a ``mnemonic" for a Fokker-Planck equation that is unambiguously defined.  We also note that the requirement  $\mathcal{H}(0,x_\alpha)=0$ is motivated by our path integral derivation of the MSR Lagrangian, while the condition $\mathrm{Im}(\mathcal{H}) \le 0$ will ensure the physical requirement that the noise variance is positive.

The utility of the MSR Lagrangian comes from its ability to very beautifully incorporate the assumption that the system approaches a known steady-state \cite{Huang:2023eyz}, which manifests as a simple symmetry of the MSR Lagrangian.  In this paper, we are interested in systems where this steady state is thermal equilibrium (\ref{eq:thermaleq}) \cite{haehl2016fluid,eft1,eft2,jensen2018dissipative,Liulec} where, in many relevant settings, the MSR Lagrangian can be shown to be time-reversal symmetric:\footnote{While one could envision an even more general definition of time-reversal symmetry, the one below will suffice for the present paper.} \begin{defn}[Time-reversal symmetry] Given a desired thermal steady state \eqref{eq:thermaleq}, 
    MSR Lagrangian $L$ is time-reversal symmetric if $\mathcal{H}$, defined in \eqref{eq:MSR}, obeys \begin{equation}
        \mathcal{H}(\pi_\alpha(t) , x_\alpha(t)) = \mathcal{H}(-\eta_\alpha (\pi_\alpha(-t) - \mathrm{i}\beta \partial_\alpha E), \eta_\alpha x_\alpha(-t)). \label{eq:time-reversal}
    \end{equation}for some $\eta_\alpha \in \lbrace \pm 1\rbrace$ (there is no sum on repeated $\alpha$ index here). 
    \label{def:TRS1}
\end{defn}

At the (not rigorous) level of path integrals, it is simple \cite{Huang:2023eyz} to show that time-reversal symmetry implies the microscopic detailed balance condition \begin{equation}
    P(x_\alpha(t)=a_\alpha | x_\alpha(0) = b_\alpha) \mathrm{e}^{-\beta E(b_\alpha)} = P(x_\alpha(t)=\eta_\alpha b_\alpha | x_\alpha(0) = \eta_\alpha a_\alpha) \mathrm{e}^{-\beta E(\eta_\alpha a_\alpha)},
\end{equation} 
as long as $E(\eta a) = E(a)$.

In general, \eqref{eq:26} will certainly not be invariant under (\ref{eq:time-reversal}).  Indeed, so long as the $x_\alpha$ are all defined on the real line (so that the phase space of the dynamics is topologically trivial), the most general theory we can consider, whose MSR Lagrangian is quadratic in $\pi$ and is invariant under \eqref{eq:time-reversal}, is \cite{Huang:2023eyz} \begin{equation} \label{eq:MSRquadratic}
    L = \pi_\alpha \dot x_\alpha + \mathrm{i}K_{\beta \alpha}(x) \pi_\beta (\pi_\alpha - \mathrm{i}\beta \partial_\alpha E).
\end{equation}
This is time-reversal symmetric so long as (no sum on repeated indices): \begin{equation}
    K_{\beta \alpha} = K_{\alpha \beta } \eta_\alpha \eta_\beta .
\end{equation}

The MSR Lagrangians above are very useful for both direct comparison to a Langevin equation, as well as to the Fokker-Planck equation, as we will discuss shortly.  However, we will also find it useful when building MSR Lagrangians for circuits, to put them in a slightly different form, which we will call ``dissipative Lagrangians" to remind the reader of a few technical differences spelled out below.  Suppose for simplicity that the matrix $K$ is invertible: \begin{equation}
    K_{\alpha\beta} Q_{\beta\gamma} = T \delta_{\alpha\gamma}, 
\end{equation} 
where \begin{equation}
    T = \frac{1}{\beta}
\end{equation}
denotes temperature.  By making the change of variable \begin{equation}
    \pi_\alpha = Q_{\alpha\beta}\Pi_\beta,
\end{equation}
we can write \eqref{eq:MSRquadratic} as \begin{equation} \label{eq:MSRquadratic2}
    L =  \mathrm{i}T Q_{\beta\alpha} \Pi_\alpha ( \Pi_\beta - \mathrm{i}\beta  \dot x_\beta) + \Pi_\alpha \partial_\alpha E.
\end{equation}
With these new variables, we can modify our previous definitions as follows:
\begin{defn}[Dissipative Lagrangian/time-reversal] A dissipative Lagrangian $L(\Pi_\alpha,x_\alpha)$ obeys $L(0,x_\alpha)=0$, $\mathrm{Im}(L) \ge 0$, and (no sum on $\alpha$):
\begin{equation} \label{eq:TRS2}
    L(-\eta_\alpha (\Pi_\alpha(-t) - \mathrm{i}\beta \dot x_\alpha(-t)), \eta_\alpha x_\alpha(-t)) = \tilde L(\Pi_\alpha,x_\alpha) + \mathrm{i}\beta \frac{\mathrm{d}E}{\mathrm{d}t}.
\end{equation}
$\tilde L$ must also obey $\tilde L(0,x_\alpha)=0$ and $\mathrm{Im}(\tilde L) \ge 0$.  If $\tilde L = L$, the Lagrangian is time-reversal symmetric.
    \label{def:TRS2}
\end{defn}
If (\ref{eq:MSRquadratic}) is time-reversal symmetric by Definition \ref{def:TRS1}, then \eqref{eq:MSRquadratic2} is time-reversal symmetric by Definition \ref{def:TRS2}. Also, Definition \ref{def:TRS2} can be interpreted as Kubo-Martin-Schwinger symmetry in the Schwinger-Keldysh path integral for dissipative systems \cite{Liulec}.

\subsection{Fokker-Planck equation}
We will also find it useful to interpret the stochastic equation (\ref{eqn:sde}) in an alternative way, wherein we solve a partial differential equation called the Fokker-Planck equation (FPE) for the probability density $P(x,t)$ of finding the system at microstate $x$ at time $t$.   In the setting where $M_{\alpha \mu}$ is constant in \eqref{eqn:sde}, the derivation of such a Fokker-Planck equation is unambiguous and can be found in textbooks \cite{gardiner}: \begin{prop}[Fokker-Planck equation]
    Given stochastic equation \eqref{eqn:sde} with constant $M_{\alpha\mu}$, the probability density $P(x,t)$ obeys the Fokker-Planck equation:
    \begin{equation}\label{eq:propFPE}
        \partial_t P = -\partial_\alpha (f_\alpha P) + \frac{1}{2}M_{\alpha \mu} M_{\beta \mu} \partial_\alpha \partial_\beta P.
    \end{equation}
    $P(x(t)=a|x(0)=b)$ is the Green's function to the Fokker-Planck equation \eqref{eq:propFPE}.
\end{prop}

Now, following \cite{Huang:2023eyz}, we will define the following procedure for converting an MSR Lagrangian into a FPE, when noise is multiplicative: \begin{defn}[Fokker-Planck equation for MSR Lagrangian with multiplicative noise]\label{def:FPEgood}
Consider a system with thermal steady state \eqref{eq:thermaleq} and MSR Lagrangian \eqref{eq:MSRquadratic}.  If $K_{\beta\alpha}$ is positive definite and invertible, then the Fokker-Planck equation defining the stochastic process is \begin{equation} \label{eq:FPEdef}
    \partial_t P = \partial_\beta \left(K_{\beta \alpha}(x) \left(\partial_\alpha P + \beta P \partial_\alpha E\right)\right).
\end{equation}
A solution to this differential equation is the thermal steady state: $P(x,t) = \exp[-\beta E(x)]$.
\end{defn}

The reason we define the FPE first via \eqref{eq:FPEdef} is that we are guaranteed to have a thermal steady state. That is not trivial to achieve for general multiplicative noise, starting from \eqref{eqn:sde}.  In particular, using the Ito regularization of stochastic equations, which is most commonly employed, one wishes to write \begin{equation}\label{eq:itoFPE}
    \partial_t P = \partial_\alpha \partial_\beta \left(\frac{1}{2}D_{\alpha\beta} P\right) - \partial_\alpha \left(f_\alpha P\right).
\end{equation}
Comparing \eqref{eq:FPEdef} and \eqref{eq:itoFPE}, we notice that \begin{subequations}
    \begin{align}
        D_{\alpha\beta} &= K_{\alpha\beta}+K_{\beta\alpha}, \\
        f_\alpha &= -\beta K_{\alpha\beta}\partial_\beta E + \partial_\beta \frac{K_{\alpha\beta}+K_{\beta\alpha}}{2}. 
    \end{align}
\end{subequations}
The MSR path integral is ``formally" defined assuming regularization \eqref{eq:itoFPE}, whereas we want \eqref{eq:FPEdef}.  In practice, we will always use Definition \ref{def:FPEgood} in the presence of multiplicative noise before carrying out any physical calculation, although we will see that the language of dissipative and MSR Lagrangians makes obtaining this eventual FPE much more transparent for electrical circuit mechanics.

\section{Theory of circuit dynamics}\label{sec:circuits}
Having reviewed a general formalism for dissipative and stochastic dynamics, we now apply these methods to the analysis of classical RLC circuits.  In this section, we focus on the simplifying problem of linear resistors, which we will see corresponds to constant (not multiplicative) noise.  In this setting we can perform exact manipulations with dissipative Lagrangians, which cleanly connect to standard expectations about Johnson noise in the literature.   The more subtle problem of nonlinear resistors, which implies multiplicative noise, is the subject of Section \ref{sec:nonlinear}.
\subsection{Overview of circuits}
We begin by reviewing a mathematical formalism for describing circuits \cite{fluxcharge}.  Firstly, given any set $X$,  let
\begin{equation}
    \mathcal D(X) = \text{span}\left(|x\rangle \,:\, x \in X\right).
\end{equation}
With this notation in mind, it is natural to view a circuit as a directed graph.   
\begin{defn}[Graphs]
    Let $\mathcal{V}$ denote a set of vertices associated to some graph $G$.  The edge set $\mathcal{E} \subset \mathcal{V}\times \mathcal{V}$ is an ordered pair of vertices corresponding to the start and end point of the directed edge.  
    \end{defn}
 
    \begin{defn}[Loop] \label{def:loops} A \textbf{loop} $l$ of length $n$ is a special subset of edges of the form \begin{equation}
    l = \lbrace (v_1,v_2), (v_2,v_3),\ldots, (v_n,v_1)\rbrace \subseteq \mathcal{E}.
\end{equation}
It is also acceptable for some edges in the loop to be backwards, e.g. $l=\lbrace (v_1,v_2), (v_3,v_2), \ldots \rbrace$.
We associate a loop set $\mathcal{L}$ to graph $G$, such that for arbitrary loop $l$, there exist $p\ge 1$ loops $m_1,\ldots, m_p \in \mathcal{L}$
 such that $l\subseteq m_1 \cup \cdots \cup m_p$. 
 \end{defn}\begin{defn}[Cut]\label{def:cut}
  $c\subseteq \mathcal{E}$ is a \textbf{cut} of graph $G$ if and only if there exists a partition of the vertex set $\mathcal{V}  =\mathcal{V}_1\cup \mathcal{V}_2$ (with $\mathcal{V}_1 \cap \mathcal{V}_2 = \emptyset$) such that \begin{equation}
      c = \lbrace (u,v)\in \mathcal{E} : \lbrace u, v\rbrace \not\subset \mathcal{V}_1 \text{ and }\lbrace u, v\rbrace \not\subset \mathcal{V}_2 \rbrace.
  \end{equation}
  In other words, a cut corresponds to the set of edges that cross between the partition $\mathcal{V}_1$ and $\mathcal{V}_2$.
 \end{defn}

\begin{defn}[Boundary maps]\label{def:boundary}
  The boundary maps of a graph are the following matrices $A:\mathcal D(\mathcal V) \rightarrow \mathcal D (\mathcal E)$ and $B:\mathcal D(\mathcal E)\rightarrow \mathcal D(\mathcal L)$:
\begin{subequations}
    \begin{align}
    A_{ev} = \langle e | A|v\rangle &= 
    \begin{cases}
        1 & \exists \, u \, \text{such that} \, e = (u,v) \\
        -1 & \exists \, u \, \text{such that} \, e = (v,u) \\ 
        0 & \text{otherwise}
    \end{cases} \\ 
    B_{le} = \langle l | B |e\rangle &= 
    \begin{cases}
        1 & e \text{ and } l\text{ are oriented alike}\\
        -1 & e \text{ and } l\text{ are oriented unalike} \\ 
        0 & e \text{ and } l \text{ do not intersect}.
    \end{cases}
    \end{align}
\end{subequations}
\end{defn}

Notice that the definition of the loop set $\mathcal{L}$ in Definition \ref{def:loops} is chosen to ensure that the matrix $B_{le}$ defined above provides a map from arbitrary loops to edges.  Definition \ref{def:boundary}, among other upcoming definitions, is illustrated in a simple example in Figure \ref{fig:example}. We state without proof the following well known fact from graph theory, which amounts to the statement that a loop has no boundary:
\begin{prop}
    $\mathcal D(\mathcal V) \xrightarrow{A} \mathcal D(\mathcal E) \xrightarrow{B}\mathcal D(\mathcal L)$ is a short exact sequence over vector spaces, and
\begin{equation}
    \mathrm{Ker}(B) = \mathrm{Im}(A).
\end{equation}
\end{prop}

\begin{figure}
    \centering
    \begin{subfigure}{0.49\linewidth}
    \begin{circuitikz}[scale=2.5]
       \draw[thick] (0,0) to[short,i_=$e_1$] (2,0) to[short,i_=$e_2$] (1,2) to[short,i_=$e_3$] (0,0) to[short,i_=$e_4$] (1,.7) to[short,i_=$e_5$](1,2);
       \draw[thick] (2,0) to[short,i^=$e_6$] (1,.7);
       \draw[thick] (2,0) to[short,i_=$e_7$] (3,0) to[short,i_=$e_8$] (3,1) to[short,i_=$e_9$] (3,2) to[short,i_=$e_{10}$] (1,2) to[short,i_=$e_{11}$] (3,1);
       \draw[thick] (2,0) to[short,i^=$e_{12}$] (3,1);
       \filldraw[black] (0,0) circle (1pt) node[anchor=north]{$v_1$};
       \filldraw[black] (2,0) circle (1pt) node[anchor=north]{$v_2$};
       \filldraw[black] (3,0) circle (1pt) node[anchor=north]{$v_3$};
       \filldraw[black] (1,.7) circle (1pt) node[anchor=north]{$v_4$};
       \filldraw[black] (1,2) circle (1pt) node[anchor=south]{$v_5$};
       \filldraw[black] (3,2) circle (1pt) node[anchor=south]{$v_6$};
       \filldraw[black] (3,1) circle (1pt) node[anchor=west]{$v_7$};
    \draw[thin, <-] (1,0.3) node{$l_1$}  ++(-60:0.16) arc (-60:170:0.16);
    \draw[thin, <-] (1.4,.7) node{$l_2$}  ++(-60:0.16) arc (-60:170:0.16);
    \draw[thin, <-] (.6,.7) node{$l_3$}  ++(-60:0.16) arc (-60:170:0.16);
    \draw[thin, <-] (2,.85) node{$l_4$}  ++(-60:0.16) arc (-60:170:0.16);
    \draw[thin, <-] (2.7,.35) node{$l_5$}  ++(-60:0.16) arc (-60:170:0.16);
    \draw[thin, <-] (2.55,1.6) node{$l_6$}  ++(-60:0.16) arc (-60:170:0.16);
    %\draw[thin, <-] (1.5,1)node{$f_2$}  ++(-60:0.3) arc (-60:170:0.3);
    %\draw[thin, <-] (0.5,1.5)node{$f_3$}  ++(-60:0.3) arc (-60:170:0.3);
    \draw[thin,->] (.1,-.2) to (3.3,-.2) to[short,l_=$l_7$] (3.3,1.2);
    \end{circuitikz}
    \caption{
      A directed graph with $\mathcal V = \{v_1, v_2, v_3, v_4, v_5, v_6, v_7\}$, $\mathcal E = \{e_1, e_2, e_3, e_4, e_5, e_6, e_7, e_8, e_9, e_{10}, e_{11}, e_{12}\}$, and $\mathcal L = \{ l_1, l_2, l_3, l_4, l_5, l_6, l_7\}$. 
    }
    \end{subfigure}
    \begin{subfigure}{0.49\linewidth}
    \begin{circuitikz}[scale=2.5]
       \draw[thick,capc] (0,0) to[capacitor,i_=$e_1$,color=capc] (2,0);
       \draw[thick,sc] (2,0) to[resistor,i_=$e_2$,color=sc] (1,2); 
       \draw[thick,indc] (1,2) to[inductor,i_=$e_3$,color=indc] (0,0);
       \draw[thick,indc] (0,0) to[inductor,i^=$e_4$,color=indc] (1,.7);
       \draw[thick,pc] (1,.7) to[resistor,i_=$e_5$,color=pc](1,2);
       \draw[thick,capc] (2,0) to[capacitor,i_=$e_6$,color=capc] (1,.7);
       \draw[thick,sc] (2,0) to[resistor,i_=$e_7$,color=sc] (3,0);
       \draw[thick,indc] (3,0) to[inductor,i_=$e_8$,color=indc] (3,1); 
       \draw[thick,capc] (3,1) to[capacitor,i_=$e_9$,color=capc] (3,2); 
       \draw[thick,indc] (3,2) to[inductor,i_=$e_{10}$,color=indc] (1,2);
       \draw[thick,pc] (1,2) to[resistor,i_=$e_{11}$,color=pc] (3,1);
       \draw[thick,capc] (2,0) to[capacitor,i^=$e_{12}$,color=capc] (3,1);
       \filldraw[black] (0,0) circle (1pt) node[anchor=north]{$v_1$};
       \filldraw[black] (2,0) circle (1pt) node[anchor=north]{$v_2$};
       \filldraw[black] (3,0) circle (1pt) node[anchor=north]{$v_3$};
       \filldraw[black] (1,.7) circle (1pt) node[anchor=north]{$v_4$};
       \filldraw[black] (1,2) circle (1pt) node[anchor=south]{$v_5$};
       \filldraw[black] (3,2) circle (1pt) node[anchor=south]{$v_6$};
       \filldraw[black] (3,1) circle (1pt) node[anchor=west]{$v_7$};
    \draw[thin, <-] (1,0.3) node{$l_1$}  ++(-60:0.16) arc (-60:170:0.16);
    \draw[thin, <-] (1.2,1.0) node{$l_2$}  ++(-60:0.16) arc (-60:170:0.16);
    \draw[thin, <-] (.8,1.0) node{$l_3$}  ++(-60:0.16) arc (-60:170:0.16);
    \draw[thin, <-] (2,.85) node{$l_4$}  ++(-60:0.16) arc (-60:170:0.16);
    \draw[thin, <-] (2.7,.35) node{$l_5$}  ++(-60:0.16) arc (-60:170:0.16);
    \draw[thin, <-] (2.55,1.6) node{$l_6$}  ++(-60:0.16) arc (-60:170:0.16);
    %\draw[thin, <-] (1.5,1)node{$f_2$}  ++(-60:0.3) arc (-60:170:0.3);
    %\draw[thin, <-] (0.5,1.5)node{$f_3$}  ++(-60:0.3) arc (-60:170:0.3);
    \draw[thin,->] (.1,-.2) to (3.3,-.2) to[short,l_=$l_7$] (3.3,1.2);
    \end{circuitikz}
    \caption{For the circuit at hand, the edges are colored according to set inclusion. The sets of interest for this diagram are $\mathcal C = \{e_1, e_6, e_{12}, e_9\}$, $\mathcal I = \{e_3, e_4, e_8, e_{10}\}$, $\mathcal R = \mathcal E \setminus(\mathcal C \cup \mathcal I)$, $\mathcal S = \{e_2, e_7\}$, and $\mathcal P = \{e_5, e_{11}\}$. The sets $\mathcal S$ and $\mathcal P$ are chosen to be consistent with the algorithm given in Section \ref{sec:magic}, but any choice of $\mathcal S$ and $\mathcal P$ such that $\mathcal S \cup \mathcal P = \mathcal R$ is physically acceptable. }
    \end{subfigure}
    \begin{subfigure}{0.6\linewidth}
        \centering
        \includegraphics[width=\linewidth]{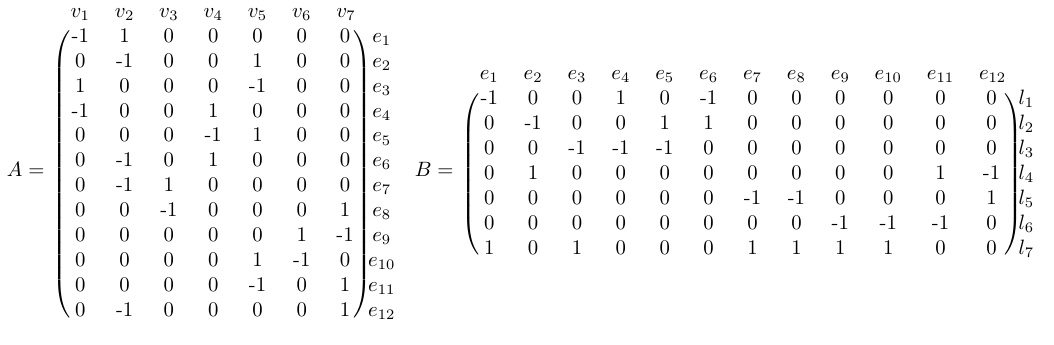}
        \caption{The matrices $A$ and $B$ for the graph and circuit drawn in (a) and (b).}
    \end{subfigure}
    \begin{subfigure}{0.39\linewidth}
        \begin{circuitikz}
           \draw[thick,capc] (0,0) to[capacitor,color=capc,l_={Capacitor}] (0,1);
           \draw[thick,indc] (0,1.2) to[inductor,color=indc,l_={Inductor}] (0,2.3);
           \draw[thick,pc] (3,1.2) to[resistor,color=pc,l_={Element of $\mathcal P$}] (3,2.3);
           \draw[thick,sc] (3,-.1) to[resistor,color=sc,l_={Element of $\mathcal S$}] (3,1);
        \end{circuitikz}
        \caption{Depicted above is the color scheme that will be employed in all diagrams in this paper.}
    \end{subfigure}
    \caption{Pictured above are a pair of drawings. The drawing (a) is simply a graph, while (b) is a full--fledged circuit. Since the boundary maps $A$ and $B$ depend only upon the underlying graph structure of a circuit, $A$ and $B$ may be read just as easily off of the drawing in (a) as off of the drawing (b). On the other hand, various projectors onto edge sets may be read only off of the drawing in (b). }
    \label{fig:example}
\end{figure}
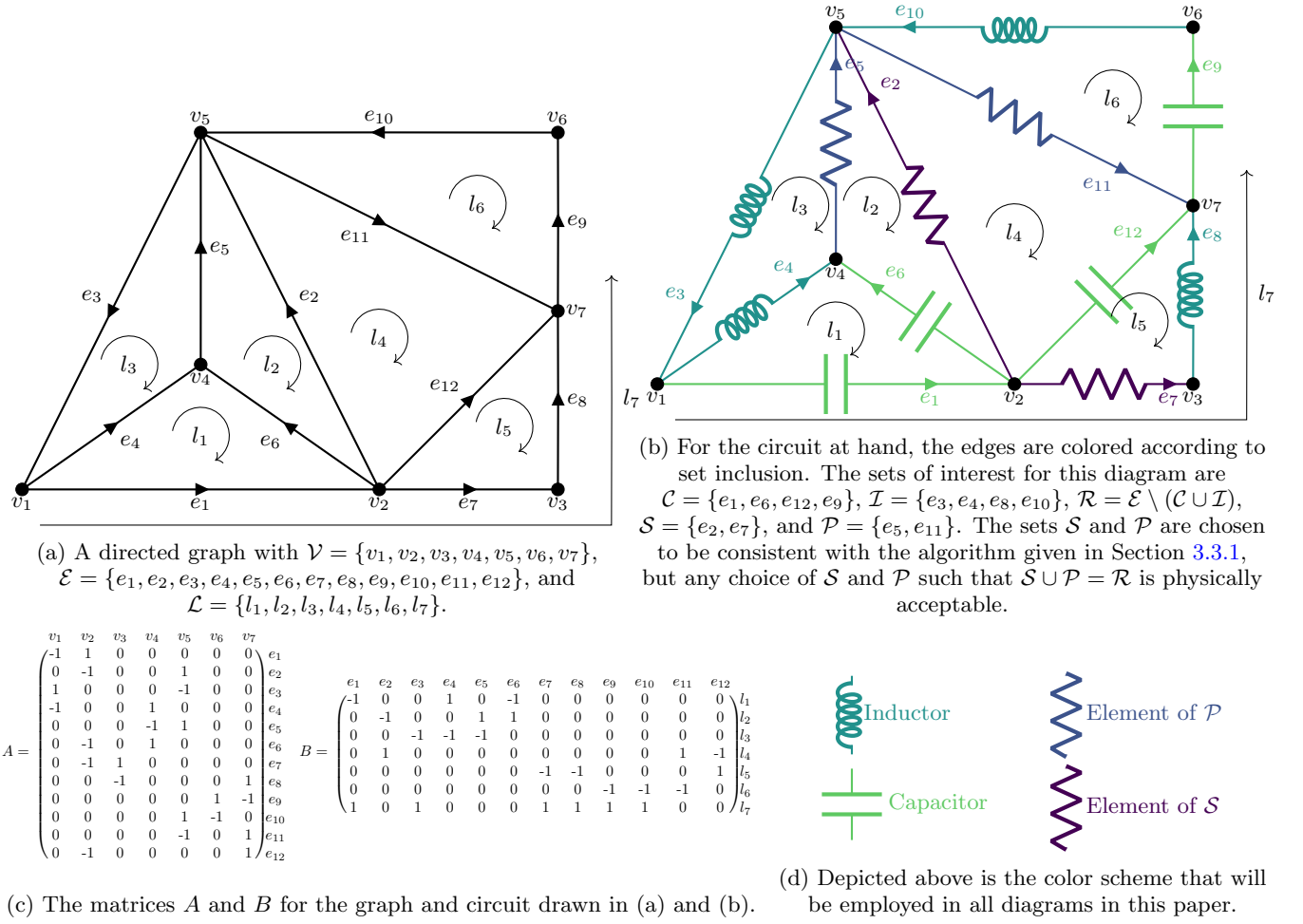

We are now ready to define a circuit as a graph with edges of different ``flavors".  Obviously, the definition below is not meant to capture everything that could reasonably called an electrical circuit; instead, it is serves to clearly delineate the scope of what we will study in this paper.

\begin{defn}[Circuit]
    In this paper, we define a circuit to be a graph $G$ with a partition of the edge set $\mathcal{E} = \mathcal{C} \cup \mathcal{I} \cup \mathcal{R}$, where these three sets are all disjoint.  Physically, this partition of edges into sets $\mathcal R$, $\mathcal C$, and $\mathcal I$ will correspond to resistive, capacitive, and inductive elements of the circuit, respectively. 
We define projectors on $\mathcal D(\mathcal E)$ which restrict to one of the three types of edges: 
\begin{subequations}
    \begin{align}
        P_I = \sum_{e \in \mathcal I} |e \rangle \langle e | , \\ 
        P_C = \sum_{e \in \mathcal C} |e \rangle \langle e | , \\ 
        P_R = \sum_{e \in \mathcal R} |e \rangle \langle e |.
    \end{align}
\end{subequations}
\end{defn}

For reasons that will be clear as we develop the formalism, we will further partition the set $\mathcal R$ into the disjoint sets $\mathcal S$ and $\mathcal P$: $\mathcal{R} = \mathcal{S}\cup \mathcal{P}$ with $\mathcal{S}\cap \mathcal{P} = \emptyset$.  The sets $\mathcal S$ and $\mathcal P$ will, roughly, correspond to resistors that are in series with a desired element and resistors that are in parallel with a desired element.  These ``desired" elements will be specific, useful parasitic elements when we discuss nonlinear circuits in Section \ref{sec:nonlinear}. The choice of $\mathcal{S}$ and $\mathcal{P}$ is not unique, and we defer a detailed discussion of good choices to Algorithm \ref{algo:sp}.  Nevertheless, we define $P_R = P_S+P_P$ where \begin{subequations}
    \begin{align}
        P_S = \sum_{e \in \mathcal S} |e \rangle \langle e | , \\ 
        P_P = \sum_{e \in \mathcal P} |e \rangle \langle e | .
    \end{align}
\end{subequations}
%As we will see, for linear resistors, any partition of $\mathcal R$ is equivalent since constitutive relations and Kirchhoff's laws may be derived from the Lagrangian that is the subject of this manuscript in for any choice of $\mathcal S$ and $\mathcal P$. A convenient choice of $\mathcal S$ and $\mathcal P$ serves only the function of making calculations simpler.

The dynamical degrees of freedom of a circuit are a subset of flux variables $\phi_v$ defined at each vertex of a graph, and charge variables $q_l$ defined on each loop in a graph.  More abstractly we could think of $|\phi\rangle \in \mathcal{D}(\mathcal{V})$ and $|q\rangle \in \mathcal{D}(\mathcal{L})$. 
For each edge $e\in \mathcal{I}$ or $e\in \mathcal{C}$, we define an energy 
\begin{equation}\label{eqn:energies}
    E_e = \begin{cases}
        \displaystyle E_e^I\left(\sum_v A_{ev}\phi_v \right) & e \in \mathcal I \\ 
        \displaystyle E_e^C\left(\sum_l q_l B_{le} \right) & e \in \mathcal C
    \end{cases}.
\end{equation}
\begin{equation}
    R = \sum_{e \in \mathcal R} R_e |e\rangle\langle e|. 
\end{equation}
As a slight abuse of notation, we will write 
\begin{equation}
    R^{-1} = \sum_{e \in \mathcal R} \frac{1}{R_e} |e\rangle\langle e|.
\end{equation}
While $R^{-1}$ is not precisely the inverse of $R$, it satisfies $R R^{-1} = R^{-1} R = P_R$, and in our formalism $R$ will never appear without $P_R$.   To build our formalism for dissipative circuits, we need to make one further assumption:
%Clearly, $R$ is invertible if $\lambda\ne 0$. Notice that $R$ only appears next to $P_R$, so the choice of $\lambda$ is unimportant.   
%%\AO{will postpone nonlinear resistors until I have thought about them more} \andy{I think at the level of MSR, nonlinear resistors will appear simple; however, the conversion to fokker-planck equation (which in turn may require using a modified choice of $R$...) could be subtle.  nonlinear resistors will just have $R_e$ dependent on other $q$ and $\phi$ variables?}

\begin{assumpn}[Thermal equilibrium]
\eqref{eq:thermaleq} describes the statistical steady state of the circuit, where \begin{equation}
        E = \sum_{e\in \mathcal{I}\cup \mathcal{C}} E_e.
    \end{equation}
    The circuit is in thermal equilibrium, with energy stored in inductors and capacitors.
\end{assumpn}

\subsection{The dissipative Lagrangian}
We are now ready to write down a locally valid dissipative Lagrangian (Definition \ref{def:TRS2}) for general RLC circuits, as defined in the previous subsection.  As has been discussed previously \cite{rymarz_consistent_2022,Parra-Rodriguez:2023ykw,osborne2023symplectic}, it is possible that there are global singularities which require care to resolve, and we do not focus on resolving all such possible singularities in the present paper.  We now introduce a summary of the notation we will use in what follows: 
\begin{defn}[Dissipative Lagrangian for RLC circuit]
    The physical degrees of freedom in a circuit are \begin{equation}
        |X\rangle = \begin{pmatrix} \sum_v\phi_v|v\rangle \\ \sum_l q_l|l\rangle \rangle \end{pmatrix} = \begin{pmatrix} |\phi \rangle \\ |q\rangle\end{pmatrix},
    \end{equation}
    and we define ``noise variables" \begin{equation}\label{eqn:notation}
    |\Pi\rangle = \begin{pmatrix} \sum_v \pi_v |v\rangle \\ \sum_l\sigma_l|l\rangle \rangle \end{pmatrix} = \begin{pmatrix} |\pi \rangle \\ |q\rangle \end{pmatrix}
\end{equation}
for each.  The following matrices will also appear frequently: \begin{subequations}\label{eqn:matrices}
    \begin{align}
        \Lambda&= \begin{pmatrix} A & 0 \\ 0 & B^\intercal \end{pmatrix} \\ 
        %\Delta &= \begin{pmatrix} B & 0 \\ 0 & A^\intercal \end{pmatrix} \\ 
        % \Gamma &= \begin{pmatrix} \frac{P_R}{4 R} & \frac{3}{4}P_R + P_C \\ \frac{3}{4} P_R + P_I & \frac{R P_R}{4}\end{pmatrix}\\
        \Gamma &= \begin{pmatrix} 
          P_P R^{-1} & P_S + P_C \\ P_P + P_I & P_S R
      \end{pmatrix}\\
        K &= \Lambda^\intercal \Gamma \Lambda \\ 
    K^{\pm} &= \frac{1}{2} (K \pm K^\intercal).
    \end{align}
\end{subequations}
Intuitively: $\Lambda$ plays the role of a boundary map on the phase space (whose coordinates are captured by $|X\rangle$).  With this notation, the dissipative Lagrangian for an RLC circuit is 
\begin{equation}\label{eqn:msr}
L = \langle \Pi| K |\dot X\rangle + \mathrm i T \langle \Pi| K^+ |\Pi \rangle + \langle \Pi|h\rangle,
\end{equation}
where we have defined $|h\rangle$ to be the gradients of the nondissipative energy: \begin{equation}
    |h\rangle = \begin{pmatrix} \displaystyle
        \sum_v \dfrac{\partial E}{\partial \phi_v} |v\rangle \\ \displaystyle
        \sum_l \dfrac{\partial E}{\partial q_l}|l\rangle .
    \end{pmatrix}
\end{equation}
\end{defn}

The definiton above serves as one of the main results of the paper.  It is, in a very precise sense, the correct Lagrangian for a dissipative circuit.  As we will see, it encodes the correct fluctuation-dissipation theorem (i.e. Johnson noise) for resistors, as well as incorporating the physical assumptions of Kirchoff's current and voltage laws.  In Section \ref{sec:rayleigh}, we will see that our results reduce to those based on Rayleigh dissipative functions when stochastic fluctuations are neglected.  We now explain that:

\begin{obs}\label{thm:const}
  The Euler-Lagrange equations for $\dot \phi_v$ and $\dot q_l$ associated with \eqref{eqn:msr} are only satisfied if Kirchoff's current and voltage laws for a circuit are obeyed, with each element of the circuit obeying the correct constitutive relations.
\end{obs}
 
To justify Observation \ref{thm:const}, first notice that the equations of motion for $\dot \phi_v$ and $\dot q_l$ are:
\begin{subequations}\label{eqn:eoms}
\begin{align}
  0=\frac{\delta S}{\delta \pi_v} &= \begin{pmatrix} \langle v| & 0 \end{pmatrix} \left[ K |\dot X\rangle + 2 \mathrm i T K^+ |\Pi\rangle +  |h\rangle \right] \notag  \\ 
                                & =\langle v| A^\intercal P_P R^{-1} A | \dot \phi\rangle + \langle v | A^\intercal P_S B^\intercal |\dot q\rangle + \langle v | A^\intercal P_C B^\intercal |\dot q\rangle + \frac{\partial E}{\partial \phi_v}  + 2 \mathrm i T \langle v | A^\intercal P_P R^{-1} A | \pi\rangle\notag  \\  
                                &= \sum_{e}  A_{ev}\left[\mathbb I [e \in P]\frac{1}{R_e} \left(\dot \phi_e + 2 \mathrm i T \pi_e \right) + \mathbb I [e \in S] \dot q_e + \mathbb I[e \in \mathcal C] \dot q_e + \mathbb I[e \in \mathcal I] \frac{\partial E_e}{\partial \phi_e}\right]  \\ 
    0=\frac{\delta S}{\delta\sigma_l} &= \begin{pmatrix}0& \langle l| \end{pmatrix} \left[ K |\dot X\rangle + 2 \mathrm i T K^+ |\Pi\rangle +  |h\rangle \right] \notag \\
                                    &= \langle l | B (P_P + P_I)A| \dot \phi\rangle  + \langle l | B P_S R B^\intercal| \dot q\rangle + \frac{\partial E}{\partial q_l} + 2 \mathrm i T \langle l|  B P_S B^\intercal | \sigma \rangle \notag \\
                                    &= \sum_{e} B_{le} \left[ \mathbb I[e \in P] \dot \phi_e + \mathbb I[e \in I] \dot \phi_e + \mathbb I[e \in S]R_e\left(\dot q_e + 2 \mathrm i T \sigma_e \right) + \mathbb I[e \in \mathcal C] \frac{\partial E_e}{\partial q_e}
                                    \right]
    \end{align}
\end{subequations}
While these equations are not in the ``standard form" of a Langevin equation presented in \eqref{eqn:sde}, they are actually in a highly amenable form to explain how they encode stochastic Kirchoff's laws.  Indeed, recall that 
    \begin{subequations}
    \begin{align}
        P_P + P_S + P_C + P_I &= \mathbb I, \\ 
        BA &= 0.
        \end{align}
    \end{subequations}
As these two facts imply, e.g., \begin{equation}
    B(P_P+P_I)A = -B(P_S+P_C)A,
\end{equation}
we then write
    \begin{subequations}\label{eqn:simplified}
        \begin{align}
            0=\frac{\delta S}{\delta \pi_v} = \sum_{e \in \mathcal{P}} A_{ev} \left[ \frac{1}{R_e} (\dot\phi_e + 2 \mathrm i T \pi_e ) - \dot q_e\right] + \sum_{e \in \mathcal I} A_{ev}\left[ \frac{\partial E}{\partial \phi_e} - \dot q_e\right] \label{eqn:simplified1} \\ 
            0=\frac{\delta S}{\delta \sigma_l} = \sum_{e \in \mathcal{S}}B_{le}\left[ R_e (\dot q_e + 2\mathrm i T \sigma_e ) - \dot \phi_e\right] + \sum_{e \in \mathcal C} B_{le} \left[\frac{\partial E}{\partial q_e} - \dot \phi_e\right].
        \end{align}
    \end{subequations}
    Since we have demanded that the deterministic, energetic part of the Lagrangian depends on $q_l$ only as $\sum_l q_l B_{le}$ and $\phi_v$ only as $\sum_v A_{ev}\phi_v$, we have written 
    \begin{equation}\label{eqn:convention}
        \frac{\partial E}{\partial \phi_v} = \sum_{e \in \mathcal I} A_{ev}\frac{\partial E}{\partial \phi_e}. 
    \end{equation}
    We have made a similar observation for the derivative of $E$ with respect to $q_l$.    Notice that if each term in brackets in \eqref{eqn:simplified} vanished separately, then Kirchoff's current and voltage laws would all be obeyed. To see that this is indeed the case, we prove:
    
    \begin{prop}[Gauge symmetry]\label{prop:gauge}
    The equations of motion \eqref{eqn:simplified}, and the dissipative Lagrangian \eqref{eqn:msr}, are both invariant under the following \textbf{gauge symmetries} (redundancies in coordinate descriptions): \begin{itemize}
        \item Given any $|l\rangle \in \mathcal{D}(\mathcal{L})$ such that $(P_P+P_I)B^{\mathsf{T}}|l\rangle = B^{\mathsf{T}}|l\rangle$, for arbitrary function $f(t)$, $L$ is invariant under
        \begin{equation}
            |q\rangle \rightarrow |q\rangle + f(t)|l\rangle. \label{eq:qgauge}
        \end{equation}
        In other words, each loop in $\mathcal{P}\cup\mathcal{I}$ removes one $q_l$ degree of freedom.
        \item Given any $|u\rangle \in \mathcal{D}(\mathcal{V})$ such that  $(P_S+P_C)A|u\rangle = A|u\rangle$, $L$ is invariant under \begin{equation}
            |\phi\rangle \rightarrow |\phi\rangle + f(t)|u\rangle.
        \end{equation}
        In other words, each cut in $\mathcal{S}\cup\mathcal{C}$ removes one $\phi_v$ degree of freedom.
    \end{itemize} 
    Therefore, the constitutive relations in \eqref{eqn:simplified}  hold on each edge of the circuit.
    \end{prop}
    \begin{proof}
The first line of (\ref{eqn:simplified}) may be read as 
    \begin{equation} \label{eq:APIPS}
        0= \sum_e A^\intercal (P_I + P_P) |\gamma\rangle
    \end{equation}
    with $|\gamma \rangle = \sum_e \gamma_e|e\rangle$ taking on the values
    \begin{equation}
        \gamma_e = 
        \begin{cases}
          \dfrac{1}{R_e}(\dot\phi_e - 2 \mathrm i T \pi_e ) - \dot q_e & e \in P \\ 
          \dfrac{\partial E}{\partial \phi_e} - \dot q_e & e \in \mathcal I
        \end{cases}.
    \end{equation}
    We now invoke the following useful fact, stated without proof:
    \begin{prop}[\cite{fluxcharge}, Corollary B.5]\label{prop:Mnullvec}
        If $|\gamma\rangle \neq 0$ satisfies \eqref{eq:APIPS}, there exists some vector $|\psi\rangle$ such that 
        \begin{equation}
            (P_I + P_P) |\gamma\rangle = B^\intercal |\psi\rangle. \label{eq:APIPS2}
        \end{equation}
        Likewise, if $|\gamma'\rangle \neq 0$ satisfies 
        \begin{equation}
            0 = B (P_I + P_P)|\gamma'\rangle
        \end{equation}
        then there exists some vector $|\varphi\rangle$ such that 
        \begin{equation}
            (P_I + P_P) |\gamma'\rangle = A|\varphi \rangle.
        \end{equation}
    \end{prop}
    For the case \eqref{eq:APIPS}, we can now use \eqref{eq:APIPS2} to conclude that $|\gamma\rangle \ne 0$ only if \eqref{eq:qgauge} is obeyed for some loop contained in $\mathcal{P}\cup\mathcal{I}$, as stated in the proposition.   Such a loop modifies \eqref{eqn:simplified1} to \begin{equation}
        \frac{\delta S}{\delta \pi_v} = \sum_{e \in P} A_{ev} \left[ \frac{1}{R_e} (\dot\phi_e + 2 \mathrm i T \pi_e ) - \dot q_e\right] + \sum_{e \in \mathcal I} A_{ev}\left[ \frac{\partial E}{\partial \phi_e} - \dot q_e\right] - \sum_{e\in\mathcal{P}\cup \mathcal{I}} \dot f(t) \psi_l B_{le}A_{ev}
    \end{equation}
    for some coefficients $\psi_l \in \lbrace 0,\pm 1\rbrace$ that encode the particular loop in $\mathcal{P}\cup \mathcal{I}$.  Since the loop by construction has edges entirely within $\mathcal{P}\cup \mathcal{I}$, we can use $BA=0$ to show that the last term above vanishes.  Therefore, we are free to take the solution where $\gamma_e=0$, and the constitutive relations hold.

    The result that the constitutive relation for resistors in $\mathcal S$ and capacitors holds is acquired very similarly, and we do not present it explicitly.

Lastly we confirm that \eqref{eqn:msr} is invariant. By construction of $K$ in \eqref{eqn:matrices}, notice that the two gauge transformations (when packaged into $|X\rangle$) are right null vectors of $K$.  The gauge transformation \eqref{eq:qgauge} cannot change the $q$ variables on any element of $\mathcal{C}$, since the loop is contained entirely in $\mathcal{P}\cup\mathcal{I}$; hence, $|h\rangle$ is independent of \eqref{eq:qgauge}.\footnote{This is easy to see by picking $|l\rangle$ in \eqref{eq:qgauge} as a basis vector in $\mathcal{L}$, which implies that $E_C^e$ must all be independent of $q_l = \langle q|l\rangle$.} A similar argument holds for flux gauge transformations.
    \end{proof}

   Since the constitutive relations for circuit elements are always obeyed, it follows that circuits are physically equivalent to their simplified (via ordinary addition of circuit elements in series and parallel) versions. 
    While the simplification of networks via constitutive relations is certainly well understood, it it not necessarily transparent from the dissipative Lagrangian alone how this arises. We will see how to achieve this in Section \ref{sec:elim}.

%Because Theorem \ref{thm:const} provides constitutive relations, it is possible, then, to interpret $\frac{\delta S}{\delta \pi_v}$ and $\frac{\delta S}{\delta \sigma_l}$ as a demand that Kirchhoff's current and voltage laws are obeyed on $v$ and $l$ respectively. Since the proof of Theorem \ref{thm:const} did not suppose resistors were chosen to be in $\mathcal S$ or $\mathcal P$ for any particular reason, it follows that the proof and it's corollaries hold for an arbitrary choice. Thus, one is afforded the freedom to choose $\mathcal S$ and $\mathcal P$ in any manner that is convenient. 

\subsection{Removing constraints}\label{sec:simplify}

  In many approaches to dissipative circuit mechanics \cite{Parra-Rodriguez:2023ykw,brayton,bahar_generalized,weissmathis}, Kirchhoff's laws are incorporated into a theory in the form of constraints, generalized forces, etc. Indeed, as will describe in detail, the dissipative Lagrangian \eqref{eqn:msr} itself is generically subject to constraints -- not all degrees of freedom are independent. A systematic removal of constrained degrees of freedom played a crucial role in finding an efficient algorithm for finding the Hamiltonian formulation of circuit mechanics, and ultimately circuit quantization of an equal number of  charges and flux degrees of freedom.  Clearly in general, \eqref{eqn:msr} need not have an equal number of $\phi_v$ and $q_l$ degrees of freedom.
  We therefore endeavor to find a complete classification, and systematic prescription to remove, all constrained variables in a transparent way. 
  
In our framework, which will closely follow \cite{osborne2023symplectic,fluxcharge}, one motivation for the importance of removing constraints can be more concretely seen as follows. Lagrangians of the form (\ref{eqn:msr}) are not in the MSR form (\ref{eq:MSRquadratic}). In order to study fully nonlinear stochastic circuit dynamics, it is essential to obtain a theory of the form \eqref{eq:MSRquadratic} so that we may invoke Definition \ref{def:FPEgood}. Naively, all that is required to massage a quadratic (in noise variables) dissipative Lagrangian into a quadratic MSR Lagrangian is the ability to invert $K$. However, $K$ is singular.
  Nonetheless, if $K$ were of the form 
  \begin{equation}\label{eqn:simpleK}
      K = \begin{pmatrix}
          \tilde K & 0 \\
          0 & 0
      \end{pmatrix}
  \end{equation}
  in some basis with $\tilde K$ invertible, we have hope of obtaining (\ref{eq:MSRquadratic}) by treating the lower block of $|\Pi\rangle$ and $|X\rangle$ as Lagrange multipliers that enforce constraints on this set of degrees of freedom, rendering them not independent.  By enforcing such constraints, we could then reduce the degrees of freedom to those in the block where $\tilde K$ is invertible.
  %The process of finding $\tilde K$ and the basis where $K$ takes on the form (\ref{eqn:simpleK}) is equivalent to ``integrating out" all of the constraints present in a dissipative Lagrangian (\ref{eqn:msr}). 
  Happily, this task may be accomplished in some level of generality.

\subsubsection{A good choice of $\mathcal S$ and $\mathcal P$}\label{sec:magic}
%It is often best to choose $\mathcal S$ and $\mathcal P$ so that the constraints mentioned above \andy{namely, that the loop exception to Prop 3.8 does not occur, if at all possible?}\AO{not possible to forbid the loop exception of prop 3.8. I plan to delete this sentence} are simple. 

In our endeavor to remove constraints, it turns out that our freedom to choose $\mathcal S$ and $\mathcal P$ according to convenience is  surprisingly  powerful. 
While it  is generally true that constraints arising from (\ref{eqn:msr}) may be solved equivalently with any choice of $\mathcal S$ and $\mathcal P$, the problem of formally resolving all possible constraints with a generic choice is nontrivial. 
As such, we will leverage our freedom of choice and specify an algorithm for choosing $\mathcal S$ and $\mathcal P$ that is amenable to both formal discussion and practical calculation. 
While eventually we will give an explicit enumeration of all constraints present in a generic disspative Lagrangian, and prove that they are all solvable, we first present our particular choice of $\mathcal S$ and $\mathcal P$ in the form of an algorithm:
%While it can be the case that other choices of $\mathcal S$ and $\mathcal P$ are more convenient, the choice given by this section is of particular use for rigorous demonstrations. We present this choice in the form of an algorithm. 
%We execute the following algorithm:
\begin{algo}\label{algo:sp}
    Order the edges in $\mathcal R$ such that $\mathcal R = \{e_1,e_2,\dots,e_{|\mathcal R|}\}$. Define the set $\mathcal S^{(0)}$ to include every edge in $\mathcal R$, and define $\mathcal P^{(0)}$ to be empty. For $i = 1$ to $i = |\mathcal R|$, perform the following recursive operation: \begin{itemize} 
    \item If $e_i$ is in a loop involving only edges in $\mathcal S^{(i-1)}$ and $\mathcal C$, then define 
    \begin{subequations}
        \begin{align}
            \mathcal S^{(i)} = \mathcal S^{(i-1)} \setminus \{e_i\} \\
            \mathcal P^{(i)} = \mathcal P^{(i-1)} \cup \{e_i\}
        \end{align}
    \end{subequations}
    \item Otherwise, define 
    \begin{subequations}
                \begin{align}
            \mathcal S^{(i)} = \mathcal S^{(i-1)} \\ 
            \mathcal P^{(i)} = \mathcal P^{(i-1)}
        \end{align}
    \end{subequations}
    \end{itemize}
    Let $\mathcal S = \mathcal S^{(|\mathcal R|)}$ and $\mathcal P = \mathcal P^{(|\mathcal R|)}$. 
\end{algo}

The following results will make clear the utility of the choice of $\mathcal S$ and $\mathcal P$ espoused above.
 \begin{prop}\label{thm:evencut}
     Let $C \subset \mathcal E$ be a cut of some graph $G$, and let $L \subset \mathcal E$ be a loop of the same graph. There exists some natural number $n \in \mathbb N$ such that 
     \begin{equation}
         \left| C \cap L\right|  = 2n.  \label{eq:CL2n}
     \end{equation}
     In other words, a loop and a cut cannot have an odd number of edges in common.
 \end{prop}
 \begin{proof}
    A cut $C$ induced a partition of $\mathcal V$ into the sets $\mathcal V_1$ and $\mathcal V_2$. Define a function $f: \mathcal V \rightarrow \mathbb Z_2$ such that 
    \begin{equation}
        f(v) = \begin{cases}
            0 & v \in \mathcal V_1 \\ 
            1 & v \in \mathcal V_2
        \end{cases}.
    \end{equation}
    Let the vertices $v_1, v_2, \dots, v_k$ be the vertices traversed by the loop $L$. If an edge $e = (u,v)$ is in $C$, then
    \begin{equation}
        f(u) - f(v) = 1 \mod 2
    \end{equation}
    since $e$ is an edge that connects $\mathcal V_1$ to $\mathcal V_2$. The sum, with $v_{-1} = v_k$,  
    \begin{equation}\label{eqn:edgecount}
        \sum_{i=1}^k (f(v_i) - f(v_{i-1}) )= 0 \mod 2
    \end{equation}
    since it is a telescoping sum.
    However, (\ref{eqn:edgecount}) also counts the number of edges in common between $L$ and $C$ mod 2. Hence, the number of shared edges must be even if (\ref{eqn:edgecount}) holds; this implies \eqref{eq:CL2n}.
 \end{proof}
 \begin{cor}\label{cor:goodpick}
     If $\mathcal S$ and $\mathcal P$ are chosen according to Algorithm \ref{algo:sp}, then:
     \begin{enumerate}
         \item If there is a loop of edges $l \subset \mathcal E$ such that $l \cap (\mathcal I \cup \mathcal P) = \emptyset$, then $l \subseteq \mathcal C$. 
         \item If there is a cut of edges $c \subset \mathcal E$ such that $c \cap (\mathcal C \cup \mathcal S) = \emptyset$, then $c \subseteq \mathcal I$.
     \end{enumerate}
 \end{cor}
 \begin{proof}
   Point 1 is satisfied by construction; it is only point 2 which must be demonstrated.
   Choose a cut $c$ satisfying the hypothesis of point 2 and suppose $|c \cap \mathcal P| = n >0$. Every edge in $\mathcal P$ is in \emph{some} loop $l$ otherwise consisting only of edges in $\mathcal C \cup \mathcal S$. By Prop. \ref{thm:evencut}, $|c \cup l| = 0 \mod 2$. However, there is only one edge in $l$ outside of $\mathcal C \cup \mathcal S$, so $c$ must also contain some edge in $\mathcal S \cup \mathcal C$. The assumption that $|\mathcal{C}\cap \mathcal{P}|>0$ was a contradiction, hence $\mathcal{C}\subseteq \mathcal{I}$.
   \end{proof}
For the remainder of the paper, we take for granted that we can find such a good partition of $\mathcal{R}$ into $\mathcal S$ and $\mathcal P$, which satisfies the criteria of Corollary \ref{cor:goodpick}.

\subsubsection{Spanning tree coordinates}\label{sec:alg}
We now analyze $K$.  Firstly, it is very straightforward to show that $K$ is positive semidefinite. 
The symmetric part of $K$ is given by 
\begin{equation}
  K^+ = \Lambda^\intercal 
  \begin{pmatrix}
    P_P R^{-1} & \frac{1}{2}\mathbb I\\
    \frac{1}{2}\mathbb I & P_S R 
  \end{pmatrix}
  \Lambda.
\end{equation}
Since 
\begin{equation}
  \Lambda^\intercal 
  \begin{pmatrix}
    0 & \mathbb I \\
    \mathbb I & 0 
  \end{pmatrix}
  \Lambda = 0, 
\end{equation}
we may write 
\begin{equation}
  K^+ = 
  \begin{pmatrix}
    A^\intercal R^{-1} P_P A & 0 \\
    0 & B P_S R B^\intercal
  \end{pmatrix},
\end{equation}
which is clearly positive semidefinite provided that $R$ is positive definite (which follows from $R_e \ge 0$). 

\begin{defn}[Connection matrix]
The antisymmetric part of $K$ is related to the \textbf{connection matrix}
 \begin{equation}\label{eqn:form}
   M = \frac{1}{2} B (P_C + P_S - P_P - P_I) A.
\end{equation} 
In particular, \begin{equation}
  K^- = 
  \begin{pmatrix}
    0 & M \\
    - M^\intercal & 0
  \end{pmatrix}.
\end{equation}
\end{defn}

Using this new definition, we see that Proposition \ref{prop:Mnullvec} shows the left (right) null vectors of $M$ are cycles or cuts of $G$ that lie entirely within the sets $\mathcal C\cup\mathcal S$ or $\mathcal P\cup\mathcal I$.  This leads us to quote the following important result from \cite{osborne2023symplectic}: 

\begin{thm}[Spanning tree coordinates]\label{thm:spanningtree}
    Consider circuit $G$ with decomposition $\mathcal{S}$ and $\mathcal{P}$ obeying Corollary \ref{cor:goodpick}.  Let $G^\prime$ be the (possibly disconnected) subcircuit of $G$ consisting only of the edges in $\mathcal C \cup \mathcal S$. 
Let $\mathcal T = \{e_1, e_2, \dots, e_k\}$ be a spanning tree of $G'$. 
Define for each $e \in \mathcal T$
\begin{equation}
    \Phi_e = \sum_v A_{e v}\phi_v.  
\end{equation}
There exists a matrix $\rho$ such that 
\begin{equation}\label{eqn:diffdefn}
    F = \frac{1}{2}\sum_{l,v} \left(\sum_{e \in \mathcal C \cup \mathcal S} - \sum_{e \in \mathcal I \cup \mathcal S}\right) q_l B_{le} A_{ev} \phi_v = \sum_{e \in \mathcal T} Q_e \Phi_e
\end{equation}
with 
\begin{equation}
    Q_e = \sum_l \rho_{el} q_l. 
\end{equation}
\end{thm}
\begin{proof}
Using the fact that 
\begin{equation}
    \frac{1}{2} B (P_C + P_S - P_I - P_P) A = B(P_C + P_S) A, 
\end{equation}
we rewrite 
\begin{equation}
    F  =\sum_{l, v} \sum_{e \in \mathcal C \cup \mathcal S} q_l B_{le} A_{ev}\phi_v.
\end{equation}
With our choice of $\mathcal T$ in hand, we further rewrite 
\begin{equation}
    F = \sum_{l,v} \sum_{e \in \mathcal T} q_l B_{le} \Phi_e + \sum_{l,v}\sum_{e \not\in\mathcal T} q_{l}B_{le}A_{ev}\phi_v.
\end{equation}
Since the path between any pair of nodes in a tree is unique, it follows that for each $e \in (\mathcal S \cup \mathcal C ) \setminus \mathcal T$, there exist some coefficients $\lambda_{ee'}$ such that
\begin{equation}
    \sum_v A_{ev} \phi_v = \sum_{e' \in \mathcal T} \lambda_{ee'} \Phi_{e'}
\end{equation}
and thus
\begin{equation}\label{eqn:sympterm}
    F = \sum_{l,v} \sum_{e\in\mathcal T}  \left( q_l B_{l e} + \sum_{e'  \in (\mathcal C \cup \mathcal S)\setminus \mathcal T} q_l B_{le'}\lambda_{e'e}\right) \Phi_e.
\end{equation}The object in parentheses defines $Q_e$ along with $\rho$.
\end{proof}
\begin{defn}[Spanning tree variable]
Let $G$ be an RLC circuit with dissipative Lagrangian $L$. Further suppose that $\mathcal S $ and $\mathcal P$ adhere to Corollary \ref{cor:goodpick}. 
The variable definitions $Q$ and $\Phi$ furnished by Theorem \ref{thm:spanningtree} will be referred to as \textbf{spanning tree variables}. 
\end{defn}

%A number of formal results from \cite{osborne2023symplectic} apply to the spanning tree variables. The first result of note is that (\ref{eqn:sympterm}) holds. This is transparent in our discussion, though it is useful enough to emphasize nonetheless.

The choice of spanning tree variables for a given circuit need not be unique.  Different choices of spanning tree variables are physically equivalent, and in the discussion that follows we simply choose one set to work with.  Notice that we could have alternatively chosen a spanning tree in $\mathcal P \cup \mathcal I$, but we have chosen the convention that is most similar to \cite{osborne2023symplectic}.  For details of these points, we refer the reader to the relevant discussion in \cite{osborne2023symplectic}.

\subsubsection{Reducing degrees of freedom in the dissipative Lagrangian}\label{sec:elim}

Qualitatively, spanning tree variables are a formally and practically convenient choice of variables because of the algorithmic nature of their definition. 
One can always define spanning tree variables for some circuit in a way that does not demand any creativity, and the Lagrangian for some circuit is always significantly simplified by such a choice. As a passing remark, the manipulations above are very amenable to automation and could easily be used as a subroutine in some software, such as \cite{chitta_computer-aided_2022}. 
Variables which are not expressible as a linear combination of spanning tree variables may be integrated out of any circuit Lagrangian provided that null vector constraints are soluble. 
Unlike the statements above, the proof that constraints are always soluble provided in earlier works are not applicable to (\ref{eqn:msr}), and must be revisited. It is this task  to which we now turn our attention. 

\begin{prop}\label{prop:eliminate}
Given dissipative Lagrangian \eqref{eqn:msr} of some nonsingular circuit that abides by Corollary \ref{cor:goodpick}, if $|\gamma \rangle\neq 0$ is a vector such that 
    \begin{equation}
        K^{-}|\gamma\rangle = 0,
    \end{equation}
    then the equation
    \begin{equation}
        0 = \sum_{l} \langle l| \gamma\rangle \frac{\delta S}{\delta \sigma_l} + \sum_v \langle v| \gamma\rangle \frac{\delta S}{\delta \pi_v} 
    \end{equation}
    may be used to eliminate one linear combination of physical variables from $L$ by algebraic solution. 
\end{prop}
\begin{proof} Proposition \ref{prop:gauge} shows that null vectors of $K^-$ which are cuts of edges in $\mathcal{C}\cup\mathcal{S}$, and loops of edges in $\mathcal{I}\cup\mathcal{P}$, are trivial and can be removed. Since Corollary \ref{cor:goodpick} forbids from appearing any loop involving both edges in $\mathcal C$ and $\mathcal S$, or a cut involving edges from $\mathcal{I}$ and $\mathcal{P}$, the remaining null vectors of $K^-$ are a loop of edges contained entirely in $\mathcal{C}$, or a cut of edges contained entirely in $\mathcal{I}$.  Notice that these latter possibilities are also null vectors of $K^+$.  

For a loop $l$ of capacitors, the charge on one capacitor must be determined in terms of the others as is required to satisfy Kirchhoff's voltage law. The assumption of nonsingularity is equivalent to the demand that such loop constraints always have a unique solution.  More precisely, notice that in \eqref{eqn:simplified}, the $\pi_v$ equation would not depend on $q_l=\langle q|l\rangle$, meaning that \begin{equation}
    0 = \sum_{e\in l}B_{le}\left[\frac{\partial E}{\partial q_e} - \dot\phi_e\right] = \frac{\partial E}{\partial q_l}. \label{eq:318q}
\end{equation}
The sum over $\dot \phi_e$ vanishes due to $BA=0$ (voltages around a loop vanish), while the second equality above follows from the relation between $q_l$ and $q_e$ in \eqref{eqn:simplified}.

The analysis of a cut involving only inductors follows analogously; if the cut partitions $\mathcal{V}$ into $\mathcal{V}_1\cup \mathcal{V}_2$, we find that \begin{equation}
    0 = \sum_{v\in\mathcal{V}_1} \frac{\partial E}{\partial \phi_v} =  \sum_{v\in\mathcal{V}_2} \frac{\partial E}{\partial \phi_v}.\label{eq:318phi}
\end{equation}
In either case, \eqref{eq:318q} or \eqref{eq:318phi} lead to the fixing of one constrained degree of freedom for each independent such loop or cut.  Although naively there are two constraints in \eqref{eq:318phi}, upon choosing the first such cut to be $\mathcal{V}_1=\mathcal{V}$ (which implies the freedom to pick a grounded node), each additional identified cut only removes one further degree of freedom \cite{osborne2023symplectic}.
\end{proof}

\begin{cor}
    Suppose $L$ is a dissipative Lagrangian with $\mathcal S$ and $\mathcal P$ chosen according to Algorithm \ref{algo:sp}. Choose a spanning tree $\mathcal T\subset \mathcal C \cup \mathcal S$. $L$ can be written in terms of only spanning tree variables,  
\begin{subequations}
    \begin{align}
        %\mathcal T &= \{e_1, e_2, \dots, e_k\} \\ 
        \Phi_e &= \sum_{v} A_{e v}\phi_v  \\
        \Xi_e &= \sum_{v} A_{e v}\pi_v  \\
        Q_e &= \sum_{l} \rho_{el}q_l \\ 
        \Sigma_e &= \sum_{l} \rho_{el}\sigma_l 
    \end{align}
\end{subequations}
where $\rho$ is the matrix implicitly defined in Theorem \ref{thm:spanningtree} and the variables $\Xi_e$ and $\Sigma_e$ are the spanning tree noise variables.
Further, there exist functions $\bar{h}$, $\tilde{h}$, and $E(Q,\Phi)$ such that 
\begin{subequations}\label{eqn:niceham}
    \begin{align}
    \langle \Pi| h\rangle &= \sum_{e \in \mathcal T} \Sigma_e \bar{h}_e(Q_{e_1}, Q_{e_2}, \dots) + \sum_{e\in\mathcal T} \Xi_e \tilde{h}_e(\Phi_{e_1}, \Phi_{e_2}, \dots) \\ 
        \bar{h}_e &= \frac{\partial E}{\partial Q_e} \\ 
        \tilde{h}_e &= \frac{\partial E}{\partial \Phi_e}. 
    \end{align}
\end{subequations}
In terms of these variables, there exists a matrix $Y:\mathcal{D}(\mathcal{E}) \rightarrow \mathcal{D}(\mathcal{P})$ such that \begin{align}\label{eqn:magic}
    L &= \sum_{e \in \mathcal T} \left[\Xi_e\left(\dot  Q_e + \frac{\partial E}{\partial \Phi_e}\right) - \Sigma_e \left(\dot \Phi_e - \frac{\partial E}{\partial Q_e}\right)
     \right]  \notag \\
     &\;\;\; + \sum_{e,e' \in \mathcal T, \epsilon \in \mathcal P } \left[\Xi_e Y_{\epsilon e} R_{\epsilon}^{-1} Y_{\epsilon e'} (\dot \Phi_{e'}  + \mathrm i T \Xi_{e'})\right] + \sum_{ e \in \mathcal S} \left[\Sigma_e R_{e} (\dot Q_e + \mathrm i T \Sigma_e)\right].
\end{align}
\end{cor}
\begin{proof}
Algorithm \ref{algo:sp} implies that any loops  in $\mathcal C \cup \mathcal S$ are subsets of $\mathcal C$ alone.  Hence, we can always choose a spanning tree $\mathcal T$ so that $\mathcal S \subset \mathcal T$. 
%Combining the results above with our labeling convention, we find that 
By Proposition \ref{prop:eliminate}, after eliminating the nonsingular constrained variables,  
\begin{align}
    L &= \langle \Pi| K | \dot X \rangle + \mathrm i T \langle \Pi| K^+ | \Pi \rangle + \langle \Pi| h\rangle \notag \\
     &= \sum_{e,e' \in \mathcal T} \begin{pmatrix} \Xi_e & \Sigma_e \end{pmatrix} 
    \begin{pmatrix}
        \sum_{\epsilon \in \mathcal P} Y_{ \epsilon e} R^{-1}_{\epsilon} Y_{\epsilon e'} & \delta_{ee'} \\
        -\delta_{ee'} & R_e \mathbb I[e \in \mathcal S] \delta_{ee'}
    \end{pmatrix}
    \begin{pmatrix}
    \dot \Phi_e + \mathrm i T \Xi_e \\ 
    \dot Q_e + \mathrm i T \Sigma_e 
    \end{pmatrix}  + \sum_{e \in \mathcal T} \Xi_e  \frac{\partial E}{\partial \Phi_e}  + \Sigma_e \frac{\partial E}{\partial Q_e},
    \end{align}
    where the matrix
    \begin{equation}
        Y_{\epsilon e} = B_{l(\epsilon)e} \mathbb I [e \in \mathcal S \cup \mathcal C]
    \end{equation}
    is defined by the fact that for each edge $\epsilon$ in $\mathcal P$, there is a unique loop $l(\epsilon)$ such that $l(\epsilon) \setminus (\mathcal C \cup \mathcal S) = \{\epsilon\}$. 
\end{proof}

As we will see, (\ref{eqn:magic}) will be a powerful tool in what follows.
  As was the case in \cite{osborne2023symplectic}, such spanning tree variables may be treated as independent, yet it could be possible to \emph{further reduce} the number of degrees of freedom.  As one example, in the absence of resistors, if $E$ is independent of $\Phi_e$, then the Euler-Lagrange equations above imply that $\dot Q_e = 0$, so $Q_e$ is constant.  In this Hamiltonian setting \cite{osborne2023symplectic}, it is always possible to further remove $\Phi_e$ as a degree of freedom via the method of symplectic reduction.  In the dissipative setting, it is slightly more involved: the second line of \eqref{eqn:magic} becomes relevant as well.  Regardless, we will not stress this possibility of further reducing the number of degrees of freedom in the discussion that follows, until Section \ref{sec:noether} and \ref{sec:unclean}. Doing so is a matter of convenience.
  
  %The resolution of Noether current constraints is a matter of convenience and does not need to be performed in order to calculate physical properties. 
  %This is a point of contrast between Noether current constraints and null vector constraints. See the examples in 
  %Still, some such constraints may need to be resolved in order to produce the simplest possible Lagrangian. 
  %See the example in Section \ref{sec:noether} for some further discussion.

Recall from the discussion of the relationship between the MSR path integral and the Fokker-Planck equation that to write the latter, we need to calculate the inverse of the matrix
\begin{equation}\label{eqn:tk}
    \tilde{K} = 
    \begin{pmatrix}
        Y^\intercal P_P R^{-1} Y & \mathbb I \\
        - \mathbb I & P_S R
    \end{pmatrix}
\end{equation}
which appears in (\ref{eqn:magic}).  This is always possible:
\begin{prop}\label{thm:invert}
    The matrix $\tilde K$ defined in \eqref{eqn:tk} is invertible.
\end{prop}
\begin{proof}
Since $K$ is positive semidefinite, so too must be $\tilde{K}$. 
More pointedly, the diagonal blocks of $\tilde K$ must be positive semidefinite, and because the lower two blocks of \eqref{eqn:tk} commute, we find that
\begin{equation}
    \det  \tilde{K}  = \det(Y^\intercal P_P R^{-1} Y P_R R  + \mathbb{I} ) \geq 1.
\end{equation}
So, indeed $\tilde K$ is invertible. 
\end{proof}
Proposition \ref{thm:invert} shows that the problem posed in the discussion around (\ref{eqn:simpleK}) always has a solution. 
With a firm grasp on the nonsingular submatrix of $K$ we can, in principle, derive the correct nonlinear fluctuation-dissipation theorem for any dissipative circuit (within the definitions of this paper) of interest. Still, it would be very desirable to discern a compact and convenient formula for $\tilde K^{-1}$ in general. 
Unfortunately, the problem of inverting $\tilde K$ in general is nontrivial because the block $Y^\intercal P_P R^{-1} Y$ can be complicated. We do not expect that a general expression for the inverse of $\tilde{K}$ is a simple problem because of the following example. Consider a single loop RC circuit consisting of one resistor and one capacitor in parallel between the vertices $u$ and $v$. Then attach an arbitrary resistor network across the vertices $u$ and $v$. The problem of inverting $\tilde{K}$ amounts to the simplification of the arbitrary resistor network which is a problem that is known to be difficult \cite{MCMT}. 
In the case of linear resistors, this problem will not prevent us from deriving the consequences of the fluctuation-dissipation theorem in Section \ref{sec:johnson}, but this problem is an obstruction in deriving a result of such generality in the case of nonlinear resistors in Section \ref{sec:nonlinear}, without further restrictions on the circuit.

\subsubsection{Time-reversal symmetry}

Now, we show by direct calculation that the dissipative time--reversal transformation (Definition \ref{def:TRS2}) leaves the dissipative circuit Lagrangian invariant provided that 
all inductive and capacitive energies are time--reversal invariant. 
The relevant time--reversal transformation $\mathfrak{T}$ is given by 
\begin{subequations}
        \begin{align}
        t &\rightarrow -t, \\ 
        Q_e &\rightarrow Q_e, \\ 
        \Phi_e &\rightarrow -\Phi_e, \\ 
        \Xi_e &\rightarrow 
        -\Xi_e + \frac{\mathrm i }{T} \dot \Phi_e, \\ 
        \Sigma_e &\rightarrow 
        \Sigma_e - \frac{\mathrm i}{T} \dot Q_e.
    \end{align}
\end{subequations}
This transformation acts on the Lagrangian (\ref{eqn:magic}) as 
\begin{equation}\label{eqn:nasty}
    \begin{aligned}
        L \rightarrow L' = \mathfrak{T}\, L  &
     = \sum_{e \in \mathcal T} \left[\left( 
        -\Xi_e + \frac{\mathrm i }{T} \dot \Phi_e
     \right)
     \left(-\dot  Q_e + \mathfrak{T} \frac{\partial E}{\partial \Phi_e}\right) - \left(
        \Sigma_e - \frac{\mathrm i}{T} \dot Q_e
     \right)
     \left(\dot \Phi_e - \mathfrak{T}\frac{\partial E}{\partial Q_e}\right)
     \right]  \\
     &+ \sum_{e,e' \in \mathcal T, \epsilon \in \mathcal P } \left[
     \left(
        -\Xi_e + \frac{\mathrm i }{T} \dot \Phi_e
     \right)
     Y_{\epsilon e} R_{\epsilon}^{-1} Y_{\epsilon e'} \left(\dot \Phi_{e'}  + \mathrm i T 
     \left(
        -\Xi_{e'} + \frac{\mathrm i }{T} \dot \Phi_e
     \right)
     \right)\right] \\
     &+ \sum_{ e \in \mathcal S} \left[
     \left(
        \Sigma_e - \frac{\mathrm i}{T} \dot Q_e
     \right)
     R_{e} \left(-\dot Q_e + \mathrm i T 
     \left(
        \Sigma_e - \frac{\mathrm i}{T} \dot Q_e
     \right)
     \right)\right] \\ 
     &= \sum_{e \in \mathcal T} \left[
        \Xi_e\left( \dot Q_e - \mathfrak{T} \frac{\partial E}{\partial \Phi_e} \right) - \Sigma_e \left( \dot \Phi_e - \mathfrak{T} \frac{\partial E}{\partial Q_e} \right)
     \right] \\
     &+ \sum_{e,e' \in  \mathcal T , \epsilon \in \mathcal P} \left[
       \left( \mathrm i T  \Xi_e + \dot \Phi_e \right) Y_{\epsilon e} R_{\epsilon}^{-1} Y_{\epsilon e'} \Xi_{e'}  
     \right] + \sum_{e\in \mathcal S} \left[
        \left(\mathrm i T \Sigma_e + \dot Q_e\right) R_e \Sigma_e
     \right] \\ 
     &+ \sum_{e \in \mathcal T} \left[
       \frac{\mathrm i }{T} \dot \Phi_e \left(- \dot Q_e + \mathfrak T \frac{\partial E}{\partial \Phi_e}\right) +  \frac{\mathrm i }{T} \dot Q_e \left(\dot \Phi_e - \mathfrak T \frac{\partial E}{\partial Q_e} \right)
     \right].
    \end{aligned}
\end{equation}
The last  term in (\ref{eqn:nasty}) is a total time derivative of exactly the form required by Definiton \ref{def:TRS2}.
If
\begin{equation}\label{eqn:eeven}
    \mathfrak T E = E,
\end{equation}
then we simply recover \eqref{eq:TRS2}. Note that (\ref{eqn:eeven}) holds if and only if $E$ is invariant under the transformation that takes each $\Phi_e$ to $-\Phi_e$ at once. 
We remark that many circuits of practical interest are time--reversal symmetric. Any circuit made of linear inductors, linear capacitors, Josephson junctions, quantum phase slips, and resistors is time--reversal symmetric. 

\subsection{Johnson noise for linear resistors}\label{sec:johnson}
Let us now show that the MSR path integral reproduces the standard formulas for Johnson noise through linear resistors \cite{Twiss}; nonlinear resistors are discussed in Section \ref{sec:nonlinear}.
For each edge $e \in \mathcal R$, define a variable $\zeta_e$, and
\begin{equation}
    V = \sum_{e \in \mathcal S,l} \sqrt{R_e} \sigma_l B_{le} \zeta_e + \sum_{e \in \mathcal P, v} \frac{1}{\sqrt{R_e}}\zeta_e A_{ev} \pi_v + \sum_{e\in\mathcal R} \frac{\mathrm i}{4 T} \zeta_e^2.
\end{equation}
The Lagrangian 
\begin{equation}
    L = \langle \Pi | K |\dot X\rangle + \langle \Pi| h\rangle  + V
    %+ \langle \Pi| J |\zeta\rangle + \frac{\mathrm i }{4 T} \langle \zeta | \zeta\rangle 
\end{equation}
has a number of useful properties that follow from the definitions above. 
The first property to note is that, since $\zeta_e$ appears at only quadratic order, it can be integrated out.
The $\zeta$ equation of motion is 
\begin{equation}
    0 = \frac{\delta S}{\delta \zeta_e} = 
    \begin{cases}
        \frac{\mathrm i }{2 T} \zeta_e + \sqrt{R_e}\sum_l \sigma_l B_{le} & e \in \mathcal S \\ 
        \frac{\mathrm i }{2 T} \zeta_e + \frac{1}{\sqrt{R_e}}\sum_v A_{ev}\phi_v & e \in \mathcal P 
    \end{cases}
    %0= \frac{\mathrm i }{2 T}\zeta_i + \langle \Pi| J | i\rangle.
\end{equation}
After integrating out all $\zeta$ variables, we find \eqref{eqn:msr}.
%Because of (\ref{eqn:jjcal}), it is then clear that $L$ encodes an dissipative circuit Lagrangian. 

However, if we refrain from integrating out the $\zeta$ variables, it is then possible to derive Langevin equations in a more transparent form. Define
\begin{subequations}
   \begin{align}
       \mathcal D \Pi &= \prod_{v \in \mathcal V} \mathcal D \pi_{v_i} \prod_{l \in \mathcal L} \mathcal D \sigma_l \\ 
       \mathcal D X &= \prod_{v\in \mathcal V} \mathcal D \phi_v \prod_{l \in \mathcal L} \mathcal D q_l  \\ 
       \mathcal D \zeta &= \prod_{e \in \mathcal R} \mathcal D \zeta_e  \\ 
       \delta\left(\frac{\partial L}{\partial \Pi}\right) &= \prod_{v \in \mathcal V} \prod_{l \in \mathcal L} \delta\left(\frac{\partial L}{\partial \pi_v}\right)\delta\left(\frac{\partial L}{\partial \sigma_l}\right).
   \end{align} 
\end{subequations}
With these definitions in hand, we can use the standard $\delta$ function identity to evaluate the integral with respect to $\Pi$ variables, yielding a transition probability (as in \eqref{eq:MSR1st}):
\begin{equation}\label{eqn:partition}
    P(X(t)|X(0)) = \int \mathcal D \Pi \mathcal D X \mathcal D \zeta  \text{exp}\left[\mathrm i \int \mathrm d t L\right] = \int \mathcal D X \mathcal D \zeta \delta\left(\frac{\partial L}{\partial \Pi}\right) e^{-\frac{1}{4 T}\sum_{e \in \mathcal R}\int \zeta_e^2}.
\end{equation}
The manipulation above implies that we interpret the resulting stochastic equations in the Ito formalism \cite{Huang:2023eyz}. The system of stochastic differential equations is determined by the argument of the $\delta$ function:
\begin{equation}\label{eqn:stocheoms}
    \begin{aligned}
      0 =  K | \dot X\rangle +|h\rangle +  \sum_l \frac{\partial V }{\partial \sigma_l} |l\rangle + \sum_v \frac{\partial V}{\partial \pi_v} |v\rangle.
    \end{aligned}
\end{equation}

The Langevin equation (\ref{eqn:stocheoms}) is sufficient to recover a fluctuation--dissipation theorem for RLC circuits with linear resistors in general. To show this concretely, notice that
%Despite this fact, we will endeavor to produce an FPE to describe LRC circuits in general.
 (\ref{eqn:stocheoms}) contains only additive noise, since we have assumed for this section that each resistance is a constant. In this setting, 
the Langevin equation (\ref{eqn:stocheoms}) is unambiguously defined and we do not need to worry about operator ordering in the FPE. 
Rescale 
\begin{equation}
    \bar{\zeta}_e = \sqrt{\rho_e} \zeta_e
\end{equation}
with 
\begin{equation}
    \rho_e = \begin{cases}
        R_e & e \in \mathcal S \\ 
        \frac{1}{R_e} & e \in \mathcal P,
    \end{cases}
\end{equation}
such that $\bar{\zeta_e}$ has units of current if $e$ is in $\mathcal P$ and units of voltage if $e$ is in $\mathcal S$. 
Using this rescaling, we may rewrite (\ref{eqn:stocheoms}) as 
\begin{subequations}
    \begin{align}
       0 &=\frac{\delta S}{\delta \pi_v} =  A_{ev}\left[\sum_{e \in \mathcal P} \left(\frac{1}{R_e}\dot \phi_e + \bar{\zeta}_e\right) + \sum_{e \in \mathcal S \cup \mathcal C} \dot q_e  + \sum_{e \in \mathcal I} \frac{\partial E}{\partial \phi_e} \right], \\ 
       0 &= \frac{\delta S}{\delta \sigma_l} = B_{le}\left[
         \sum_{e \in \mathcal S} (R_e \dot q_e + \bar{\zeta}_e) + \sum_{e \in \mathcal I \cup \mathcal P} \dot \phi_e  + \sum_{e \in \mathcal C} \frac{\partial E}{\partial q_e}
       \right].
    \end{align}
\end{subequations}
where we have again used the notation in (\ref{eqn:convention}). 
Upon using (\ref{eqn:partition}) to calculate expectation values, we find 
\begin{equation}\label{eqn:jny}
    \langle \bar{\zeta}_e(t) \bar{\zeta}_{e'}(0)\rangle  = 2 T \rho_e \delta_{ee'} \delta(t).
\end{equation}
Thus, we can interpret (\ref{eqn:jny}) as an expression that encodes the variance of both voltage and current fluctuations across all resistors in a generic LRC circuit (with linear resistors only). 
It is possible to recover textbook formulas describing the fluctuations in voltage measured over a small range by performing a Fourier transform and integrating over some range of frequencies: the power spectrum of voltage fluctuations in a small range of frequencies $[-f_0, f_0]$ of a resistor on edge $e$ in $\mathcal S$ is given by
\begin{equation}
    S_e = \int\limits_{-f_0}^{f_0}\mathrm{d}f \left| \int\limits_{-\infty}^\infty \mathrm{d}t \; \mathrm{e}^{\mathrm{i}2\pi f_0 t} \langle \bar{\zeta}_e(t) \bar{\zeta}_{e'}(0)\rangle \right|^2 = 4 R_e T \Delta f_0.
\end{equation}

\subsection{Circuit dualities}

Circuit duality is a map defined on drawings of circuits such that the physical observables of one circuit are a relabeling of the observables of its dual. In classical Hamiltonian systems, duality is a canonical transformation.  In our dissipative formalism, circuit duality will be similarly transparent in the setting of planar graphs (nonplanar circuit dualities appear to be not fully understood \cite{fluxcharge}).  Note the following definition of a planar circuit: \begin{defn}[Planar circuits and faces]
    A circuit (graph) is \textbf{planar} if it can be embedded (drawn) on the surface of a two-dimensional sphere without any two edges intersecting (except at a vertex).  Given such an embedding, a \textbf{face} corresponds to a loop which encloses no vertices.  We define the loop set $\mathcal{L}$ for a planar graph to consist of the set of all faces, when drawn on a sphere.
\end{defn}  

We now quote the following simple result from graph theory (see e.g. \cite{fluxcharge}):

\begin{prop} \label{prop:322}
    The matrix $B$ has one left null vector, and the matrix $A$ has one right null vector.
\end{prop}

The simplest representation of circuit duality is graphical. To take the dual of a planar dissipative circuit, execute the following algorithm: 
\begin{algo}\label{algo:dual}
Suppose $G$ is a planar dissipative circuit. 
To construct $G^*$, the dual of $G$, 
   \begin{enumerate}
       \item Draw $G$ on the plane with all points at infinity identified such that each $l$ in $\mathcal L$ is on the boundary of some face of the drawing of $G$. For planar circuits, this is always possible \cite{fluxcharge}. 
       \item In every interior face and on the lone exterior face, draw a vertex. If the loop at the boundary of a face is labeled $l$, label the vertex $l^*$.
       \item Every edge $e$ is on the boundary of exactly two faces. If $e$ is in the loops $l$ and $l'$, draw an edge between $l^*$ and $(l')^*$ and label it $e^*$. 
       \item If $e$ is an inductor (capacitor), make $e^*$  a capacitor (inductor). If $e$ is a resistor with resistance $R_e$ in  $\mathcal S$ ($\mathcal P$), make $e^*$ a resistor with resistance $R_{e^*} = R_e^{-1}$ in $\mathcal P$ ($\mathcal S$). 
       \item Each face in the drawing of $G^*$ encloses exactly one vertex $v$ in $G$. Label the loop in $G^*$ as $v^*$. 
   \end{enumerate} 
\end{algo}
An example of a dual circuit constructed by using Algorithm \ref{algo:dual} is provided in Figure \ref{fig:duality}. 

    \begin{figure}
    \centering
    %\begin{subfigure}{0.49\linewidth}
    \scalebox{0.66}{
    \begin{circuitikz}[scale=3]
       \draw[thick,capc] (0,0) to[capacitor,color=capc] (2,0);
       \draw[thick,sc] (2,0) to[resistor,color=sc] (1,2); 
       \draw[thick,indc] (1,2) to[inductor,color=indc] (0,0);
       \draw[thick,indc] (0,0) to[inductor,,color=indc] (1,.7);
       \draw[thick,pc] (1,.7) to[resistor,color=pc](1,2);
       \draw[thick,capc] (2,0) to[capacitor,color=capc] (1,.7);
       \draw[thick,sc] (2,0) to[resistor,color=sc] (3,0);
       \draw[thick,indc] (3,0) to[inductor,color=indc] (3,1); 
       \draw[thick,capc] (3,1) to[capacitor,color=capc] (2.5,2); 
       \draw[thick,indc] (2.5,2) to[inductor,color=indc] (1,2);
       \draw[thick,pc] (1,2) to[resistor,color=pc] (3,1);
       \draw[thick,capc] (2,0) to[capacitor,color=capc] (3,1);
       \filldraw[black] (0,0) circle (.5pt);
       \filldraw[black] (2,0) circle (.5pt);
       \filldraw[black] (3,0) circle (.5pt);
       \filldraw[black] (1,.7) circle (.5pt);
       \filldraw[black] (1,2) circle (.5pt);
       \filldraw[black] (2.5,2) circle (.5pt);
       \filldraw[black] (3,1) circle (.5pt);
    %\draw[thin, <-] (1,0.3) node{$l_1$}  ++(-60:0.16) arc (-60:170:0.16);
    %\draw[thin, <-] (1.2,1.0) node{$l_2$}  ++(-60:0.16) arc (-60:170:0.16);
    %\draw[thin, <-] (.8,1.0) node{$l_3$}  ++(-60:0.16) arc (-60:170:0.16);
    %\draw[thin, <-] (2,.85) node{$l_4$}  ++(-60:0.16) arc (-60:170:0.16);
    %\draw[thin, <-] (2.7,.35) node{$l_5$}  ++(-60:0.16) arc (-60:170:0.16);
    %\draw[thin, <-] (2.55,1.6) node{$l_6$}  ++(-60:0.16) arc (-60:170:0.16);
    %\draw[thin,->] (.1,-.2) to (3.3,-.2) to[short,l_=$l_7$] (3.3,1.2);

       \draw[thick,gray,opacity=0.5] (1,0.3) to[inductor,color=gray,opacity=0.5] (1.4,.9); 
       \draw[thick,gray,opacity=0.5] (1,0.3) to[capacitor,color=gray,opacity=0.5] (.65,.9);
       \draw[thick,gray,opacity=0.5] (1.4,.9) to[resistor,color=gray,opacity=0.5] (0.65,.9);
       \draw[thick,gray,opacity=0.5] (1.4,.9) to[resistor,color=gray,opacity=0.5] (2.3,.85);
       \draw[thick,gray,opacity=0.5] (2.3,.85) to[inductor,color=gray,opacity=0.5] (2.9,.25);
       \draw[thick,gray,opacity=0.5] (2.3,.85) to[resistor,color=gray,opacity=0.5] (2.5,1.6);
       \draw[thick,gray,opacity=0.5] (1,0.3) to[inductor,color=gray,opacity=0.5] (.6,-.3) to (-.1,-.3) to (-.1,2.2) to (1.8,2.2) to (2,2.1) ; 
       \draw[thick,gray,opacity=0.5] (2.9,.25) to (3.2,.25) to[capacitor,color=gray,opacity=0.5] (3.2,2.15) to (2.2,2.15) to (2,2.1); 
       \draw[thick,gray,opacity=0.5] (2.9,.25) to (2.9,-.1) to (3.35, -.1) to[resistor,color=gray,opacity=0.5] (3.35, 1) to (3.35,2.2) to (2,2.2) to (2,2.1); 
       \draw[thick,gray,opacity=0.5] (2.5,1.6) to[inductor,color=gray,opacity=0.5] (2.9,2.1) to (2,2.1);
       \draw[thick,gray,opacity=0.5] (2.5,1.6) to[capacitor,color=gray,opacity=0.5] (2,2.1);
       \draw[thick,gray,opacity=0.5] (.65,.9) to[capacitor,color=gray,opacity=0.5] (.65,2.1) to (2,2.1);

       \filldraw[gray,opacity=0.5] (1.4,0.9) circle (0.5pt);
       \filldraw[gray,opacity=0.5] (1,0.3) circle (0.5pt);
       \filldraw[gray,opacity=0.5] (2.3,0.85) circle (0.5pt);
       \filldraw[gray,opacity=0.5] (2.9,.25) circle (0.5pt);
       \filldraw[gray,opacity=0.5] (2,2.1) circle (0.5pt);
       \filldraw[gray,opacity=0.5] (.65,0.9) circle (0.5pt);
       \filldraw[gray,opacity=0.5] (2.5,1.6) circle (0.5pt);
    \end{circuitikz}}
    \caption{A circuit drawn with elements colored according to set inclusion. Its dual, constructed by Algorithm \ref{algo:dual}, is drawn in grey, with labels omitted on both circuits.}
    \label{fig:duality}
\end{figure}
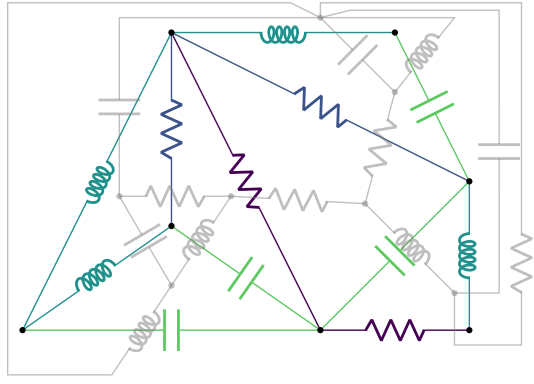

\begin{defn}[Dual circuit]\label{defn:duality}
Let $G$ be a dissipative circuit. 
%Primed quantities are those of $G'$ while unprimed symbols are those of $G$.  
    The \textbf{dual circuit} $G^*$ has properties defined via
    \begin{subequations}\label{eq:duals}
        \begin{align}
            \phi_v &\rightarrow (\phi_v)^* = q_{v^*}' ,\\ 
            q_l &\rightarrow (q_l)^* = \phi_{l^*}' , \\ 
            \pi_v &\rightarrow (\pi_v)^* = \sigma_{v^*}' ,\\ 
            \sigma_l &\rightarrow (\sigma_l)^* = \pi_{l^*}',\\
            A &\rightarrow A^* = (B')^\intercal, \\ 
            B &\rightarrow B^* = (A')^\intercal, \\ 
            \mathcal V &\rightarrow \mathcal V^* = \mathcal L', \\ 
            \mathcal E &\rightarrow \mathcal E^* = \mathcal E', \\
            \mathcal L &\rightarrow \mathcal L^* = \mathcal V',\\
            \mathcal C &\rightarrow \mathcal{C}^* = \mathcal I', \\ 
            \mathcal I &\rightarrow \mathcal{I}^* = \mathcal C', \\ 
            \mathcal S &\rightarrow \mathcal{S}^* = \mathcal P', \\ 
            \mathcal P &\rightarrow \mathcal{P}^* = \mathcal S', \\ 
            R_e &\rightarrow (R_e)^* = \frac{1}{R_e'}.
        \end{align}
    \end{subequations}
    The Lagrangian of the dual circuit is \eqref{eqn:msr} expressed in the dual coordinates.
\end{defn}
    The transformation given in Definition \ref{defn:duality} is a relabeling transformation. In this sense, it is clear that 
    $L^*$ is equivalent to $L$. Moreover, $L^*$ is a dissipative Lagrangian that describes the circuit $G^*$ constructed by means of Algorithm \ref{algo:dual}.  Hence, there is a very transparent generalization of circuit duality in our formalism.  In particular, we notice that the dual of the dual circuit is the original one:  $G=(G^*)^*$, and that the dual circuit obeys all of the nice properties of the original circuit (such as the existence of spanning tree coordinates), because: 

\begin{prop}
    Corollary \ref{cor:goodpick}, and therefore Theorem \ref{thm:spanningtree}, hold for the dual circuit $G^*$.
\end{prop}
\begin{proof}
    This result follows immediately from the identification of sets in \eqref{eq:duals}, together with (\emph{1}) Proposition \ref{prop:322} which guarantees that $A^*$ and $B^*$ are valid boundary maps, and (\emph{2}) the fact that a loop and cut are mapped on to each other under the duality transformation \cite{fluxcharge}.
\end{proof}

\subsection{Comparison to formalism based on Rayleigh dissipative function}\label{sec:rayleigh}
The Rayleigh dissipative function is a construction that may be used to incorporate dissipative terms into Lagrangian and Hamiltonian mechanics \cite{Parra-Rodriguez:2023ykw,mariantoni2024quantum,Strutt,mariantoni2021energy}. We now explain that one can ``interpret" our approach, at least for linear resistors, in terms of such a Rayleigh dissipative function.   As the discussion below does not rely on any specific details of circuit mechanics, we will discuss a general Hamiltonian dynamical system with canonical Poisson brackets; the relation between a Rayleigh system of the form below, and circuit mechanics, can be found in \cite{Parra-Rodriguez:2023ykw,mariantoni2024quantum,Strutt,mariantoni2021energy}.  

\begin{defn}[Rayleigh system]
    Consider a Lagrangian
    \begin{equation}\label{eqn:tamelag}
        L(p_\alpha,x_\alpha) = \sum_\alpha \dot x_\alpha p_\alpha - E
    \end{equation}
    with $E$ a generic function of $x_\alpha$ and $p_\alpha$.
    Given positive semidefinite matrices $M$ and $N$, define 
    \begin{equation} \label{eq:rayleighR}
        \mathfrak{R} =\sum_{\alpha,\beta}\left[ M_{\alpha \beta} \dot x_{\alpha} \dot x_{\beta} + N_{\alpha \beta} \dot p_\alpha \dot p_\beta\right].
    \end{equation} 
    The pair $(L,\mathfrak{R})$ together with the equations of motion 
    \begin{subequations}\label{eq:rayleighsystem}
    \begin{align}
        0 &= \frac{\partial \mathfrak{R}}{\partial \dot x_\alpha} + \dot p_\alpha + \frac{\partial E}{\partial x_\alpha}  \\ 
        0 &= -\frac{\partial \mathfrak{R}}{\partial \dot p_\alpha} + \dot x - \frac{\partial E}{\partial p_\alpha} 
        \end{align}
    \end{subequations} 
    is called a \textbf{Rayleigh system}.  
\end{defn}
\begin{obs}
    Let $(L,\mathfrak R)$ be a Rayleigh system with $\mathfrak{R}$ defined in \eqref{eq:rayleighR}. The dissipative Lagrangian 
    \begin{equation}
        L' = \sum_{\alpha,\beta} \begin{pmatrix}
            \pi_\alpha & \sigma_\alpha
        \end{pmatrix}
        \begin{pmatrix}
            M_{\alpha \beta} & \delta_{\alpha \beta} \\ 
            -\delta_{\alpha \beta} & N_{\alpha \beta} 
        \end{pmatrix}
        \begin{pmatrix}
            \dot x_\beta+ \mathrm i T \pi_\beta \\ 
            \dot p_\beta + \mathrm i T \sigma_\beta
        \end{pmatrix}
        + \sum_{\alpha} \left[\pi_\alpha \frac{\partial E}{\partial x_\alpha} + \sigma_\alpha \frac{\partial E }{\partial p_\alpha} \right]
    \end{equation}
    has the property that 
    \begin{subequations}
        \begin{align}
            0 &= \left. \frac{\partial L'}{\partial \pi_\alpha}\right|_{\pi = \sigma = 0} = \frac{\partial \mathfrak{R}}{\partial \dot x_\alpha} + \dot p_\alpha + \frac{\partial E}{\partial x_\alpha}  \\ 
            0  &= \left.\frac{\partial L'}{\partial \sigma_\alpha}\right|_{\pi = \sigma = 0} = -\frac{\partial \mathfrak{R}}{\partial \dot p_\alpha} + \dot x - \frac{\partial E}{\partial p_\alpha}. 
        \end{align}
    \end{subequations}
    The constraints on $M$ and $N$ in \eqref{eq:rayleighR}, together with the relative signs in front of $\mathfrak{R}$ in \eqref{eq:rayleighsystem}, are imposed in dissipative Lagrangian $L^\prime$ transparently, since $\mathrm{Im}(L^\prime) \ge 0$.
\end{obs}
In short, the framework of dissipative Lagrangians is a more physically transparent way of implementing the same methodology of a Rayleigh dissipation function in Hamiltonian mechanics. We strongly prefer the presentation we have done above, which emphasizes the deep \emph{physical} connection betewen the dissipative Lagrangian and statistical mechanics via the fluctuation-dissipation theorem. We remind the reader that due to the existence of the spanning tree coordinates, we can systematically map dissipationless circuit mechanics onto the form \eqref{eqn:tamelag}.

\section{Nonlinear resistors}\label{sec:nonlinear}
We now return to a discussion of nonlinear resistors.  A fully general discussion of nonlinear resistors is challenging without introducing additional degrees of freedom and constraints.  As we now explain, for circuits with suitable parasitic inductors or capacitors, it is possible to exhaustively analyze the dynamics.

Our first task is to make precise statements about what kinds of dissipative circuit elements may be described by our formalism. In general, resistors are described by a so--called current--voltage curve. 
For a generic resistor, there exists some function $f(\dot\phi, \dot q)$ such that the relationship between current and voltage is specified by 
\begin{equation}
    f(\dot \phi, \dot q) = 0. 
\end{equation}
However, some resistors have current--voltage relationships that are poorly behaved. 
For this reason, it will be useful for us to restrict our attention to circuits made of elements with special current--voltage relationships. 
\begin{defn}[Current fixed]
   A dissipative circuit element is called \textbf{current fixed} if there exists some function $I(\dot \phi)$ such that the current across the element satisfies
   \begin{equation}
       \dot q = I(\dot \phi). 
   \end{equation}
\end{defn}
\begin{defn}[Voltage fixed]
A dissipative circuit element is called \textbf{voltage fixed} if there exists some function $V(\dot q)$ such that the voltage across the element satisfies
\begin{equation}
    \dot \phi = V(\dot q). 
\end{equation}
\end{defn}
Of course, a resistor may be both voltage fixed and current fixed. An example of some such circuit element is the linear resistor, which admits 
\begin{subequations}
    \begin{align}
        I(\dot \phi) &= \frac{1}{R} \dot \phi \\ 
        V(\dot q) &= R \dot q.
    \end{align}
\end{subequations}
\begin{defn}[Doubly fixed]
    A dissipative circuit element that is both voltage fixed and current fixed is called \textbf{doubly fixed}.
\end{defn}
For a dissipative circuit element that is either voltage fixed or current fixed, it is possible to define a function $R$ such that a version of Ohm's law is made to hold: 
\begin{equation}
    R := \begin{cases}
        \dfrac{\dot \phi}{I(\dot\phi)} & e \text{ is current fixed} \\
        \dfrac{V(\dot q)}{\dot q} & e \text{ is voltage fixed}.
    \end{cases}
\end{equation}
When we refer to nonlinear resistances, this is tentatively the quantity to which we refer, although we will ultimately see some subtleties related to the fluctuation-dissipation theorem.  For circuit elements which are doubly fixed,  
one is free to choose to represent $R$ either as a pure function of $\dot \phi$ or as a pure function of $\dot q$. 

We aim to identify a restricted class of circuit where the matrix $\tilde{K}$ from (\ref{eqn:tk}) is easy to invert. 
In order to proceed, we need to make some additional restrictions on the circuits we study.  First, we introduce a few more relevant definitions.
\begin{defn}[Parallel/series circuit elements]
    We say that circuit elements in a circuit $G$ on edges $e$ and $e'$ are in parallel (series) if every cut (loop) containing $e$ also contains $e'$ and vice versa. 
\end{defn}
%This definition of series and parallel elements is equivalent to other common definitions.  
Note that this definition precludes us from talking about series/parallel lumped elements, but for the discussion that follows this restriction is unimportant.  With this definition in mind, we observe that:
\begin{prop}
    In a circuit with more than two edges, if $e$ and $e'$ are in parallel, then neither edge is in series with any edge.
    Likewise, if $e$ and $e'$ are in series, then neither edge is in parallel with any edge.
\end{prop}
\begin{proof}
    Suppose $e = (u,v)$ and $e'$ are in parallel. Then $\langle e'| A|v\rangle$  and $\langle e'|A|u\rangle$ must be nonzero lest there be some cut involving $e$ but not $e'$. Thus, either $e' = (u,v)$ or $e'=(v,u)$. It follows that the only edge that can be in series with $e$ is $e'$, and only if $u$ and $v$ are the only vertices in $G$. 
    The remainder of the theorem follows analogously. 
\end{proof}
\begin{cor}\label{cor:notboth}
    If $e$ and $e'$ are edges in series (parallel) then there is a cut (loop) involving only $e$ and $e'$. 
\end{cor}
\begin{defn}
    A dissipative circuit is called a \textbf{clean circuit} if:
    \begin{enumerate}
        \item Every resistor is current fixed, voltage fixed, or doubly fixed. 
        \item Every resistor that is current fixed but not voltage fixed is in  parallel with a linear capacitor.
        \item Every resistor that is voltage fixed but not current fixed is in series with a linear inductor.
        \item Every resistor that is doubly fixed is in parallel with a linear capacitor or in series with a linear inductor.  By Corollary \ref{cor:notboth} only one of these possibilities can be realized.
    \end{enumerate}
\end{defn}

From a physical perspective, we regard our restriction to clean circuits as mild in the sense that it is common practice to incorporate parasitic inductances and capacitances for resistors, such that they have an effective frequency-dependent impedance: see e.g. \cite{Parra-Rodriguez:2023ykw,vool_introduction_2017,nathan2024selfcorrecting,Abdo_2013,minev2021energyparticipation,Sorin_2021}.
We also remark that our requirement that parasitic elements be linear is not wholly a matter of convenience. 
If parasitic capacitances are allowed to be nonlinear, singularity may arise if $\partial E/\partial q$ or $\partial E/\partial \phi$ is not invertible.  But for the sake of simplicity in the notation (namely to make this inverse simple), we just stick to parasitic linear elements.

Let $G$ be a clean circuit with $\mathcal S$ and $\mathcal P$ chosen according to Algorithm \ref{algo:sp}.
The resistive edges in $\mathcal S$ are all in series with an inductor, and the resistive edges in $\mathcal P$ are all in parallel with a capacitor. 
Furthermore, all edges in $\mathcal S$ are voltage fixed, while all edges in $\mathcal P$ are current fixed. 
We may build the dissipative Lagrangian for $G$ according to 
\begin{equation}
    L = \langle \Pi| K | \dot X\rangle + \mathrm i T \langle \Pi| K | \Pi \rangle + T \langle \Pi | \mu\rangle.
\end{equation}
For every edge $e$ in $\mathcal S$, there is some vertex $v$ such that 
\begin{equation}\label{eqn:vertexdet}
    0 = \frac{\delta S}{\delta \pi_v} = \dot q_e + T h_v. 
\end{equation}
Moreover, since $e \in \mathcal S$, it is voltage fixed, and thus has a resistance $R_e(\dot q)$. Thus, it is possible to use (\ref{eqn:vertexdet}) to write $R_e(-T h_v)$. In other words, the dependence upon the current of the resistance on edge $e$ may be expressed in terms on only $\phi$ variables, rather than the naively expected dependence on $\dot q$. 
Furthermore, in any valid choice of spanning tree variables, $h_v$ is expressible as a function of only spanning tree  variables. 
%In this way, we see that the resistance of an edge in $\mathcal S$ is a function of only $\Phi$ variables that correspond to the ``good"
%fluxes in some choice of spanning tree.
A similar argument may be used to show that the resistance of an edge in $\mathcal P$ may be expressed in terms of only charge spanning tree variables. 

Choose a spanning tree $\mathcal T$ in $\mathcal S \cup \mathcal C$.  
Every edge $e$ in $\mathcal P$ is in parallel with some edge in $\mathcal C$. We will use a $\hat{\,}$ symbol to denote the fact that 
$\hat{e} \in \mathcal C$ is in parallel with $e \in \mathcal P$.
This notation is helpful because spanning tree variables are indexed by edges in $\mathcal C \cup \mathcal S$, but we need a compact way to express a sum that runs over $\mathcal P$.
With this notation and our chosen spanning tree in mind, the dissipative Lagrangian for a clean circuit may be expressed as
\begin{equation}\label{eqn:obsk}
    L = \sum_{e \in\mathcal T} \left[\Xi_e \left(\dot Q_e + \frac{\partial E}{\partial \Phi_e}\right)  - \Sigma_e\left( \dot \Phi_e- \frac{\partial E}{\partial Q_e}\right) \right]  + \sum_{e \in \mathcal P} \frac{1}{R_e(Q)} \Xi_{\hat{e}} (\dot \Phi_{\hat{e}} + 2\mathrm i T \Xi_{\hat{e}}) + \sum_{e \in \mathcal S} \Sigma_e R_e(\Phi) (\dot Q_e + 2\mathrm i T \Sigma_e).
\end{equation}
With the definitions 
\begin{subequations}
\begin{align}
    |Q\rangle &= \sum_{e\in\mathcal T} |e \rangle Q_e \\ 
    |\Phi \rangle &= \sum_{e\in\mathcal T} |e \rangle \Phi_e \\ 
    |\Xi\rangle &= \sum_{e\in\mathcal T} |e\rangle \Xi_e \\ 
    |\Sigma \rangle &= \sum_{e\in\mathcal T} |e \rangle \Sigma_e \\ 
    %|\bar{\mu}\rangle = \sum_i |i \rangle \bar{\mu}_i \\ 
    %|\tilde{\mu} \rangle = \sum_i |i \rangle \tilde{\mu}_i \\ 
    M_1 &= \sum_{e\in\mathcal P} \frac{1}{R_{e}}|\hat{e} \rangle \langle \hat{e} | \\ 
    M_2 &= \sum_{e\in\mathcal S} R_{e} |e \rangle \langle e |, 
\end{align}
\end{subequations}
we may rewrite (\ref{eqn:obsk}) as a matrix equation in the form 
\begin{equation}
    L = \begin{pmatrix}
        \langle \Xi| & \langle \Sigma| 
    \end{pmatrix}
    \begin{pmatrix}
        M_1 & - \mathbb I \\
        \mathbb I & M_2
    \end{pmatrix}
    \begin{pmatrix}
        |\dot \Phi \rangle + \mathrm i T |\Xi\rangle \\ 
        |\dot Q \rangle + \mathrm i T |\Sigma \rangle
    \end{pmatrix}
    + \sum_i \Xi_i \frac{\partial E}{\partial \Phi_i} + \Sigma_i \frac{\partial E}{\partial Q_i}
\end{equation}
Since we have demonstrated that the matrix 
\begin{equation}
    \tilde{K} = 
    \begin{pmatrix}
        M_1 & -\mathbb I \\ 
        \mathbb I & M_2
    \end{pmatrix}
\end{equation}
is generally invertible, we may take advantage of the fact that $M_1 M_2 = 0$, to write 
\begin{equation}
    \tilde{K}^{-1} = 
    \begin{pmatrix}
        M_2 & \mathbb I \\ 
        - \mathbb I & M_1 
    \end{pmatrix}.
\end{equation}
Now, by a linear transformation of noise variables by $\tilde{K}^{-1}$, 
\begin{equation}\label{eqn:nonlinearL}
\begin{aligned}
    L &= 
    \begin{pmatrix}
        \langle \Xi '| & \langle \Sigma '| 
    \end{pmatrix}
    \begin{pmatrix}
        |\dot \Phi \rangle \\ 
        |\dot Q \rangle 
    \end{pmatrix}
    + \mathrm i T \begin{pmatrix}
        \langle \Xi ' | & \langle \Sigma' | 
    \end{pmatrix}
    \begin{pmatrix}
        M_2 & \mathbb I \\ 
        - \mathbb I & M_1 
    \end{pmatrix}
    \begin{pmatrix}
        |\Xi'\rangle  - \mathrm i|\tilde{h}\rangle  \\ 
        |\Sigma'\rangle - \mathrm i  |\bar{h}\rangle 
    \end{pmatrix} \\ 
    %%&= \sum_{i = 1}^n (\Xi'_i \dot \Phi_i + \Sigma'_i \dot Q_i) + \mathrm i T \left( \Xi'_i \frac{1}{R_{e_i}(Q)} (\Xi'_i - \mathrm i \tilde{\mu}_i(\Phi) ) \right) + \mathrm i T \left( \Sigma'_i R_{e_i}(\Phi) (\Sigma_i' - \mathrm i \bar{\mu}_i(Q)) 
    %%\right)
    \end{aligned}
\end{equation}
By inspection, (\ref{eqn:nonlinearL}) is of the standard MSR form (\ref{eq:MSRquadratic}).
From this point, a FPE may be derived from (\ref{eqn:nonlinearL}) according to the prescription given in Sec \ref{sec:background}. 
For $e \in \mathcal S$, the argument of $R_e$ is some linear combination of fluxes divided by an inductance. 

It is not generally the case that $R_e$ is a function of only $\Phi_e$. However, $R_e$ must depend upon $\Phi_e$ since $\Phi_e$ is a spanning tree variable, and the inductor in series with $e$ must only depend upon differences of flux across nodes in spanning tree. Likewise, for $e \in \mathcal P$, the argument of $R_e$ is a linear combination of spanning tree charges divided by a capacitance. 
%\andy{terminology a bit informal here...maybe also just remark that we will see this linear combination issue in section 5.1.} 
An example where this linear combination can be nontrivial (i.e. more than one spanning tree variable in the argument) is given in Section \ref{sec:noether}.
Still, if $R_{e_3}$ is a voltage fixed nonlinear resistor, then $R_{e_3}$ may be expressed as a function only of $h_{e_2}$.
\begin{align}\label{eqn:fpe}
    \partial_t P = \sum_{e \in \mathcal T} \left(\frac{\partial P}{\partial \Phi_e} \frac{\partial E}{\partial Q_e} - \frac{\partial P}{\partial Q_e}\frac{\partial E}{\partial \Phi_e}\right) &+ \sum_{e \in \mathcal S} \frac{\partial }{\partial \Phi_e} \left[ R_e(\tilde{h}_e) \left(T \frac{\partial P}{\partial \Phi_e} + \frac{\partial E}{\partial \Phi_e}P \right)\right] \notag \\
    &+ \sum_{e \in \mathcal P} \frac{\partial }{\partial Q_{\hat{e}}} \left[\frac{1}{R_{e}(\bar{h}_{\hat{e}})} \left( T \frac{\partial P}{\partial Q_{\hat{e}}}  + \frac{\partial E}{\partial Q_{\hat{e}}}P\right)\right]
    \end{align}
%We remark that in above and below expressions, derivatives without an explicit target act upon all objects to their right. That is to say that $\frac{\partial E}{\partial \Phi_i} P$ denotes a derivative of only $E$ with respect to $\Phi_i$ while $\frac{\partial}{\partial \Phi_i} E P$ denotes a derivative of the function $E \cdot P$ with respect to $\Phi_i$. 
%As another example, $\frac{\partial}{\partial \Phi_i}  \left(f + g\right)P$ can be written $\frac{\partial (f P) }{\partial \Phi_i } + \frac{\partial (gP)}{\partial \Phi_i}$.
In order to read off noise correlations from this FPE, it is necessary to rearrange (\ref{eqn:fpe}) in the form
\begin{align}
    \partial_t P = \sum_{e \in \mathcal T} \left(\frac{\partial P}{\partial \Phi_e} \frac{\partial E}{\partial Q_e} - \frac{\partial P}{\partial Q_e}\frac{\partial E}{\partial \Phi_e}\right) &+ \sum_{e\in\mathcal S}\left(\frac{\partial}{\partial \Phi_e}\left[ \left( R_{e} \frac{\partial E}{\partial \Phi_e} - T\frac{\partial R_{e}}{\partial \Phi_e} \right)P\right] + T\frac{\partial^2}{\partial \Phi_e^2} ( R_{e} P )\right) \notag \\ 
    &+ \sum_{e \in \mathcal P}\left(\frac{\partial}{\partial Q_{\hat{e}}} \left[\left(R_{e}^{-1} \frac{\partial E}{\partial Q_{\hat{e}}}  - T \frac{\partial R_{e}^{-1}}{\partial Q_{\hat{e}}} \right)P\right] + T \frac{\partial^2 }{\partial Q_{\hat{e}}^2 }(R_{e}^{-1} P)\right). \label{eq:FPElast}
    \end{align}
With the FPE in the form \eqref{eq:FPElast}, we can read off the experimentally observable form of Ohm's Law.  For example, let us consider a resistor in $\mathcal{S}$, such that the current through the resistor is fixed by its series linear inductor.  From \eqref{eq:FPElast} we can calculate 
\begin{equation}\label{eqn:ohmseff}
    \partial_t \langle \Phi_e\rangle = \int \left(\prod_e \mathrm{d}\Phi_e\mathrm{d}Q_e\right) \; \Phi_e \left[\frac{\partial}{\partial \Phi_e}\left( \left(\frac{\partial E}{\partial Q_e} + R_e \frac{\partial E}{\partial \Phi_e} - T\frac{\partial R_e}{\partial \Phi_e}\right)P \right) + \cdots\right] 
    =  -\left\langle  R_e\frac{\partial E}{\partial \Phi_e} - T \frac{\partial R_e}{\partial \Phi_e}\right\rangle
\end{equation}where $\cdots$ denote terms in the FPE which will not affect the answer because their derivative motif annihilates $\Phi_e$ upon integration by parts over all coordinates, and where $\langle \cdots \rangle$ denotes the average over the noise (e.g. averaging over $P$). 
To see that (\ref{eqn:ohmseff}) is an effective Ohm's law, note that the current across $e$ is $-\frac{\partial E}{\partial \Phi_e}$. 
Rearranging terms for clarity, 
\begin{equation}
    \partial_t \langle \Phi_e\rangle = \left\langle \left(R_e  - T\left(\frac{\partial E}{\partial \Phi_e}\right)^{-1} \frac{\partial R_e}{\partial \Phi_e}\right) \left(-\frac{\partial E}{\partial \Phi_e} \right)\right\rangle.
\end{equation}
In the case that $R_e$ is constant, this reduces to 
\begin{equation}
    \partial_t \langle \Phi_e\rangle = R_e \left\langle - \frac{\partial E}{\partial \Phi_e}\right\rangle
\end{equation}
which is the textbook form of Ohm's law. 

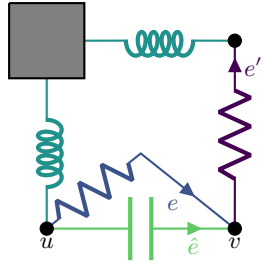
\begin{figure}
    \centering
    \begin{circuitikz}[scale=2.5]
        \draw[thick,capc] (0,0) to[capacitor,color=capc,i_=$\hat{e}$] (1,0); 
        \draw[thick,pc] (0,0) to[resistor,color=pc] (0.5,0.4) to[short,i_=$e$,color=pc] (1,0);
        \draw[thick,indc] (0,.8) to[inductor,color=indc] (0,0);
        \draw[thick,sc] (1,0) to[resistor,color=sc,i_=$e'$] (1,1); 
        \draw[thick,indc] (1,1) to[inductor,color=indc] (.2,1);
        \draw[thick,fill=gray] (-.2,1.2) to (.2,1.2) to (.2,.8) to (-.2,.8) to (-.2,1.2);
        \filldraw[black] (0,0) circle (1pt) node[anchor=north]{$u$}; 
        \filldraw[black] (1,0) circle (1pt) node[anchor=north]{$v$}; 
        \filldraw[black] (1,1) circle (1pt); 
    \end{circuitikz}
    \caption{A part of a circuit which illustrates the subtleties with extracting the proper form of Ohm's law for nonlinear resistors.  The box shaded in gray is meant to denote some arbitrary configuration of circuit elements. Elements explicitly drawn are colored according to set inclusion. }
    \label{fig:pdiff}
\end{figure}
On the other hand, for edges $e$ in $\mathcal P$, the Ohm's law takes on a superficially different form due to the asymmetrical way that our spanning tree construction treats edges in $\mathcal S$ and $\mathcal P$. 
Explicitly, 
\begin{equation}\label{eqn:pohm}
    I_e = \left \langle
    \left( R_{e}^{-1} - T\left(\frac{\partial E}{\partial Q_{\hat{e}}}\right)^{-1} \frac{\partial R_e}{\partial Q_{\hat{e}}} \right)  \frac{\partial E}{\partial Q_{\hat{e}}} \right\rangle 
\end{equation}
where (\ref{eqn:pohm}) has been derived in exactly the same manner as (\ref{eqn:ohmseff}). 
The left hand side of (\ref{eqn:pohm}) is given by
\begin{equation}
    I_e = \left \langle \frac{\partial E}{\partial \Phi_{\hat{e}}} \right\rangle - \partial_t \langle Q_{\hat{e}} \rangle 
\end{equation}
and may be interpreted as the current through edge $e$.
To understand the difference between (\ref{eqn:pohm}) and (\ref{eqn:ohmseff}), consider the colored subcircuit drawn in Figure \ref{fig:pdiff}. 
Since $e$ is an edge in $\mathcal P$, and we have restricted our attention to clean circuits, $\hat{e}$ is in our chosen spanning tree, while $e$ is not. 
Further, since there is a linear capacitor on $\hat{e}$, 
\begin{equation}
    \frac{\partial E}{\partial Q_{\hat{e}}} = \frac{Q_{\hat{e}}}{C}, 
\end{equation}
which is the voltage across $\hat{e}$. Since $e$ and $\hat{e}$ are in parallel, $\frac{\partial E}{\partial Q_{\hat{e}}}$ is also the voltage across $e$. 
Hence, the quantity $\frac{\partial E}{\partial \Phi_{\hat{e}}}$ may be interpreted as the total current flowing from $u$ to $v$. 
As a point of contrast, since $e'$ is in $\mathcal S$, $Q_{e'}$ is a spanning tree variable, but no terms in $E$ may depend upon $Q_{e'}$ since $e'$ is a resistor. Thus, 
\begin{equation}
    \frac{\partial E}{\partial Q_{e'}} = 0.
\end{equation}

Lastly, let us discuss the implications of the positivity condition on $R_e$, again focusing for simplicity on resistors $e\in\mathcal{S}$.  It is clear that the positivity condition which ensures that systems are dissipative is that \begin{equation}
    R_e \ge 0,
\end{equation}
so that the FPE \eqref{eq:FPElast} is well-defined and has positive noise variance.  In contrast, we can use the chain rule to write the effective $I-V$ curve measured for the resistor is 
%\andy{I guess this is not in general $\Phi_e$ but might be some more complicated linear combination...} \AO{fixed some minus signs}
\begin{equation}
     \dot\Phi_e = -\frac{\partial E}{\partial \Phi_e} R_e + T\frac{\partial R_e}{\partial \Phi_e} = \dot Q_e \left(R_e - T\frac{\partial R_e}{\partial E}\right).
\end{equation}
To see that the second equality holds, note that 
\begin{equation}
    \frac{\partial E}{\partial \Phi_e} = h_e = \frac{1}{L_e}\sum_{e'} \Phi_{e'}
\end{equation}
and $R_e$ is a function of only $h_e$. 
Further, 
\begin{equation}
    \frac{1}{h_e} \frac{\partial}{\partial \Phi_e} R_e(h_e) = \left.\frac{\partial}{\partial E} R_e(\sqrt{2 E_e/ L_e}) \right|_{E_e = L_e h_e^2/2}.
\end{equation}
%such that 
The ``ohmic resistance" (i.e. the one that appears in an experimentally measured Ohm's law) is given by 
\begin{equation}
    \tilde R_e = \left.\left(R_e - T \frac{\partial R_e}{\partial E_e}\right)\right|_{E_e = L_e h_e^2 /2 }.
\end{equation}
Notice that while $R_e \ge 0$ is required by positivity of fluctuations, in contrast we do \emph{not} require that $\tilde R_e \ge 0$. 
%We note that this observation postulates that the energy stored in the circuit is wholly captured by the (electromagnetic) energy stored in the capacitive and inductive elements, as given in \eqref{}; it may be possible that if resistors that violate the constraints above are realized in experiment so long as either the circuit does not thermalize to a constant temperature, or the effective energy stored in the circuit deviates from the simple formula \eqref{}.

\section{Examples}
We now give a few simple examples of how to apply our formalism to small circuits, to provide concrete illustrations of the abstract ideas in the sections above.
\subsection{An RLC circuit}\label{sec:noether}
Consider the circuit drawn in Figure \ref{fig:rlc}. 
We will begin by constructing a Lagrangian and producing its equations of motion by direct simplification and identification of unnecessary degrees of freedom.
Our convention for choosing $\mathcal S$ and $\mathcal P$ dictates that $\mathcal S = \{e_3\}$ and $\mathcal P = \emptyset$. 
From this choice it is straightforward to produce 
\begin{equation}
    K = 
    \begin{pmatrix}
        0 & 0 & 0 & 0 & 0 \\
        0 & 0 & 0 & 1 & -1 \\ 
        0 & 0 & 0 & -1 & 1 \\ 
        0 & -1 & 1 & R & - R \\
        0 & 1 & -1 & -R &  R 
    \end{pmatrix}
\end{equation}
and thus
\begin{equation}\label{eqn:loopnodeex}
\begin{aligned}
    L &= (\pi_{v_2} - \pi_{v_3})(\dot q_{l_1} - \dot q_{l_2}) - (\dot\phi_{v_2} - \dot\phi_{v_3})(\sigma_{l_1} - \sigma_{l_2}) + R (\sigma_{l_1} - \sigma_{l_2})(\dot q_{l_1} - \dot q_{l_2} + \mathrm i T( \sigma_{l_1} - \sigma_{l_2})) \\ 
    &+ \frac{1}{L}(\pi_{v_3} - \pi_{v_2})(\phi_{v_3} - \phi_{v_2}) + \frac{1}{C}(\sigma_{l_1} - \sigma_{l_2})(q_{l_1} - q_{l_2}) \\
    &= (\pi_{v_2} - \pi_{v_3})\left[\dot q_{l_1} - \dot q_{l_2} + \frac{1}{L} (\phi_{v_2} - \phi_{v_3})\right] + (\sigma_{l_1} - \sigma_{l_2})\left[ -(\dot\phi_{v_2} - \dot\phi_{v_3}) + R (\dot q_{l_1} - \dot q_{l_2} ) + \mathrm i  T R(\sigma_{l_1} - \sigma_{l_2})  + \frac{1}{C} (q_{l_1} - q_{l_2})\right] \\ 
    &= \Xi \left(\dot Q + \frac{1}{L} \Phi\right)  - \Sigma \left(\dot \Phi  - \frac{1}{C} Q\right) + \Sigma R ( \dot Q + \mathrm i T \Sigma )
    \end{aligned}
\end{equation}
with the definitions
\begin{subequations}
    \begin{align}
        Q &= q_{l_1} - q_{l_2} \\ 
        \Phi &= \phi_{v_2} - \phi_{v_3} \\ 
        \Xi &= \pi_{v_2} - \pi_{v_3} \\ 
        \Sigma &= \sigma_{l_1} - \sigma_{l_2}.
    \end{align}
\end{subequations}
\begin{figure}
    \centering
    \begin{circuitikz}[scale=2]
        \draw[thick,capc] (-1,0) to[capacitor,i_=$e_1$,color=capc] (1,0);
        \draw[thick,indc] (1,0) to[inductor,i_=$e_2$,color=indc] (0,1);
        \draw[thick,sc] (0,1)  to[resistor,i_=$e_3$,color=sc] (-1,0);
        \filldraw[black] (-1,0) circle (1pt) node[anchor = east]{$v_1$};
        \filldraw[black] (1,0) circle (1pt) node[anchor = west]{$v_2$};
        \filldraw[black] (0,1) circle (1pt) node[anchor = south]{$v_3$};
    \end{circuitikz}
    \caption{A minimal one--loop RLC circuit. Edges are colored according to set inclusion as in Fig. \ref{fig:example}.  }
    \label{fig:rlc}
\end{figure}
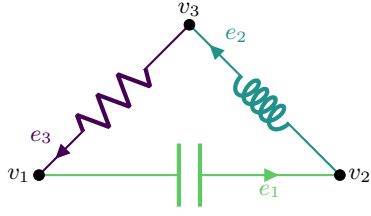

On the other hand, if we use spanning tree variables from the outset, we find 
\begin{equation}\label{eqn:spanningtreeex}
    \begin{aligned}
        %E &= \frac{1}{2 L} (\Phi_{e_1} + \Phi_{e_3})^2 + \frac{1}{2 C} Q_{e_1}^2 \\ 
        L' &= \Xi_{e_1} \left(\dot Q_{e_1} + \frac{1}{ L} (\Phi_{e_1} + \Phi_{e_3}) \right)  + \Xi_{e_3}\left(\dot Q_{e_3} + \frac{1}{ L}(\Phi_{e_1} + \Phi_{e_3}) \right) - \Sigma_{e_1} \left(\dot \Phi_{e_1} - \frac{1}{C} Q_{e_1} \right) - \Sigma_{e_3} \dot \Phi_{e_3} \\ 
        &+ \Sigma_{e_3} R (\dot Q_{e_3} +  \mathrm i T \Sigma_{e_3}).
    \end{aligned}
\end{equation}
While (\ref{eqn:loopnodeex}) and (\ref{eqn:spanningtreeex}) encode the same physics, there is a notable difference between them. 
In (\ref{eqn:spanningtreeex}), $Q_{e_1}$ and $Q_{e_3}$ are apparently independent, but physically they must be the same since the circuit in Figure \ref{fig:rlc} has only one loop. 
In (\ref{eqn:spanningtreeex}), this fact is encoded by the fact that 
\begin{equation}
    0 = \frac{\delta S}{\delta \Xi_{e_1}} -\frac{\delta S}{\delta \Xi_{e_3}} = \dot Q_{e_1} - \dot Q_{e_3}.
\end{equation}
A similar manipulation shows that $\Sigma_{e_1} = \Sigma_{e_3}$, which allows us to rewrite $L'$ as 
\begin{equation}
    L' = (\Xi_{e_1} + \Xi_{e_3}) \left(\dot Q_{e_1} + \frac{1}{L} (\Phi_{e_1} + \Phi_{e_3}) \right) + \Sigma_{e_1} \left(\dot \Phi_{e_1} + \dot \Phi_{e_3} - \frac{1}{C} Q_{e_1} \right) + \Sigma_{e_1} R (\dot Q_{e_1} + \mathrm i T \Sigma_{e_1})
\end{equation}
which recovers (\ref{eqn:loopnodeex}) with the identification
\begin{subequations}
   \begin{align}
       Q &= Q_{e_1} \\ 
       \Sigma &= \Sigma_{e_1}\\
       \Phi &= \Phi_{e_1} + \Phi_{e_3} \\ 
       \Xi &= \Xi_{e_1} + \Xi_{e_3}.
   \end{align} 
\end{subequations}

Hence, for a given problem, there are two routes to a maximally simplified Lagrangian. The first is to write down $L$ according to the definition (\ref{eqn:msr}) and make variable definitions based on its simplified form after cancelling like terms. The second is to start with (\ref{eqn:magic})
and identify any Noether current constraints that are present and redefine variables after integrating out the constrained degrees of freedom. 
%The first approach requires no creativity while the second is more straightforward logistically. 
We emphasize that if one starts from (\ref{eqn:magic}) as was done in producing (\ref{eqn:spanningtreeex}), one need not integrate out Noether current constraints in order to produce an FPE and any applicable FDTs. 
Noether current constraints are exactly analogous to those that arise in \cite{osborne2023symplectic}, and there is one for each cut consisting only of edges in $\mathcal C \cup \mathcal S$. 

The manipulations of Section \ref{sec:nonlinear} applied to (\ref{eqn:spanningtreeex}) lead to an FPE 
\begin{equation}\label{eqn:redundant}
    \partial_t P =  
    \frac{Q_{e_1}}{C}\frac{\partial P}{\partial \Phi_{e_1}} - \frac{\Phi_{e_1} + \Phi_{e_3}}{L}\frac{\partial P}{\partial Q_{e_1}}
    - \frac{\Phi_{e_1} + \Phi_{e_3}}{L}
     \frac{\partial P}{\partial Q_{e_3}}+ 
    R \left[\frac{\Phi_{e_1} + \Phi_{e_3}}{L}
    \frac{\partial P}{\partial \Phi_{e_3}} + T R \frac{\partial^2P}{\partial \Phi_{e_3}^2}\right].
\end{equation}
This is superficially different from the FPE determined by (\ref{eqn:loopnodeex}) 
\begin{equation}\label{eqn:clean}
    \partial_t P = \frac{Q}{C} \frac{\partial P}{\partial \Phi} - \frac{\Phi}{L} \frac{\partial P}{\partial Q} + R \left[\frac{\Phi}{L} \frac{\partial P}{\partial \Phi} + T R \frac{\partial^2 P}{\partial \Phi^2} \right].
\end{equation}
%The difference between the two may be rectified by demanding that $P$ be a function of only $\Phi = \Phi_{e_1} + \Phi_{e_3}$ and $Q = Q_{e_1}$.  
Nevertheless, (\ref{eqn:redundant}) encodes the same physics as (\ref{eqn:clean}), while it also keeps track of a pair of ``redundant" degrees of freedom. To see that this is the case, change variables so that (\ref{eqn:redundant}) is written in terms of 
\begin{equation}
   \Phi_{\pm} = \Phi_{e_1} \pm \Phi_{e_3}. 
\end{equation}
 (\ref{eqn:redundant}) may be rewritten as 
\begin{align}\label{eqn:cov}
    \partial_t P &=  
    \frac{Q_{e_1}}{C}\left(\frac{\partial P}{\partial \Phi_+} + \frac{\partial P}{\partial \Phi_-}\right) - \frac{\Phi_+}{L}\frac{\partial P}{\partial Q_{e_1}}
    - \frac{\Phi_+}{L}
     \frac{\partial P}{\partial Q_{e_3}} \notag \\
     &\;\;\;\; +
    R \left[\frac{\Phi_+}{L}
    \left(\frac{\partial P}{\partial \Phi_{+}}-\frac{\partial P}{\partial \Phi_-}\right) + T R \left(\frac{\partial^2P}{\partial \Phi_{+}^2} + \frac{\partial^2 P}{\partial \Phi_-^2} - 2 \frac{\partial^2 P}{\partial \Phi_+ \partial \Phi_-}\right)\right].
\end{align}
Given any solution $P$ of (\ref{eqn:cov}), define 
\begin{equation}
    P' = \int \mathrm d \Phi_- \mathrm d Q_{e_3} P.
\end{equation}
Direct calculation shows that
\begin{equation}\label{eq:513}
    \partial_t P' = \frac{Q_{e_1}}{C} \frac{\partial P'}{\partial \Phi_+} - \frac{\Phi_+}{L} \frac{\partial P'}{\partial Q_{e_1}} + R \left[\frac{\Phi_+}{L} \frac{\partial P'}{\partial \Phi_+} + T R \frac{\partial^2 P'}{\partial \Phi_+^2} \right],
\end{equation}
namely $P^\prime$ is a solution of \eqref{eqn:clean}. In subsequent examples, we will not belabor the various routes toward simplification that exist with the understanding that such techniques all lead to the same place.

\subsection{An unclean circuit}\label{sec:unclean}
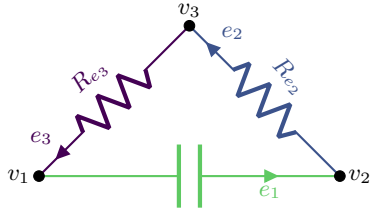
\begin{figure}
    \centering
    \begin{circuitikz}[scale=2]
        \draw[thick,capc] (-1,0) to[capacitor,i_=$e_1$,color=capc] (1,0);
        \draw[thick,pc] (1,0) to[resistor,i_=$e_2$,l_=$R_{e_2}$,color=pc] (0,1);
        \draw[thick,sc] (0,1)  to[resistor,i_=$e_3$,l_=$R_{e_3}$,color=sc] (-1,0);
        \filldraw[black] (-1,0) circle (1pt) node[anchor = east]{$v_1$};
        \filldraw[black] (1,0) circle (1pt) node[anchor = west]{$v_2$};
        \filldraw[black] (0,1) circle (1pt) node[anchor = south]{$v_3$};
    \end{circuitikz}
    \caption{A minimal example of a circuit that is not a clean circuit. Note that every circuit consisting of one dissipative element and one nondissipative element is clean.  Edges are colored according to set inclusion. }
    \label{fig:unclean}
\end{figure}
%\AO{something about these calculations are messed up}
Consider the circuit drawn in Figure \ref{fig:unclean}. We will show that, although the cleanliness of a circuit guarantees that the inversion of $\tilde K$ in (\ref{eqn:tk}) is simple, it is nonetheless possible to perform matrix inversion by hand for a circuit that is not clean. 
For this circuit
\begin{subequations}
    \begin{align}
        \mathcal C = \{e_1\} \\ 
        \mathcal P = \{e_2\} \\ 
        \mathcal S = \{e_3\}
    \end{align}
\end{subequations}
We could have equivalently chosen 
\begin{subequations}
    \begin{align}
        \mathcal S = \{e_2\} \\ 
        \mathcal P = \{e_3\}.
    \end{align}
\end{subequations}
Our convention does not specify which of these choices is preferable, even in the case that the resistors at hand are nonlinear.  
Our spanning tree construction yields 
\begin{equation}
    \tilde{K} = \,\,
     \begin{blockarray}{ccccc}
     e_1  & e_3 & e_1 & e_3 \\ 
    \begin{block}{(cccc)@{\hspace{7pt}}c}
    \frac{1}{R_{e_2}} & \frac{1}{R_{e_2}} & 1 & 0 & e_1  \\ 
    \frac{1}{R_{e_2}} & \frac{1}{R_{e_2}} & 0 & 1 & e_3  \\ 
    -1 & 0 & 0 & 0 & e_1  \\ 
    0 &-1 & 0 & R_{e_3} & e_3  \\ 
    \end{block}
    \end{blockarray} 
\end{equation}
which is invertible. 
Explicitly, 
\begin{equation}
    L = \begin{pmatrix}
        \Xi_{e_1} & \Xi_{e_3} & \Sigma_{e_1} & \Sigma_{e_3} 
    \end{pmatrix}
    \begin{pmatrix}
        \frac{1}{R_{e_2}} & \frac{1}{R_{e_2}} & 1 & 0 \\
        \frac{1}{R_{e_2}} & \frac{1}{R_{e_2}} & 0 & 1 \\
        -1 & 0 & 0 & 0 \\
        0 & -1 & 0 & R_{e_3}
    \end{pmatrix}
    \begin{pmatrix}
        \dot \Phi_{e_1} + \mathrm i T \Xi_{e_1} \\ 
        \dot \Phi_{e_3} + \mathrm i T \Xi_{e_3} \\ 
        \dot Q_{e_1} + \mathrm i T \Sigma_{e_1} \\ 
        \dot Q_{e_3} + \mathrm i T \Sigma_{e_3} \\ 
    \end{pmatrix}
    + \Sigma_{e_1} \frac{Q_{e_1}}{C}.
\end{equation}
Once again 
\begin{equation}
    0 = \frac{\delta S}{\delta \Xi_{e_1}} - \frac{\delta S }{\delta \Xi_{e_3}} = \dot Q_{e_1} - \dot Q_{e_3}, 
\end{equation}
and $\Sigma_{e_1} = \Sigma_{e_2}$  follows analogously. Further defining 
\begin{subequations}
    \begin{align}
        \Phi &= \Phi_{e_1} + \Phi_{e_3} \\ 
        \Xi &= \Xi_{e_1} + \Xi_{e_3} \\ 
        Q &= Q_{e_1} = Q_{e_3} \\ 
        \Sigma &= \Sigma_{e_1} = \Sigma_{e_3}
    \end{align}
\end{subequations}
we write 
\begin{equation}\label{eqn:caninvert}
    L = \begin{pmatrix}
        \Xi & \Sigma
    \end{pmatrix}
    \begin{pmatrix}
        \frac{1}{R_{e_2}} & 1 \\
        -1 & R_{e_3}
    \end{pmatrix}
    \begin{pmatrix}
        \dot\Phi + \mathrm i T \Xi \\
        \dot Q + \mathrm i T \Sigma
    \end{pmatrix}
    + \frac{1}{C} \Sigma Q. 
\end{equation}
Now, the path integral manipulation
\begin{equation}\label{eqn:simpint}
    \int \mathcal D \Phi \mathcal D \Xi \mathcal D \Sigma \text{exp}\left[\mathrm i \int \mathrm d t \left( \frac{1}{R_{e_2}}  \Xi - \Sigma\right)\dot \Phi + \frac{1}{R_{e_2}} \mathrm i T \Xi^2 + \Xi \dot Q\right] = \int \mathcal D \Sigma \text{exp} \left[\mathrm i \int \mathrm d t \,\Sigma R_{e_2} (\dot Q + \mathrm i T \Sigma)\right],
\end{equation}
corresponding to integrating out $\Xi$, allows us to simplify the dissipative Lagrangian to 
\begin{equation}
    L = \Sigma (R_{e_2} + R_{e_3}) (\dot Q + \mathrm i T \Sigma) + \frac{1}{C} \Sigma Q.
\end{equation}
Finally, the resulting FPE is given by 
\begin{equation}\label{eqn:smaller}
    \partial_t P = \frac{\partial }{\partial Q} \frac{1}{R_{e_2} + R_{e_3}} \left[ \frac{Q}{C}P + T \frac{\partial P}{\partial Q} \right].
\end{equation}
Alternatively, we could forgone the simplification afforded by (\ref{eqn:simpint}) and proceeded to write an FPE directly from (\ref{eqn:caninvert}). Simplifications at the Lagrangian level are generally less difficult than simplifications in the FPE directly.

If one elects to add the resistors on edge $e_2$ and $e_3$ together using constitutive relations, then
the resulting simplified circuit is clean. Hence (\ref{eqn:fpe}) may be invoked to yield
\begin{equation}\label{eqn:bigger}
    \partial_t P' = \frac{Q}{C} \frac{\partial P'}{\partial \Phi} + 
    \frac{\partial }{\partial Q} \frac{1}{R_{e_2} + R_{e_3}} \left[ \frac{Q}{C}P' + T \frac{\partial P'}{\partial Q} \right].
\end{equation}
To see that (\ref{eqn:smaller}) and (\ref{eqn:bigger}) are equivalent, we can integrate out $\Phi$, analogously to \eqref{eq:513}.

\subsection{A circuit with a nonlinear resistor}
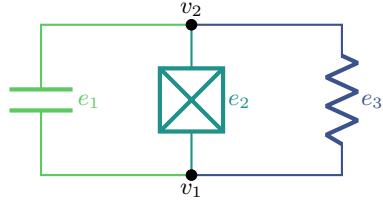
\begin{figure}
    \centering
    \begin{circuitikz}[scale=2]
        \draw[thick, capc] (0,0) to (-1,0) to[capacitor,l_=$e_1$,color=capc] (-1,1) to (0,1);
        \draw[thick,indc] (0,0) to[josephson,color=indc,l_=$e_2$] (0,1);
        \draw[thick,pc] (0,0) to (1,0) to[resistor,color=pc,l_=$e_3$] (1,1) to (0,1);
        \filldraw[black] (0,0) circle (1pt) node[anchor=north]{$v_1$};
        \filldraw[black] (0,1) circle (1pt) node[anchor=south]{$v_2$};
    \end{circuitikz}
    \caption{A circuit consisting of three edges, colored according to set inclusion.  Edge $e_2$ contains a Josephson junction with $E = -E_J \cos \phi_{e_2}$.}
    \label{fig:jnoise}
\end{figure}
Consider the circuit depicted in Figure \ref{fig:jnoise}, which we will refer to as $G$. 
We suppose that $R$ is current fixed so that 
\begin{equation}
    R = R(\dot \phi_{e_3}).
\end{equation}
Since $G$ is a clean circuit, 
it follows that $\dot \phi_{e_3}$ is expressible in terms of some derivative of $E$. 
In terms of spanning tree variables, we have that 
\begin{equation}
    \dot \phi_{e_3} = \dot \Phi = \frac{Q}{C}. 
\end{equation}
Thus, we may write 
\begin{equation}
    R = R\left(\frac{Q}{C}\right).
\end{equation}
We will suppress the dependence of $R$ upon $Q$ henceforth.

One may express the Lagrangian of $G$ in the form (\ref{eqn:obsk}) to produce 
\begin{equation}
    L = \begin{pmatrix}
        \Xi &
        \Sigma 
    \end{pmatrix}
    \begin{pmatrix}
        \frac{1}{R_{e}} & 1 \\ 
        -1 & 0 
    \end{pmatrix}
    \begin{pmatrix}
        \dot \Phi + \mathrm i T \Xi \\ 
        \dot Q + \mathrm i T \Sigma
    \end{pmatrix}
    - \Xi  E_J \sin(\Phi)  + \frac{1}{C} \Sigma Q.
\end{equation}
Invoking (\ref{eqn:fpe}) immediately yields 
\begin{equation}
    \partial_t P = \frac{Q}{C} \frac{\partial P}{\partial \Phi} + E_J \sin(\Phi) \frac{\partial P}{\partial Q} + \frac{\partial }{\partial Q} \left[
     \frac{1}{R}\frac{Q}{C}P - T \frac{\partial R^{-1}}{\partial Q} P + T \frac{\partial }{\partial Q}\left(\frac{1}{R}P \right)
    \right]
\end{equation}
which agrees with \cite{vonoppen}, after a suitable change of variables from $R$ to $\tilde R$ as in Section \ref{sec:nonlinear}.  
Moreover, it is possible to read the variance of fluctuations in the Ito formalism
\begin{equation}
    \langle \xi(t) \xi(t') \rangle = \frac{2 T}{R}\delta (t - t')
\end{equation}
where we emphasize once again that $R$ depends upon $Q$.

\section{Conclusion}
We have revisited the Lagrangian mechanics of dissipative RLC circuits.  This subject has an old history \cite{brayton,bahar_generalized,weissmathis} which has recently been revisited in the context of superconducting circuit quantization \cite{Parra-Rodriguez:2023ykw}.  In this prior literature, the deep relationship between dissipative Lagrangians and Johnson noise is not transparent, and first and foremost the purpose of this paper is to remedy this issue.  Along the way, we hope that our algorithmic methods for building dissipative Lagrangians and removing constrained degrees of freedom could be  practical.

Our work is rather suggestive of a fruitful path forward for systematically modeling superconducting circuits with decoherence as open quantum systems, using the KMS-invariant Schwinger-Keldysh path integrals which represent the natural quantum mechanical analogues \cite{haehl2016fluid,eft1,eft2,jensen2018dissipative,Liulec} of the dissipative Lagrangians discussed in this paper.  However, as pointed out in \cite{Huang:2023eyz}, in the presence of multiplicative noise there appear to be subtleties in path integral regularization.  These subtleties are directly related to the two notions of resistance, $R_e$ and $\tilde R_e$, highlighted in Section \ref{sec:nonlinear}, and we hope that a careful resolution of such subtleties can enable precision experimental tests of the quantum fluctuation-dissipation theorem in a superconducting circuit.  

Given the issues raised above with the path integral formulation of the quantum problem, a desirable alternative could be to construct the Lindbladians which manifestly drive the quantum system to thermal equilibrium.  For nonlinear circuits, unfortunately, the challenge of exactly diagonalizing the Hamiltonian renders it difficult to systematically identify local Lindbladian dynamics that exactly protect a thermal steady state \cite{aqm}. If it is possible to nevertheless identify the specific jump operators that generalize \eqref{eq:FPElast} to an open quantum system, then one may find a powerful new technique for accurately accounting for the effects of decoherence in superconducting circuit, without the need for a microscopic model of the environment. Exciting recent progress \cite{Chen:2023zpu} along these lines suggests that at least \emph{numerically} it may be possible to efficiently model dissipative circuits with a thermal steady state, although a clearer analytical understanding, and physical interpretation, of the appropriate jump operators would be highly desirable.

\section*{Acknowledgements}
 This work was supported by a Research Fellowship from the Alfred P. Sloan Foundation under Grant FG-2020-13795 (AL), and by the Air Force Office of Scientific Research under Grant FA9550-24-1-0120 (AO, AL).

%\begin{appendix}
 %\section{Graph theory results}
%\end{appendix}

\bibliography{thebib}

\begin{thebibliography}{50}%
\makeatletter
\providecommand \@ifxundefined [1]{%
 \@ifx{#1\undefined}
}%
\providecommand \@ifnum [1]{%
 \ifnum #1\expandafter \@firstoftwo
 \else \expandafter \@secondoftwo
 \fi
}%
\providecommand \@ifx [1]{%
 \ifx #1\expandafter \@firstoftwo
 \else \expandafter \@secondoftwo
 \fi
}%
\providecommand \natexlab [1]{#1}%
\providecommand \enquote  [1]{``#1''}%
\providecommand \bibnamefont  [1]{#1}%
\providecommand \bibfnamefont [1]{#1}%
\providecommand \citenamefont [1]{#1}%
\providecommand \href@noop [0]{\@secondoftwo}%
\providecommand \href [0]{\begingroup \@sanitize@url \@href}%
\providecommand \@href[1]{\@@startlink{#1}\@@href}%
\providecommand \@@href[1]{\endgroup#1\@@endlink}%
\providecommand \@sanitize@url [0]{\catcode `\\12\catcode `\$12\catcode
  `\&12\catcode `\#12\catcode `\^12\catcode `\_12\catcode `\%12\relax}%
\providecommand \@@startlink[1]{}%
\providecommand \@@endlink[0]{}%
\providecommand \url  [0]{\begingroup\@sanitize@url \@url }%
\providecommand \@url [1]{\endgroup\@href {#1}{\urlprefix }}%
\providecommand \urlprefix  [0]{URL }%
\providecommand \Eprint [0]{\href }%
\providecommand \doibase [0]{http://dx.doi.org/}%
\providecommand \selectlanguage [0]{\@gobble}%
\providecommand \bibinfo  [0]{\@secondoftwo}%
\providecommand \bibfield  [0]{\@secondoftwo}%
\providecommand \translation [1]{[#1]}%
\providecommand \BibitemOpen [0]{}%
\providecommand \bibitemStop [0]{}%
\providecommand \bibitemNoStop [0]{.\EOS\space}%
\providecommand \EOS [0]{\spacefactor3000\relax}%
\providecommand \BibitemShut  [1]{\csname bibitem#1\endcsname}%
\let\auto@bib@innerbib\@empty
%</preamble>
\bibitem [{\citenamefont {Koch}\ \emph {et~al.}(2007)\citenamefont {Koch},
  \citenamefont {Yu}, \citenamefont {Gambetta}, \citenamefont {Houck},
  \citenamefont {Schuster}, \citenamefont {Majer}, \citenamefont {Blais},
  \citenamefont {Devoret}, \citenamefont {Girvin},\ and\ \citenamefont
  {Schoelkopf}}]{koch2007}%
  \BibitemOpen
  \bibfield  {author} {\bibinfo {author} {\bibfnamefont {Jens}\ \bibnamefont
  {Koch}}, \bibinfo {author} {\bibfnamefont {Terri~M.}\ \bibnamefont {Yu}},
  \bibinfo {author} {\bibfnamefont {Jay}\ \bibnamefont {Gambetta}}, \bibinfo
  {author} {\bibfnamefont {A.~A.}\ \bibnamefont {Houck}}, \bibinfo {author}
  {\bibfnamefont {D.~I.}\ \bibnamefont {Schuster}}, \bibinfo {author}
  {\bibfnamefont {J.}~\bibnamefont {Majer}}, \bibinfo {author} {\bibfnamefont
  {Alexandre}\ \bibnamefont {Blais}}, \bibinfo {author} {\bibfnamefont {M.~H.}\
  \bibnamefont {Devoret}}, \bibinfo {author} {\bibfnamefont {S.~M.}\
  \bibnamefont {Girvin}}, \ and\ \bibinfo {author} {\bibfnamefont {R.~J.}\
  \bibnamefont {Schoelkopf}},\ }\bibfield  {title} {\enquote {\bibinfo {title}
  {Charge-insensitive qubit design derived from the cooper pair box},}\ }\href
  {\doibase 10.1103/PhysRevA.76.042319} {\bibfield  {journal} {\bibinfo
  {journal} {Phys. Rev. A}\ }\textbf {\bibinfo {volume} {76}},\ \bibinfo
  {pages} {042319} (\bibinfo {year} {2007})}\BibitemShut {NoStop}%
\bibitem [{\citenamefont {Manucharyan}\ \emph {et~al.}(2009)\citenamefont
  {Manucharyan}, \citenamefont {Koch}, \citenamefont {Glazman},\ and\
  \citenamefont {Devoret}}]{manucharyan_fluxonium_2009}%
  \BibitemOpen
  \bibfield  {author} {\bibinfo {author} {\bibfnamefont {Vladimir~E.}\
  \bibnamefont {Manucharyan}}, \bibinfo {author} {\bibfnamefont {Jens}\
  \bibnamefont {Koch}}, \bibinfo {author} {\bibfnamefont {Leonid~I.}\
  \bibnamefont {Glazman}}, \ and\ \bibinfo {author} {\bibfnamefont {Michel~H.}\
  \bibnamefont {Devoret}},\ }\bibfield  {title} {\enquote {\bibinfo {title}
  {Fluxonium: {Single} {Cooper}-{Pair} {Circuit} {Free} of {Charge}
  {Offsets}},}\ }\href {\doibase 10.1126/science.1175552} {\bibfield  {journal}
  {\bibinfo  {journal} {Science}\ }\textbf {\bibinfo {volume} {326}},\ \bibinfo
  {pages} {113--116} (\bibinfo {year} {2009})}\BibitemShut {NoStop}%
\bibitem [{\citenamefont {Gyenis}\ \emph {et~al.}(2021)\citenamefont {Gyenis},
  \citenamefont {Di~Paolo}, \citenamefont {Koch}, \citenamefont {Blais},
  \citenamefont {Houck},\ and\ \citenamefont {Schuster}}]{gyenis_2021}%
  \BibitemOpen
  \bibfield  {author} {\bibinfo {author} {\bibfnamefont {Andr\'as}\
  \bibnamefont {Gyenis}}, \bibinfo {author} {\bibfnamefont {Agustin}\
  \bibnamefont {Di~Paolo}}, \bibinfo {author} {\bibfnamefont {Jens}\
  \bibnamefont {Koch}}, \bibinfo {author} {\bibfnamefont {Alexandre}\
  \bibnamefont {Blais}}, \bibinfo {author} {\bibfnamefont {Andrew~A.}\
  \bibnamefont {Houck}}, \ and\ \bibinfo {author} {\bibfnamefont {David~I.}\
  \bibnamefont {Schuster}},\ }\bibfield  {title} {\enquote {\bibinfo {title}
  {Moving beyond the transmon: Noise-protected superconducting quantum
  circuits},}\ }\href {\doibase 10.1103/PRXQuantum.2.030101} {\bibfield
  {journal} {\bibinfo  {journal} {PRX Quantum}\ }\textbf {\bibinfo {volume}
  {2}},\ \bibinfo {pages} {030101} (\bibinfo {year} {2021})}\BibitemShut
  {NoStop}%
\bibitem [{\citenamefont {Krantz}\ \emph {et~al.}(2019)\citenamefont {Krantz},
  \citenamefont {Kjaergaard}, \citenamefont {Yan}, \citenamefont {Orlando},
  \citenamefont {Gustavsson},\ and\ \citenamefont
  {Oliver}}]{krantz_quantum_2019}%
  \BibitemOpen
  \bibfield  {author} {\bibinfo {author} {\bibfnamefont {P.}~\bibnamefont
  {Krantz}}, \bibinfo {author} {\bibfnamefont {M.}~\bibnamefont {Kjaergaard}},
  \bibinfo {author} {\bibfnamefont {F.}~\bibnamefont {Yan}}, \bibinfo {author}
  {\bibfnamefont {T.~P.}\ \bibnamefont {Orlando}}, \bibinfo {author}
  {\bibfnamefont {S.}~\bibnamefont {Gustavsson}}, \ and\ \bibinfo {author}
  {\bibfnamefont {W.~D.}\ \bibnamefont {Oliver}},\ }\bibfield  {title}
  {\enquote {\bibinfo {title} {A quantum engineer's guide to superconducting
  qubits},}\ }\href {\doibase 10.1063/1.5089550} {\bibfield  {journal}
  {\bibinfo  {journal} {Applied Physics Reviews}\ }\textbf {\bibinfo {volume}
  {6}},\ \bibinfo {pages} {021318} (\bibinfo {year} {2019})}\BibitemShut
  {NoStop}%
\bibitem [{\citenamefont {Kjaergaard}\ \emph {et~al.}(2020)\citenamefont
  {Kjaergaard}, \citenamefont {Schwartz}, \citenamefont {Braumüller},
  \citenamefont {Krantz}, \citenamefont {Wang}, \citenamefont {Gustavsson},\
  and\ \citenamefont {Oliver}}]{kjaergaard_superconducting_2020}%
  \BibitemOpen
  \bibfield  {author} {\bibinfo {author} {\bibfnamefont {Morten}\ \bibnamefont
  {Kjaergaard}}, \bibinfo {author} {\bibfnamefont {Mollie~E.}\ \bibnamefont
  {Schwartz}}, \bibinfo {author} {\bibfnamefont {Jochen}\ \bibnamefont
  {Braumüller}}, \bibinfo {author} {\bibfnamefont {Philip}\ \bibnamefont
  {Krantz}}, \bibinfo {author} {\bibfnamefont {Joel I.-J.}\ \bibnamefont
  {Wang}}, \bibinfo {author} {\bibfnamefont {Simon}\ \bibnamefont
  {Gustavsson}}, \ and\ \bibinfo {author} {\bibfnamefont {William~D.}\
  \bibnamefont {Oliver}},\ }\bibfield  {title} {\enquote {\bibinfo {title}
  {Superconducting {Qubits}: {Current} {State} of {Play}},}\ }\href {\doibase
  10.1146/annurev-conmatphys-031119-050605} {\bibfield  {journal} {\bibinfo
  {journal} {Annual Review of Condensed Matter Physics}\ }\textbf {\bibinfo
  {volume} {11}},\ \bibinfo {pages} {369--395} (\bibinfo {year}
  {2020})}\BibitemShut {NoStop}%
\bibitem [{\citenamefont {Blais}\ \emph {et~al.}(2021)\citenamefont {Blais},
  \citenamefont {Grimsmo}, \citenamefont {Girvin},\ and\ \citenamefont
  {Wallraff}}]{blais_circuit_2021}%
  \BibitemOpen
  \bibfield  {author} {\bibinfo {author} {\bibfnamefont {Alexandre}\
  \bibnamefont {Blais}}, \bibinfo {author} {\bibfnamefont {Arne~L.}\
  \bibnamefont {Grimsmo}}, \bibinfo {author} {\bibfnamefont {S.~M.}\
  \bibnamefont {Girvin}}, \ and\ \bibinfo {author} {\bibfnamefont {Andreas}\
  \bibnamefont {Wallraff}},\ }\bibfield  {title} {\enquote {\bibinfo {title}
  {Circuit quantum electrodynamics},}\ }\href {\doibase
  10.1103/RevModPhys.93.025005} {\bibfield  {journal} {\bibinfo  {journal}
  {Rev. Mod. Phys.}\ }\textbf {\bibinfo {volume} {93}},\ \bibinfo {pages}
  {025005} (\bibinfo {year} {2021})}\BibitemShut {NoStop}%
\bibitem [{\citenamefont {Nathan}\ \emph {et~al.}(2024)\citenamefont {Nathan},
  \citenamefont {O'Brien}, \citenamefont {Noh}, \citenamefont {Matheny},
  \citenamefont {Grimsmo}, \citenamefont {Jiang},\ and\ \citenamefont
  {Refael}}]{nathan2024selfcorrecting}%
  \BibitemOpen
  \bibfield  {author} {\bibinfo {author} {\bibfnamefont {Frederik}\
  \bibnamefont {Nathan}}, \bibinfo {author} {\bibfnamefont {Liam}\ \bibnamefont
  {O'Brien}}, \bibinfo {author} {\bibfnamefont {Kyungjoo}\ \bibnamefont {Noh}},
  \bibinfo {author} {\bibfnamefont {Matthew~H.}\ \bibnamefont {Matheny}},
  \bibinfo {author} {\bibfnamefont {Arne~L.}\ \bibnamefont {Grimsmo}}, \bibinfo
  {author} {\bibfnamefont {Liang}\ \bibnamefont {Jiang}}, \ and\ \bibinfo
  {author} {\bibfnamefont {Gil}\ \bibnamefont {Refael}},\ }\href@noop {}
  {\enquote {\bibinfo {title} {Self-correcting gkp qubit and gates in a
  driven-dissipative circuit},}\ } (\bibinfo {year} {2024}),\ \Eprint
  {http://arxiv.org/abs/2405.05671} {arXiv:2405.05671 [cond-mat.mes-hall]}
  \BibitemShut {NoStop}%
\bibitem [{\citenamefont {\emph{et al}}(2024)}]{Miexp}%
  \BibitemOpen
  \bibfield  {author} {\bibinfo {author} {\bibfnamefont {X.~Mi}\ \bibnamefont
  {\emph{et al}}},\ }\bibfield  {title} {\enquote {\bibinfo {title} {Stable
  quantum-correlated many-body states through engineered dissipation},}\ }\href
  {\doibase 10.1126/science.adh9932} {\bibfield  {journal} {\bibinfo  {journal}
  {Science}\ }\textbf {\bibinfo {volume} {383}},\ \bibinfo {pages} {1332--1337}
  (\bibinfo {year} {2024})}\BibitemShut {NoStop}%
\bibitem [{\citenamefont {Caldeira}\ and\ \citenamefont
  {Leggett}(1983)}]{CALDEIRA1983587}%
  \BibitemOpen
  \bibfield  {author} {\bibinfo {author} {\bibfnamefont {A.O.}\ \bibnamefont
  {Caldeira}}\ and\ \bibinfo {author} {\bibfnamefont {A.J.}\ \bibnamefont
  {Leggett}},\ }\bibfield  {title} {\enquote {\bibinfo {title} {Path integral
  approach to quantum brownian motion},}\ }\href {\doibase
  https://doi.org/10.1016/0378-4371(83)90013-4} {\bibfield  {journal} {\bibinfo
   {journal} {Physica A: Statistical Mechanics and its Applications}\ }\textbf
  {\bibinfo {volume} {121}},\ \bibinfo {pages} {587--616} (\bibinfo {year}
  {1983})}\BibitemShut {NoStop}%
\bibitem [{\citenamefont {Cattaneo}\ and\ \citenamefont
  {Paraoanu}(2021)}]{Sorin_2021}%
  \BibitemOpen
  \bibfield  {author} {\bibinfo {author} {\bibfnamefont {Marco}\ \bibnamefont
  {Cattaneo}}\ and\ \bibinfo {author} {\bibfnamefont {Gheorghe~Sorin}\
  \bibnamefont {Paraoanu}},\ }\bibfield  {title} {\enquote {\bibinfo {title}
  {Engineering dissipation with resistive elements in circuit quantum
  electrodynamics},}\ }\href {\doibase https://doi.org/10.1002/qute.202100054}
  {\bibfield  {journal} {\bibinfo  {journal} {Advanced Quantum Technologies}\
  }\textbf {\bibinfo {volume} {4}},\ \bibinfo {pages} {2100054} (\bibinfo
  {year} {2021})},\ \Eprint
  {http://arxiv.org/abs/https://onlinelibrary.wiley.com/doi/pdf/10.1002/qute.202100054}
  {https://onlinelibrary.wiley.com/doi/pdf/10.1002/qute.202100054} \BibitemShut
  {NoStop}%
\bibitem [{\citenamefont {Minev}\ \emph {et~al.}(2021)\citenamefont {Minev},
  \citenamefont {Leghtas}, \citenamefont {Mundhada}, \citenamefont
  {Christakis}, \citenamefont {Pop},\ and\ \citenamefont
  {Devoret}}]{minev2021energyparticipation}%
  \BibitemOpen
  \bibfield  {author} {\bibinfo {author} {\bibfnamefont {Zlatko~K.}\
  \bibnamefont {Minev}}, \bibinfo {author} {\bibfnamefont {Zaki}\ \bibnamefont
  {Leghtas}}, \bibinfo {author} {\bibfnamefont {Shantanu~O.}\ \bibnamefont
  {Mundhada}}, \bibinfo {author} {\bibfnamefont {Lysander}\ \bibnamefont
  {Christakis}}, \bibinfo {author} {\bibfnamefont {Ioan~M.}\ \bibnamefont
  {Pop}}, \ and\ \bibinfo {author} {\bibfnamefont {Michel~H.}\ \bibnamefont
  {Devoret}},\ }\href@noop {} {\enquote {\bibinfo {title} {Energy-participation
  quantization of josephson circuits},}\ } (\bibinfo {year} {2021}),\ \Eprint
  {http://arxiv.org/abs/2010.00620} {arXiv:2010.00620 [quant-ph]} \BibitemShut
  {NoStop}%
\bibitem [{\citenamefont {Vool}\ and\ \citenamefont
  {Devoret}(2017)}]{vool_introduction_2017}%
  \BibitemOpen
  \bibfield  {author} {\bibinfo {author} {\bibfnamefont {Uri}\ \bibnamefont
  {Vool}}\ and\ \bibinfo {author} {\bibfnamefont {Michel}\ \bibnamefont
  {Devoret}},\ }\bibfield  {title} {\enquote {\bibinfo {title} {Introduction to
  quantum electromagnetic circuits},}\ }\href {\doibase 10.1002/cta.2359}
  {\bibfield  {journal} {\bibinfo  {journal} {International Journal of Circuit
  Theory and Applications}\ }\textbf {\bibinfo {volume} {45}},\ \bibinfo
  {pages} {897--934} (\bibinfo {year} {2017})}\BibitemShut {NoStop}%
\bibitem [{\citenamefont {Abdo}\ \emph {et~al.}(2013)\citenamefont {Abdo},
  \citenamefont {Kamal},\ and\ \citenamefont {Devoret}}]{Abdo_2013}%
  \BibitemOpen
  \bibfield  {author} {\bibinfo {author} {\bibfnamefont {Baleegh}\ \bibnamefont
  {Abdo}}, \bibinfo {author} {\bibfnamefont {Archana}\ \bibnamefont {Kamal}}, \
  and\ \bibinfo {author} {\bibfnamefont {Michel}\ \bibnamefont {Devoret}},\
  }\bibfield  {title} {\enquote {\bibinfo {title} {Nondegenerate three-wave
  mixing with the josephson ring modulator},}\ }\href {\doibase
  10.1103/physrevb.87.014508} {\bibfield  {journal} {\bibinfo  {journal}
  {Physical Review B}\ }\textbf {\bibinfo {volume} {87}} (\bibinfo {year}
  {2013}),\ 10.1103/physrevb.87.014508}\BibitemShut {NoStop}%
\bibitem [{\citenamefont {Parra-Rodriguez}\ and\ \citenamefont
  {Egusquiza}(2022)}]{spain22}%
  \BibitemOpen
  \bibfield  {author} {\bibinfo {author} {\bibfnamefont {A.}~\bibnamefont
  {Parra-Rodriguez}}\ and\ \bibinfo {author} {\bibfnamefont {I.~L.}\
  \bibnamefont {Egusquiza}},\ }\bibfield  {title} {\enquote {\bibinfo {title}
  {Quantum fluctuations in electrical multiport linear systems},}\ }\href
  {\doibase 10.1103/PhysRevB.106.054504} {\bibfield  {journal} {\bibinfo
  {journal} {Phys. Rev. B}\ }\textbf {\bibinfo {volume} {106}},\ \bibinfo
  {pages} {054504} (\bibinfo {year} {2022})}\BibitemShut {NoStop}%
\bibitem [{\citenamefont {Osborne}\ \emph {et~al.}(2023)\citenamefont
  {Osborne}, \citenamefont {Larson}, \citenamefont {Jones}, \citenamefont
  {Simmonds}, \citenamefont {Gyenis},\ and\ \citenamefont
  {Lucas}}]{osborne2023symplectic}%
  \BibitemOpen
  \bibfield  {author} {\bibinfo {author} {\bibfnamefont {Andrew}\ \bibnamefont
  {Osborne}}, \bibinfo {author} {\bibfnamefont {Trevyn}\ \bibnamefont
  {Larson}}, \bibinfo {author} {\bibfnamefont {Sarah}\ \bibnamefont {Jones}},
  \bibinfo {author} {\bibfnamefont {Ray~W.}\ \bibnamefont {Simmonds}}, \bibinfo
  {author} {\bibfnamefont {András}\ \bibnamefont {Gyenis}}, \ and\ \bibinfo
  {author} {\bibfnamefont {Andrew}\ \bibnamefont {Lucas}},\ }\href@noop {}
  {\enquote {\bibinfo {title} {Symplectic geometry and circuit quantization},}\
  } (\bibinfo {year} {2023}),\ \Eprint {http://arxiv.org/abs/2304.08531}
  {arXiv:2304.08531 [quant-ph]} \BibitemShut {NoStop}%
\bibitem [{\citenamefont {Parra-Rodriguez}\ and\ \citenamefont
  {Egusquiza}(2023)}]{Parra-Rodriguez:2023ykw}%
  \BibitemOpen
  \bibfield  {author} {\bibinfo {author} {\bibfnamefont {A.}~\bibnamefont
  {Parra-Rodriguez}}\ and\ \bibinfo {author} {\bibfnamefont {I.~L.}\
  \bibnamefont {Egusquiza}},\ }\bibfield  {title} {\enquote {\bibinfo {title}
  {{Geometrical description and Faddeev-Jackiw quantization of electrical
  networks}},}\ }\href@noop {} {\  (\bibinfo {year} {2023})},\ \Eprint
  {http://arxiv.org/abs/2304.12252} {arXiv:2304.12252 [quant-ph]} \BibitemShut
  {NoStop}%
\bibitem [{\citenamefont {Brayton}\ and\ \citenamefont
  {Moser}(1964)}]{brayton}%
  \BibitemOpen
  \bibfield  {author} {\bibinfo {author} {\bibfnamefont {R.~K.}\ \bibnamefont
  {Brayton}}\ and\ \bibinfo {author} {\bibfnamefont {J.~K.}\ \bibnamefont
  {Moser}},\ }\bibfield  {title} {\enquote {\bibinfo {title} {{A theory of
  nonlinear networks. I}},}\ }\href@noop {} {\bibfield  {journal} {\bibinfo
  {journal} {Quart. Appl. Math.}\ }\textbf {\bibinfo {volume} {22}},\ \bibinfo
  {pages} {1} (\bibinfo {year} {1964})}\BibitemShut {NoStop}%
\bibitem [{\citenamefont {Kwatny}\ \emph {et~al.}(1982)\citenamefont {Kwatny},
  \citenamefont {Massimo},\ and\ \citenamefont {Bahar}}]{bahar_generalized}%
  \BibitemOpen
  \bibfield  {author} {\bibinfo {author} {\bibfnamefont {H.}~\bibnamefont
  {Kwatny}}, \bibinfo {author} {\bibfnamefont {F.}~\bibnamefont {Massimo}}, \
  and\ \bibinfo {author} {\bibfnamefont {L.}~\bibnamefont {Bahar}},\ }\bibfield
   {title} {\enquote {\bibinfo {title} {The generalized lagrange formulation
  for nonlinear rlc networks},}\ }\href {\doibase 10.1109/TCS.1982.1085140}
  {\bibfield  {journal} {\bibinfo  {journal} {IEEE Transactions on Circuits and
  Systems}\ }\textbf {\bibinfo {volume} {29}},\ \bibinfo {pages} {220--233}
  (\bibinfo {year} {1982})}\BibitemShut {NoStop}%
\bibitem [{\citenamefont {Weiss}\ and\ \citenamefont
  {Mathis}(1997)}]{weissmathis}%
  \BibitemOpen
  \bibfield  {author} {\bibinfo {author} {\bibfnamefont {L.}~\bibnamefont
  {Weiss}}\ and\ \bibinfo {author} {\bibfnamefont {W.}~\bibnamefont {Mathis}},\
  }\bibfield  {title} {\enquote {\bibinfo {title} {{A Hamiltonian formulation
  for complete nonlinear RLC-networks}},}\ }\href@noop {} {\bibfield  {journal}
  {\bibinfo  {journal} {IEEE Trans. Circ. Appl. I: Fund. Theor. Appl.}\
  }\textbf {\bibinfo {volume} {44}},\ \bibinfo {pages} {843} (\bibinfo {year}
  {1997})}\BibitemShut {NoStop}%
\bibitem [{\citenamefont {You}\ \emph {et~al.}(2019)\citenamefont {You},
  \citenamefont {Sauls},\ and\ \citenamefont {Koch}}]{jens_timedep}%
  \BibitemOpen
  \bibfield  {author} {\bibinfo {author} {\bibfnamefont {Xinyuan}\ \bibnamefont
  {You}}, \bibinfo {author} {\bibfnamefont {J.~A.}\ \bibnamefont {Sauls}}, \
  and\ \bibinfo {author} {\bibfnamefont {Jens}\ \bibnamefont {Koch}},\
  }\bibfield  {title} {\enquote {\bibinfo {title} {Circuit quantization in the
  presence of time-dependent external flux},}\ }\href {\doibase
  10.1103/PhysRevB.99.174512} {\bibfield  {journal} {\bibinfo  {journal} {Phys.
  Rev. B}\ }\textbf {\bibinfo {volume} {99}},\ \bibinfo {pages} {174512}
  (\bibinfo {year} {2019})}\BibitemShut {NoStop}%
\bibitem [{\citenamefont {Ciani}\ \emph {et~al.}(2023)\citenamefont {Ciani},
  \citenamefont {DiVincenzo},\ and\ \citenamefont {Terhal}}]{Ciani:2023ubt}%
  \BibitemOpen
  \bibfield  {author} {\bibinfo {author} {\bibfnamefont {Alessandro}\
  \bibnamefont {Ciani}}, \bibinfo {author} {\bibfnamefont {David~P.}\
  \bibnamefont {DiVincenzo}}, \ and\ \bibinfo {author} {\bibfnamefont
  {Barbara~M.}\ \bibnamefont {Terhal}},\ }\bibfield  {title} {\enquote
  {\bibinfo {title} {{Lecture Notes on Quantum Electrical Circuits}},}\
  }\href@noop {} {\  (\bibinfo {year} {2023})},\ \Eprint
  {http://arxiv.org/abs/2312.05329} {arXiv:2312.05329 [quant-ph]} \BibitemShut
  {NoStop}%
\bibitem [{\citenamefont {Nigg}\ \emph {et~al.}(2012)\citenamefont {Nigg},
  \citenamefont {Paik}, \citenamefont {Vlastakis}, \citenamefont {Kirchmair},
  \citenamefont {Shankar}, \citenamefont {Frunzio}, \citenamefont {Devoret},
  \citenamefont {Schoelkopf},\ and\ \citenamefont
  {Girvin}}]{nigg_black-box_2012}%
  \BibitemOpen
  \bibfield  {author} {\bibinfo {author} {\bibfnamefont {Simon~E.}\
  \bibnamefont {Nigg}}, \bibinfo {author} {\bibfnamefont {Hanhee}\ \bibnamefont
  {Paik}}, \bibinfo {author} {\bibfnamefont {Brian}\ \bibnamefont {Vlastakis}},
  \bibinfo {author} {\bibfnamefont {Gerhard}\ \bibnamefont {Kirchmair}},
  \bibinfo {author} {\bibfnamefont {S.}~\bibnamefont {Shankar}}, \bibinfo
  {author} {\bibfnamefont {Luigi}\ \bibnamefont {Frunzio}}, \bibinfo {author}
  {\bibfnamefont {M.~H.}\ \bibnamefont {Devoret}}, \bibinfo {author}
  {\bibfnamefont {R.~J.}\ \bibnamefont {Schoelkopf}}, \ and\ \bibinfo {author}
  {\bibfnamefont {S.~M.}\ \bibnamefont {Girvin}},\ }\bibfield  {title}
  {\enquote {\bibinfo {title} {Black-{Box} {Superconducting} {Circuit}
  {Quantization}},}\ }\href {\doibase 10.1103/PhysRevLett.108.240502}
  {\bibfield  {journal} {\bibinfo  {journal} {Physical Review Letters}\
  }\textbf {\bibinfo {volume} {108}},\ \bibinfo {pages} {240502} (\bibinfo
  {year} {2012})}\BibitemShut {NoStop}%
\bibitem [{\citenamefont {Ulrich}\ and\ \citenamefont
  {Hassler}(2016)}]{ulrich_dual_2016}%
  \BibitemOpen
  \bibfield  {author} {\bibinfo {author} {\bibfnamefont {Jascha}\ \bibnamefont
  {Ulrich}}\ and\ \bibinfo {author} {\bibfnamefont {Fabian}\ \bibnamefont
  {Hassler}},\ }\bibfield  {title} {\enquote {\bibinfo {title} {Dual approach
  to circuit quantization using loop charges},}\ }\href {\doibase
  10.1103/PhysRevB.94.094505} {\bibfield  {journal} {\bibinfo  {journal}
  {Physical Review B}\ }\textbf {\bibinfo {volume} {94}},\ \bibinfo {pages}
  {094505} (\bibinfo {year} {2016})}\BibitemShut {NoStop}%
\bibitem [{\citenamefont {Riwar}\ and\ \citenamefont
  {DiVincenzo}(2022)}]{riwar_circuit_2022}%
  \BibitemOpen
  \bibfield  {author} {\bibinfo {author} {\bibfnamefont {R.-P.}\ \bibnamefont
  {Riwar}}\ and\ \bibinfo {author} {\bibfnamefont {D.~P.}\ \bibnamefont
  {DiVincenzo}},\ }\bibfield  {title} {\enquote {\bibinfo {title} {Circuit
  quantization with time-dependent magnetic fields for realistic geometries},}\
  }\href {https://www.nature.com/articles/s41534-022-00539-x} {\bibfield
  {journal} {\bibinfo  {journal} {npj Quantum Information}\ }\textbf {\bibinfo
  {volume} {8}} (\bibinfo {year} {2022})}\BibitemShut {NoStop}%
\bibitem [{\citenamefont {Thanh~Le}\ \emph {et~al.}(2020)\citenamefont
  {Thanh~Le}, \citenamefont {Cole},\ and\ \citenamefont
  {Stace}}]{thanh_le_building_2020}%
  \BibitemOpen
  \bibfield  {author} {\bibinfo {author} {\bibfnamefont {Dat}\ \bibnamefont
  {Thanh~Le}}, \bibinfo {author} {\bibfnamefont {Jared~H.}\ \bibnamefont
  {Cole}}, \ and\ \bibinfo {author} {\bibfnamefont {T.~M.}\ \bibnamefont
  {Stace}},\ }\bibfield  {title} {\enquote {\bibinfo {title} {Building a bigger
  {Hilbert} space for superconducting devices, one {Bloch} state at a time},}\
  }\href {\doibase 10.1103/PhysRevResearch.2.013245} {\bibfield  {journal}
  {\bibinfo  {journal} {Physical Review Research}\ }\textbf {\bibinfo {volume}
  {2}},\ \bibinfo {pages} {013245} (\bibinfo {year} {2020})}\BibitemShut
  {NoStop}%
\bibitem [{\citenamefont {Devoret}\ and\ \citenamefont
  {Schoelkopf}(2013)}]{devoret_fqi}%
  \BibitemOpen
  \bibfield  {author} {\bibinfo {author} {\bibfnamefont {M.~H.}\ \bibnamefont
  {Devoret}}\ and\ \bibinfo {author} {\bibfnamefont {R.~J.}\ \bibnamefont
  {Schoelkopf}},\ }\bibfield  {title} {\enquote {\bibinfo {title}
  {Superconducting circuits for quantum information: An outlook},}\ }\href
  {\doibase 10.1126/science.1231930} {\bibfield  {journal} {\bibinfo  {journal}
  {Science}\ }\textbf {\bibinfo {volume} {339}},\ \bibinfo {pages} {1169--1174}
  (\bibinfo {year} {2013})},\ \Eprint
  {http://arxiv.org/abs/https://www.science.org/doi/pdf/10.1126/science.1231930}
  {https://www.science.org/doi/pdf/10.1126/science.1231930} \BibitemShut
  {NoStop}%
\bibitem [{\citenamefont {Bravyi}\ \emph {et~al.}(2022)\citenamefont {Bravyi},
  \citenamefont {Dial}, \citenamefont {Gambetta}, \citenamefont {Gil},\ and\
  \citenamefont {Nazario}}]{brayvi_the_future}%
  \BibitemOpen
  \bibfield  {author} {\bibinfo {author} {\bibfnamefont {Sergey}\ \bibnamefont
  {Bravyi}}, \bibinfo {author} {\bibfnamefont {Oliver}\ \bibnamefont {Dial}},
  \bibinfo {author} {\bibfnamefont {Jay~M.}\ \bibnamefont {Gambetta}}, \bibinfo
  {author} {\bibfnamefont {Darío}\ \bibnamefont {Gil}}, \ and\ \bibinfo
  {author} {\bibfnamefont {Zaira}\ \bibnamefont {Nazario}},\ }\bibfield
  {title} {\enquote {\bibinfo {title} {The future of quantum computing with
  superconducting qubits},}\ }\href {\doibase 10.1063/5.0082975} {\bibfield
  {journal} {\bibinfo  {journal} {Journal of Applied Physics}\ }\textbf
  {\bibinfo {volume} {132}},\ \bibinfo {pages} {160902} (\bibinfo {year}
  {2022})},\ \Eprint {http://arxiv.org/abs/https://doi.org/10.1063/5.0082975}
  {https://doi.org/10.1063/5.0082975} \BibitemShut {NoStop}%
\bibitem [{\citenamefont {Parra-Rodriguez}\ and\ \citenamefont
  {Egusquiza}(2024)}]{Parra-Rodriguez:2024vtx}%
  \BibitemOpen
  \bibfield  {author} {\bibinfo {author} {\bibfnamefont {A.}~\bibnamefont
  {Parra-Rodriguez}}\ and\ \bibinfo {author} {\bibfnamefont {I.~L.}\
  \bibnamefont {Egusquiza}},\ }\bibfield  {title} {\enquote {\bibinfo {title}
  {{Faddeev-Jackiw quantisation of nonreciprocal quasi-lumped electrical
  networks}},}\ }\href@noop {} {\  (\bibinfo {year} {2024})},\ \Eprint
  {http://arxiv.org/abs/2401.09120} {arXiv:2401.09120 [quant-ph]} \BibitemShut
  {NoStop}%
\bibitem [{\citenamefont {Mariantoni}\ and\ \citenamefont
  {Gorgichuk}(2024)}]{mariantoni2024quantum}%
  \BibitemOpen
  \bibfield  {author} {\bibinfo {author} {\bibfnamefont {Matteo}\ \bibnamefont
  {Mariantoni}}\ and\ \bibinfo {author} {\bibfnamefont {Noah}\ \bibnamefont
  {Gorgichuk}},\ }\href@noop {} {\enquote {\bibinfo {title} {Quantum
  synchronization in nonconservative electrical circuits with
  kirchhoff-heisenberg equations},}\ } (\bibinfo {year} {2024}),\ \Eprint
  {http://arxiv.org/abs/2403.10474} {arXiv:2403.10474 [quant-ph]} \BibitemShut
  {NoStop}%
\bibitem [{\citenamefont {Haehl}\ \emph {et~al.}(2016)\citenamefont {Haehl},
  \citenamefont {Loganayagam},\ and\ \citenamefont
  {Rangamani}}]{haehl2016fluid}%
  \BibitemOpen
  \bibfield  {author} {\bibinfo {author} {\bibfnamefont {Felix~M.}\
  \bibnamefont {Haehl}}, \bibinfo {author} {\bibfnamefont {R.}~\bibnamefont
  {Loganayagam}}, \ and\ \bibinfo {author} {\bibfnamefont {Mukund}\
  \bibnamefont {Rangamani}},\ }\bibfield  {title} {\enquote {\bibinfo {title}
  {The fluid manifesto: Emergent symmetries, hydrodynamics, and black holes},}\
  }\href {\doibase 10.1007/JHEP01(2016)184} {\bibfield  {journal} {\bibinfo
  {journal} {JHEP}\ }\textbf {\bibinfo {volume} {01}},\ \bibinfo {pages} {184}
  (\bibinfo {year} {2016})},\ \Eprint {http://arxiv.org/abs/1510.02494}
  {arXiv:1510.02494 [hep-th]} \BibitemShut {NoStop}%
\bibitem [{\citenamefont {Crossley}\ \emph {et~al.}(2017)\citenamefont
  {Crossley}, \citenamefont {Glorioso},\ and\ \citenamefont {Liu}}]{eft1}%
  \BibitemOpen
  \bibfield  {author} {\bibinfo {author} {\bibfnamefont {Michael}\ \bibnamefont
  {Crossley}}, \bibinfo {author} {\bibfnamefont {Paolo}\ \bibnamefont
  {Glorioso}}, \ and\ \bibinfo {author} {\bibfnamefont {Hong}\ \bibnamefont
  {Liu}},\ }\bibfield  {title} {\enquote {\bibinfo {title} {Effective field
  theory of dissipative fluids},}\ }\href {\doibase 10.1007/jhep09(2017)095}
  {\bibfield  {journal} {\bibinfo  {journal} {JHEP}\ }\textbf {\bibinfo
  {volume} {2017}} (\bibinfo {year} {2017}),\
  10.1007/jhep09(2017)095}\BibitemShut {NoStop}%
\bibitem [{\citenamefont {Glorioso}\ \emph {et~al.}(2017)\citenamefont
  {Glorioso}, \citenamefont {Crossley},\ and\ \citenamefont {Liu}}]{eft2}%
  \BibitemOpen
  \bibfield  {author} {\bibinfo {author} {\bibfnamefont {Paolo}\ \bibnamefont
  {Glorioso}}, \bibinfo {author} {\bibfnamefont {Michael}\ \bibnamefont
  {Crossley}}, \ and\ \bibinfo {author} {\bibfnamefont {Hong}\ \bibnamefont
  {Liu}},\ }\bibfield  {title} {\enquote {\bibinfo {title} {Effective field
  theory of dissipative fluids {(II)}: Classical limit, dynamical {KMS}
  symmetry, and entropy current},}\ }\href {\doibase 10.1007/jhep09(2017)096}
  {\bibfield  {journal} {\bibinfo  {journal} {JHEP}\ }\textbf {\bibinfo
  {volume} {2017}} (\bibinfo {year} {2017}),\
  10.1007/jhep09(2017)096}\BibitemShut {NoStop}%
\bibitem [{\citenamefont {Jensen}\ \emph {et~al.}(2018)\citenamefont {Jensen},
  \citenamefont {Pinzani-Fokeeva},\ and\ \citenamefont
  {Yarom}}]{jensen2018dissipative}%
  \BibitemOpen
  \bibfield  {author} {\bibinfo {author} {\bibfnamefont {Kristan}\ \bibnamefont
  {Jensen}}, \bibinfo {author} {\bibfnamefont {Natalia}\ \bibnamefont
  {Pinzani-Fokeeva}}, \ and\ \bibinfo {author} {\bibfnamefont {Amos}\
  \bibnamefont {Yarom}},\ }\bibfield  {title} {\enquote {\bibinfo {title}
  {Dissipative hydrodynamics in superspace},}\ }\href {\doibase
  10.1007/JHEP09(2018)127} {\bibfield  {journal} {\bibinfo  {journal} {JHEP}\
  }\textbf {\bibinfo {volume} {09}},\ \bibinfo {pages} {127} (\bibinfo {year}
  {2018})},\ \Eprint {http://arxiv.org/abs/1701.07436} {arXiv:1701.07436
  [hep-th]} \BibitemShut {NoStop}%
\bibitem [{\citenamefont {Glorioso}\ and\ \citenamefont {Liu}(2018)}]{Liulec}%
  \BibitemOpen
  \bibfield  {author} {\bibinfo {author} {\bibfnamefont {Paolo}\ \bibnamefont
  {Glorioso}}\ and\ \bibinfo {author} {\bibfnamefont {Hong}\ \bibnamefont
  {Liu}},\ }\href@noop {} {\enquote {\bibinfo {title} {Lectures on
  nonequilibrium effective field theories and fluctuating hydrodynamics},}\ }
  (\bibinfo {year} {2018}),\ \Eprint {http://arxiv.org/abs/1805.09331}
  {arXiv:1805.09331 [hep-th]} \BibitemShut {NoStop}%
\bibitem [{\citenamefont {Martin}\ \emph {et~al.}(1973)\citenamefont {Martin},
  \citenamefont {Siggia},\ and\ \citenamefont {Rose}}]{MSR}%
  \BibitemOpen
  \bibfield  {author} {\bibinfo {author} {\bibfnamefont {P.~C.}\ \bibnamefont
  {Martin}}, \bibinfo {author} {\bibfnamefont {E.~D.}\ \bibnamefont {Siggia}},
  \ and\ \bibinfo {author} {\bibfnamefont {H.~A.}\ \bibnamefont {Rose}},\
  }\bibfield  {title} {\enquote {\bibinfo {title} {Statistical dynamics of
  classical systems},}\ }\href {\doibase 10.1103/PhysRevA.8.423} {\bibfield
  {journal} {\bibinfo  {journal} {Phys. Rev. A}\ }\textbf {\bibinfo {volume}
  {8}},\ \bibinfo {pages} {423--437} (\bibinfo {year} {1973})}\BibitemShut
  {NoStop}%
\bibitem [{\citenamefont {Johnson}(1928)}]{JohnJohnson}%
  \BibitemOpen
  \bibfield  {author} {\bibinfo {author} {\bibfnamefont {J.~B.}\ \bibnamefont
  {Johnson}},\ }\bibfield  {title} {\enquote {\bibinfo {title} {Thermal
  agitation of electricity in conductors},}\ }\href {\doibase
  10.1103/PhysRev.32.97} {\bibfield  {journal} {\bibinfo  {journal} {Phys.
  Rev.}\ }\textbf {\bibinfo {volume} {32}},\ \bibinfo {pages} {97--109}
  (\bibinfo {year} {1928})}\BibitemShut {NoStop}%
\bibitem [{\citenamefont {Nyquist}(1928)}]{Nyquist}%
  \BibitemOpen
  \bibfield  {author} {\bibinfo {author} {\bibfnamefont {H.}~\bibnamefont
  {Nyquist}},\ }\bibfield  {title} {\enquote {\bibinfo {title} {Thermal
  agitation of electric charge in conductors},}\ }\href {\doibase
  10.1103/PhysRev.32.110} {\bibfield  {journal} {\bibinfo  {journal} {Phys.
  Rev.}\ }\textbf {\bibinfo {volume} {32}},\ \bibinfo {pages} {110--113}
  (\bibinfo {year} {1928})}\BibitemShut {NoStop}%
\bibitem [{\citenamefont {Twiss}(1955)}]{Twiss}%
  \BibitemOpen
  \bibfield  {author} {\bibinfo {author} {\bibfnamefont {R.~Q.}\ \bibnamefont
  {Twiss}},\ }\bibfield  {title} {\enquote {\bibinfo {title} {{Nyquist's and
  Thevenin's Theorems Generalized for Nonreciprocal Linear Networks}},}\ }\href
  {\doibase 10.1063/1.1722048} {\bibfield  {journal} {\bibinfo  {journal}
  {Journal of Applied Physics}\ }\textbf {\bibinfo {volume} {26}},\ \bibinfo
  {pages} {599--602} (\bibinfo {year} {1955})},\ \Eprint
  {http://arxiv.org/abs/https://pubs.aip.org/aip/jap/article-pdf/26/5/599/18314646/599\_1\_online.pdf}
  {https://pubs.aip.org/aip/jap/article-pdf/26/5/599/18314646/599\_1\_online.pdf}
  \BibitemShut {NoStop}%
\bibitem [{\citenamefont {Kautz}\ and\ \citenamefont
  {Martinis}(1990)}]{martinis}%
  \BibitemOpen
  \bibfield  {author} {\bibinfo {author} {\bibfnamefont {R.~L.}\ \bibnamefont
  {Kautz}}\ and\ \bibinfo {author} {\bibfnamefont {John~M.}\ \bibnamefont
  {Martinis}},\ }\bibfield  {title} {\enquote {\bibinfo {title} {Noise-affected
  i-v curves in small hysteretic josephson junctions},}\ }\href {\doibase
  10.1103/PhysRevB.42.9903} {\bibfield  {journal} {\bibinfo  {journal} {Phys.
  Rev. B}\ }\textbf {\bibinfo {volume} {42}},\ \bibinfo {pages} {9903--9937}
  (\bibinfo {year} {1990})}\BibitemShut {NoStop}%
\bibitem [{\citenamefont {Steiner}\ \emph {et~al.}(2023)\citenamefont
  {Steiner}, \citenamefont {Melischek}, \citenamefont {Trahms}, \citenamefont
  {Franke},\ and\ \citenamefont {von Oppen}}]{vonoppen}%
  \BibitemOpen
  \bibfield  {author} {\bibinfo {author} {\bibfnamefont {Jacob~F.}\
  \bibnamefont {Steiner}}, \bibinfo {author} {\bibfnamefont {Larissa}\
  \bibnamefont {Melischek}}, \bibinfo {author} {\bibfnamefont {Martina}\
  \bibnamefont {Trahms}}, \bibinfo {author} {\bibfnamefont {Katharina~J.}\
  \bibnamefont {Franke}}, \ and\ \bibinfo {author} {\bibfnamefont {Felix}\
  \bibnamefont {von Oppen}},\ }\bibfield  {title} {\enquote {\bibinfo {title}
  {Diode effects in current-biased josephson junctions},}\ }\href {\doibase
  10.1103/PhysRevLett.130.177002} {\bibfield  {journal} {\bibinfo  {journal}
  {Phys. Rev. Lett.}\ }\textbf {\bibinfo {volume} {130}},\ \bibinfo {pages}
  {177002} (\bibinfo {year} {2023})}\BibitemShut {NoStop}%
\bibitem [{\citenamefont {Huang}\ \emph {et~al.}(2023)\citenamefont {Huang},
  \citenamefont {Farrell}, \citenamefont {Friedman}, \citenamefont {Zane},
  \citenamefont {Glorioso},\ and\ \citenamefont {Lucas}}]{Huang:2023eyz}%
  \BibitemOpen
  \bibfield  {author} {\bibinfo {author} {\bibfnamefont {Xiaoyang}\
  \bibnamefont {Huang}}, \bibinfo {author} {\bibfnamefont {Jack~H.}\
  \bibnamefont {Farrell}}, \bibinfo {author} {\bibfnamefont {Aaron~J.}\
  \bibnamefont {Friedman}}, \bibinfo {author} {\bibfnamefont {Isabella}\
  \bibnamefont {Zane}}, \bibinfo {author} {\bibfnamefont {Paolo}\ \bibnamefont
  {Glorioso}}, \ and\ \bibinfo {author} {\bibfnamefont {Andrew}\ \bibnamefont
  {Lucas}},\ }\bibfield  {title} {\enquote {\bibinfo {title} {{Generalized
  time-reversal symmetry and effective theories for nonequilibrium matter}},}\
  }\href@noop {} {\  (\bibinfo {year} {2023})},\ \Eprint
  {http://arxiv.org/abs/2310.12233} {arXiv:2310.12233 [cond-mat.stat-mech]}
  \BibitemShut {NoStop}%
\bibitem [{\citenamefont {Gardiner}(2009)}]{gardiner}%
  \BibitemOpen
  \bibfield  {author} {\bibinfo {author} {\bibfnamefont {C.}~\bibnamefont
  {Gardiner}},\ }\href@noop {} {\emph {\bibinfo {title} {Stochastic Methods}}}\
  (\bibinfo  {publisher} {Springer},\ \bibinfo {year} {2009})\BibitemShut
  {NoStop}%
\bibitem [{\citenamefont {Osborne}\ and\ \citenamefont
  {Lucas}(2024)}]{fluxcharge}%
  \BibitemOpen
  \bibfield  {author} {\bibinfo {author} {\bibfnamefont {Andrew}\ \bibnamefont
  {Osborne}}\ and\ \bibinfo {author} {\bibfnamefont {Andrew}\ \bibnamefont
  {Lucas}},\ }\bibfield  {title} {\enquote {\bibinfo {title} {Flux-charge
  symmetric theory of superconducting circuits},}\ }\href {\doibase
  10.1103/PhysRevB.109.174524} {\bibfield  {journal} {\bibinfo  {journal}
  {Phys. Rev. B}\ }\textbf {\bibinfo {volume} {109}},\ \bibinfo {pages}
  {174524} (\bibinfo {year} {2024})}\BibitemShut {NoStop}%
\bibitem [{\citenamefont {Rymarz}\ and\ \citenamefont
  {DiVincenzo}(2022)}]{rymarz_consistent_2022}%
  \BibitemOpen
  \bibfield  {author} {\bibinfo {author} {\bibfnamefont {Martin}\ \bibnamefont
  {Rymarz}}\ and\ \bibinfo {author} {\bibfnamefont {David~P.}\ \bibnamefont
  {DiVincenzo}},\ }\href {http://arxiv.org/abs/2208.11767} {\enquote {\bibinfo
  {title} {Consistent {Quantization} of {Nearly} {Singular} {Superconducting}
  {Circuits}},}\ } (\bibinfo {year} {2022}),\ \bibinfo {note} {arXiv:2208.11767
  [quant-ph]}\BibitemShut {NoStop}%
\bibitem [{\citenamefont {Chitta}\ \emph {et~al.}(2022)\citenamefont {Chitta},
  \citenamefont {Zhao}, \citenamefont {Huang}, \citenamefont {Mondragon-Shem},\
  and\ \citenamefont {Koch}}]{chitta_computer-aided_2022}%
  \BibitemOpen
  \bibfield  {author} {\bibinfo {author} {\bibfnamefont {Sai~Pavan}\
  \bibnamefont {Chitta}}, \bibinfo {author} {\bibfnamefont {Tianpu}\
  \bibnamefont {Zhao}}, \bibinfo {author} {\bibfnamefont {Ziwen}\ \bibnamefont
  {Huang}}, \bibinfo {author} {\bibfnamefont {Ian}\ \bibnamefont
  {Mondragon-Shem}}, \ and\ \bibinfo {author} {\bibfnamefont {Jens}\
  \bibnamefont {Koch}},\ }\href {\doibase 10.48550/arXiv.2206.08320} {\enquote
  {\bibinfo {title} {Computer-aided quantization and numerical analysis of
  superconducting circuits},}\ } (\bibinfo {year} {2022}),\ \bibinfo {note}
  {arXiv:2206.08320 [quant-ph]}\BibitemShut {NoStop}%
\bibitem [{\citenamefont {Levin}\ and\ \citenamefont {Peres}(2017)}]{MCMT}%
  \BibitemOpen
  \bibfield  {author} {\bibinfo {author} {\bibfnamefont {David~A.}\
  \bibnamefont {Levin}}\ and\ \bibinfo {author} {\bibfnamefont {Yuval}\
  \bibnamefont {Peres}},\ }\href@noop {} {\emph {\bibinfo {title} {{Markov
  Chains and Mixing Times: Second Edition}}}}\ (\bibinfo  {publisher} {American
  Mathematical Society},\ \bibinfo {year} {2017})\BibitemShut {NoStop}%
\bibitem [{\citenamefont {Strutt}(1871)}]{Strutt}%
  \BibitemOpen
  \bibfield  {author} {\bibinfo {author} {\bibfnamefont {J.~W.}\ \bibnamefont
  {Strutt}},\ }\bibfield  {title} {\enquote {\bibinfo {title} {Some general
  theorems relating to vibrations},}\ }\href {\doibase
  https://doi.org/10.1112/plms/s1-4.1.357} {\bibfield  {journal} {\bibinfo
  {journal} {Proceedings of the London Mathematical Society}\ }\textbf
  {\bibinfo {volume} {s1-4}},\ \bibinfo {pages} {357--368} (\bibinfo {year}
  {1871})},\ \Eprint
  {http://arxiv.org/abs/https://londmathsoc.onlinelibrary.wiley.com/doi/pdf/10.1112/plms/s1-4.1.357}
  {https://londmathsoc.onlinelibrary.wiley.com/doi/pdf/10.1112/plms/s1-4.1.357}
  \BibitemShut {NoStop}%
\bibitem [{\citenamefont {Mariantoni}(2021)}]{mariantoni2021energy}%
  \BibitemOpen
  \bibfield  {author} {\bibinfo {author} {\bibfnamefont {M.}~\bibnamefont
  {Mariantoni}},\ }\href@noop {} {\enquote {\bibinfo {title} {The energy of an
  arbitrary electrical circuit, classical and quantum},}\ } (\bibinfo {year}
  {2021}),\ \Eprint {http://arxiv.org/abs/2007.08519} {arXiv:2007.08519
  [physics.class-ph]} \BibitemShut {NoStop}%
\bibitem [{\citenamefont {Guo}\ \emph {et~al.}(2024)\citenamefont {Guo},
  \citenamefont {Hart}, \citenamefont {Chen}, \citenamefont {Friedman},\ and\
  \citenamefont {Lucas}}]{aqm}%
  \BibitemOpen
  \bibfield  {author} {\bibinfo {author} {\bibfnamefont {Jinkang}\ \bibnamefont
  {Guo}}, \bibinfo {author} {\bibfnamefont {Oliver}\ \bibnamefont {Hart}},
  \bibinfo {author} {\bibfnamefont {Chi-Fang}\ \bibnamefont {Chen}}, \bibinfo
  {author} {\bibfnamefont {Aaron~J.}\ \bibnamefont {Friedman}}, \ and\ \bibinfo
  {author} {\bibfnamefont {Andrew}\ \bibnamefont {Lucas}},\ }\bibfield  {title}
  {\enquote {\bibinfo {title} {{Designing open quantum systems with known
  steady states: Davies generators and beyond}},}\ }\href@noop {} {\  (\bibinfo
  {year} {2024})},\ \Eprint {http://arxiv.org/abs/2404.14538} {arXiv:2404.14538
  [quant-ph]} \BibitemShut {NoStop}%
\bibitem [{\citenamefont {Chen}\ \emph {et~al.}(2023)\citenamefont {Chen},
  \citenamefont {Kastoryano},\ and\ \citenamefont {Gily\'en}}]{Chen:2023zpu}%
  \BibitemOpen
  \bibfield  {author} {\bibinfo {author} {\bibfnamefont {Chi-Fang}\
  \bibnamefont {Chen}}, \bibinfo {author} {\bibfnamefont {Michael~J.}\
  \bibnamefont {Kastoryano}}, \ and\ \bibinfo {author} {\bibfnamefont
  {Andr\'as}\ \bibnamefont {Gily\'en}},\ }\bibfield  {title} {\enquote
  {\bibinfo {title} {{An efficient and exact noncommutative quantum Gibbs
  sampler}},}\ }\href@noop {} {\  (\bibinfo {year} {2023})},\ \Eprint
  {http://arxiv.org/abs/2311.09207} {arXiv:2311.09207 [quant-ph]} \BibitemShut
  {NoStop}%
\end{thebibliography}%
\end{document}